\documentclass[12pt]{article}
\usepackage{fullpage}
\usepackage{times}

\usepackage{doi}

\usepackage{amsthm}

\theoremstyle{plain}
\newtheorem{theorem}{Theorem}
\newtheorem{lemma}[theorem]{Lemma}

\newtheorem{corollary}[theorem]{Corollary}

\theoremstyle{definition}
\newtheorem{definition}{Definition}
\theoremstyle{remark}
\newtheorem{remark}{Remark}
\newtheorem{example}{Example}%
\usepackage{graphicx}
\usepackage{mathptmx}      %

\usepackage{amsmath}
\RequirePackage{stmaryrd} %
\usepackage{paralist}
\usepackage{xxcolor}
\definecolor{vred}{rgb}{.7,0,0}
\definecolor{vblue}{rgb}{.1,.15,.62}

\usepackage{hyperref}

\usepackage{tikz}
\usetikzlibrary{matrix}

\tikzstyle{transition}=[-stealth]
\tikzstyle{similar state}=[double]
\tikzstyle{transition exists}=[dashed]
\tikzstyle{equivalence}=[<->,double]

\def\qedhere{}
\def\smartqed{}
\usepackage[keepproof]{endproof2}

\usepackage[bbsets,Dfprime]{math}
\usepackage[modernsign,substopindex,shortmquant,mquantifiertype,mconnectiveformal,bracketinterpret,prefixflatinterpret,bracketmodalinterpret,fixformat,setfixinterpret,modifopindex,seqarrow,seqoptional,sidenotecalculus,abbrseqcontext,shortterms,nosigmaterms,novarterms]{logic}
\usepackage[pretest,nocommandblocks]{progreg}
\usepackage[bracketinterpret,prefixflatinterpret,bracketmodalinterpret,fixformat]{dL}

\definecolor{vvblue}{rgb}{.14,.21,.868}%
\renewcommand{\axkey}[1]{\textcolor{vvblue}{#1}}

\renewcommand{\dL}[1][]
{\text{\upshape\textsf{d{\kern-0.05em}L}}\xspace}

\usepackage{prettyref}
\newcommand{\rref}[2][]{\prettyref{#2}}
\newrefformat{sec}{Section\,\ref{#1}}
\newrefformat{app}{Appendix\,\ref{#1}}
\newrefformat{def}{Def.\,\ref{#1}}
\newrefformat{thm}{Theorem\,\ref{#1}}
\newrefformat{prop}{Proposition\,\ref{#1}}
\newrefformat{lem}{Lemma\,\ref{#1}}
\newrefformat{cor}{Corollary\,\ref{#1}}
\newrefformat{ex}{Example\,\ref{#1}}
\newrefformat{tab}{Table\,\ref{#1}}
\newrefformat{fig}{Fig.\,\ref{#1}}
\newrefformat{rem}{Remark\,\ref{#1}}
\newrefformat{case}{Case\,\ref{#1}}

\renewcommand{\linferPremissSeparation}{~~}%

\newcommand{\bebecomes}{\mathrel{::=}}
\newcommand{\alternative}{~|~}
\providecommand{\dfn}[2][]{\emph{#2}}

\newcommand*{\genDE}[1]{\theta}%
\newcommand{\ivr}{\psi}
\newcommand{\inv}{\varphi}

\renewcommand{\der}[1]{(#1)'}%

\newcommand{\admissible}{\text{admissible}\xspace}

\newcommand{\solf}{y}%

\newcommand*{\linterpretationsconst}[2]{\mathcal{I}}
\renewcommand{\linterpretationsname}{\mathcal{S}}

\newcommand{\I}{\vdLint[const=I,state=\nu]}
\newcommand{\It}{\vdLint[const=I,state=\omega]}
\newcommand{\If}{\DALint[const=I,flow=\varphi]}
\newcommand*{\Iff}[1][\zeta]{\vdLint[const=I,state=\varphi(#1)]}%
\newcommand{\Idot}{\vdLint[const=I,state=]}

\newcommand{\Ia}{\iadjointSubst{\sigma}{\I}}%
\newcommand{\Ita}{\iadjointSubst{\sigma}{\It}}%
\newcommand{\Iminner}{\imodif[const]{\I}{\,\usarg}{d}}%

\newcommand{\Ialt}{\vdLint[const=J,state=\tilde{\nu}]}%
\newcommand{\Italt}{\vdLint[const=J,state=\tilde{\omega}]}%

\newcommand{\Imid}{\vdLint[const=I,state=\nu']}%
\newcommand{\Itmid}{\vdLint[const=I,state=\omega']}%
\newcommand{\Imidstep}{\vdLint[const=I,state=\nu'_z]}%
\newcommand{\Itmidstep}{\vdLint[const=I,state=\omega'_z]}%

\newcommand*{\freevarsdef}[1]{\mathop{\mathrel{\textsf{\upshape F}}\joinrel\mathrel{\textsf{\upshape V}}}(#1)}%
\newcommand*{\boundvarsdef}[1]{\mathop{\mathrel{\textsf{\upshape B}}\joinrel\mathrel{\textsf{\upshape V}}}(#1)}%

  \makeatletter
  \renewcommand{\iadjointSubst}[2]{%
    \useinterpretation{#2}%
    \edef\tmpadjointconst{{#1}^*_{\Interpretation@state}{\Interpretation@const}}%
    \iconcat[const=\tmpadjointconst]{#2}
  }
  \makeatother

\usepackage{fancyhdr}
\usepackage{ifpdf}
\ifpdf
\fi

\begin{document}
\title{A Complete Uniform Substitution Calculus\\for Differential Dynamic Logic}
\author{Andr\'e Platzer\thanks{%
  Computer Science Department, Carnegie Mellon University, Pittsburgh, USA
  {aplatzer@cs.cmu.edu}
  \newline
  An extended abstract has appeared at CADE 2015 \cite{DBLP:conf/cade/Platzer15} with its proofs listed in \cite{DBLP:journals/corr/Platzer15:usubst}. 
}
}
\date{}

\maketitle
\allowdisplaybreaks
\thispagestyle{empty}

\begin{abstract}
This article introduces a relatively complete proof calculus for \emph{differential dynamic logic} (\dL) that is entirely based on \emph{uniform substitution}, a proof rule that substitutes a formula for a predicate symbol everywhere.
Uniform substitutions make it possible to use \emph{axioms} instead of axiom schemata, thereby substantially simplifying implementations.
Instead of subtle schema variables and soundness-critical side conditions on the occurrence patterns of logical variables to restrict infinitely many axiom schema instances to sound ones, the resulting calculus adopts only a finite number of ordinary \dL formulas as axioms, which uniform substitutions instantiate soundly.
The static semantics of differential dynamic logic and the soundness-critical restrictions it imposes on proof steps is captured exclusively in uniform substitutions and variable renamings as opposed to being spread in delicate ways across the prover implementation.
In addition to sound uniform substitutions, this article introduces \emph{differential forms} for differential dynamic logic that make it possible to internalize differential invariants, differential substitutions, and derivatives as first-class axioms to reason about differential equations axiomatically.
The resulting axiomatization of differential dynamic logic is proved to be sound and relatively complete.
\\[\medskipamount]
\end{abstract}

\section{Introduction}
\allowdisplaybreaks

\emph{Differential dynamic logic} (\dL) \cite{DBLP:journals/jar/Platzer08,DBLP:conf/lics/Platzer12b} is a logic for proving correctness properties of hybrid systems.
It has a sound and complete proof calculus relative to differential equations \cite{DBLP:journals/jar/Platzer08,DBLP:conf/lics/Platzer12b} and a sound and complete proof calculus relative to discrete systems \cite{DBLP:conf/lics/Platzer12b}.
Both sequent calculi \cite{DBLP:journals/jar/Platzer08} and Hilbert-type axiomatizations \cite{DBLP:conf/lics/Platzer12b} have been presented for \dL but only the former have been implemented.
The implementation of \dL's sequent calculus in \KeYmaera \cite{DBLP:conf/cade/PlatzerQ08} makes it straightforward for users to prove properties of hybrid systems, because it provides proof rules that perform natural decompositions for each operator.
The downside is that the implementation of the rule schemata and their different and subtle side conditions on occurrence constraints and relations of reading and writing of variables as well as rule applications in a formula context is quite nontrivial and inflexible in \KeYmaera.

The goal of this article is to identify how to, instead, make it straightforward to implement the proof calculus of differential dynamic logic in a parsimonious way by writing down a finite list of \emph{axioms} (concrete formulas, not axiom schemata that represent an infinite list of axioms subject to sophisticated soundness-critical schema variable matching and side condition checking implementations).
The resulting calculus features more modular axioms that can be combined with one another to regain the effect of a single \dL sequent proof rule.
The axioms are implemented in the object language without meta constructs, which enables a substantially simpler prover core.

As a mechanism for instantiating axioms, this article follows observations for differential game logic \cite{DBLP:journals/tocl/Platzer15} highlighting the significance and elegance of \emph{uniform substitution}, a classical proof rule for first-order logic \cite[\S35,40]{Church_1956}.
Uniform substitutions uniformly instantiate predicate and function symbols with formulas and terms, respectively, as functions of their arguments.
In the presence of the nontrivial binding structure that nondeterminism and differential equations of hybrid systems induce for the dynamic modalities of differential dynamic logic, flexible but sound uniform substitutions become more complex, but can still be read off directly from the static semantics.
The static semantics of \dL directly determines uniform substitutions (and variable renamings), which, in turn, are the only elements of the prover core that need to know anything about the language and its static semantics.
A proof may simply start from a \dL formula that is an axiom.

This approach is dual to other successful ways of solving the intricacies and subtleties of substitutions \cite{DBLP:journals/jsl/Church40,DBLP:journals/jsl/Henkin53} by imposing occurrence side conditions on axiom schemata and proof rules, which is what uniform substitutions get rid of.
The uniform substitution framework shares many goals with other logical frameworks \cite{DBLP:books/el/RV01/Pfenning01}, including leading to smaller soundness-critical cores, more flexibility when augmenting reasoning techniques, and reducing the gap between a logic and its theorem prover.
Logical frameworks shine when renaming and substitution of the object language are in line with those of the meta-language.
Uniform substitutions provide a simpler approach for languages with the intricate binding of imperative and especially hybrid system dynamics in which, e.g., the same occurrence of a variable can be both free and bound.

Side conditions for axiom schemata can be nontrivial.
Classical \dL calculi \cite{DBLP:journals/jar/Platzer08,DBLP:conf/lics/Platzer12b} have an axiom schema expressing that a formula $\phi$ holds always after following a differential equation \(\pevolve{\D{x}=\genDE{x}}\) (as expressed by \dL formula \(\dbox{\hevolve{\D{x}=\genDE{x}}}{\phi}\)) iff $\phi$ holds for all times $t\geq0$ after the discrete assignment \({\pupdate{\pumod{x}{\solf(t)}}}\) assigning the solution $\solf(t)$ to $x$:
\[
      \cinferenceRule[evolveb|$\dibox{'}$]{evolve}
      {\linferenceRule[equiv]
        {\lforall{t{\geq}0}{\dbox{\pupdate{\pumod{x}{\solf(t)}}}{\phi}}}
        {\dbox{\hevolve{\D{x}=\genDE{x}}}{\phi}}
        \quad
      }{\m{\D{\solf}(t)=\genDE{\solf}}}%
\]
Soundness-critical side conditions need to ensure that $t$ is a sufficiently fresh variable and that $\solf(t)$ indeed solves the differential equation and obeys the symbolic initial value condition $\solf(0)=x$.
Uniform substitutions obviate the need for such side conditions.
An axiom is simply a single object-level formula as opposed to an algorithm accepting infinitely many formulas under certain side conditions.
A proof rule is simply a pair of object-level formulas as opposed to an algorithm transforming  formulas.
Derived axioms, derived rules, rule application mechanisms, lemmas, definitions, and parametric invariant search are all definable from uniform substitutions.
\emph{Differential forms} are added to \dL in this article for the purpose of internalizing differential invariants \cite{DBLP:journals/logcom/Platzer10}, differential cuts \cite{DBLP:journals/logcom/Platzer10,DBLP:journals/lmcs/Platzer12}, differential ghosts \cite{DBLP:journals/lmcs/Platzer12}, differential substitutions, total differentials and Lie-derivatives \cite{DBLP:journals/logcom/Platzer10,DBLP:journals/lmcs/Platzer12} as separate first-class axioms in \dL.

This article presents a highly modular and straightforward approach.
It introduces \emph{differential-form} \dL and its dynamic semantics,
and proves its static semantics sound for the dynamic semantics (\rref{sec:dL}).
It then defines uniform substitutions that are proved sound from the static semantics (\rref{sec:usubst}).
Uniform substitutions enable a parsimoniously implementable axiomatization (\rref{sec:dL-axioms}) with concrete \dL formulas (or pairs for rules), which are proved sound individually from the dynamic semantics (\rref{sec:differential}) without having to worry how they might be instantiated.
Differential forms are used to obtain axioms for proving properties of differential equations (\rref{sec:differential}).
This modular approach with separate soundness proofs is to be contrasted with previous complex built-in algorithms that mix multiple axioms into special-purpose rules \cite{DBLP:journals/logcom/Platzer10,DBLP:journals/lmcs/Platzer12}.
Finally, the logic is proved to be sound and relatively complete (\rref{sec:differential}).
Proofs are provided in \rref{app:proofs}.
Overall, uniform substitutions significantly simplify prover core implementations, because uniform substitutions are straightforward and reduce implementing axioms and axiomatic rules to copy\&paste.

\newsavebox{\Rval}%
\sbox{\Rval}{$\scriptstyle\mathbb{R}$}
\irlabel{qear|\usebox{\Rval}}

\newsavebox{\USarg}%
\sbox{\USarg}{$\boldsymbol{\cdot}$}

\newsavebox{\UScarg}%
\sbox{\UScarg}{$\boldsymbol{\_}$}

\newsavebox{\Lightningval}%
\sbox{\Lightningval}{$\scriptstyle\textcolor{vred}{\lightning}$}
\irlabel{clash|clash\usebox{\Lightningval}} %
\irlabel{unsound|\usebox{\Lightningval}}  %

\irlabel{invind|ind}

\section{Differential-Form Differential Dynamic Logic} \label{sec:dL}

This section presents \emph{differential-form differential dynamic logic}, which adds differential forms to differential dynamic logic \cite{DBLP:journals/jar/Platzer08,DBLP:conf/lics/Platzer12b} in order to axiomatically internalize reasoning about differential equations and differentials as first-class citizens.
Because the logic itself did not otherwise change, the relationship to related work from previous presentations of differential dynamic logic \cite{DBLP:journals/jar/Platzer08,DBLP:conf/lics/Platzer12b} continues to apply.
The primary purpose of the uniform substitution approach is to lead to a significantly simpler implementation, which could benefit other approaches \cite{DavorenNerode_2000,DBLP:conf/aplas/LiuLQZZZZ10,DBLP:journals/iac/CimattiRT15}, too.

\subsection{Syntax}
This section defines the syntax of the language of (differential-form) differential dynamic logic \dL and its hybrid programs.
The syntax first defines terms.
The set of all \emph{variables} is $\allvars$.
Variables of the form $\D{x}$ for a variable $x\in\allvars$ are called \emph{differential symbols}.
Differential symbol $\D{x}$ is just an independent variable associated to variable $x$. 
For any subset $V\subseteq\allvars$ is \(\D{V}\mdefeq\{\D{x} : x\in V\}\) the set of \emph{differential symbols} $\D{x}$ for the variables in $V$.
The set of all variables is assumed to already contain all its differential symbols $\D{\allvars}\subseteq\allvars$.
So $x\in\allvars$ implies $\D{x},\D[2]{x}\in\allvars$ etc.\ even if $\D[2]{x}$ is not used here.
Function symbols are written $f,g,h$, predicate symbols $p,q,r$, and variables $x,y,z\in\allvars$ with corresponding differential symbols $\D{x},\D{y},\D{z}\in\D{\allvars}$.
Program constants are $a,b,c$.

\begin{definition}[Terms]
\emph{Terms} are defined by this grammar
(with $\theta,\eta,\theta_1,\dots,\theta_k$ as terms, $x\in\allvars$ as variable,
and $f$ as function symbol):
\[
  \theta,\eta ~\bebecomes~
  x
  \alternative
  f(\theta_1,\dots,\theta_k)
  \alternative
  \theta+\eta
  \alternative
  \theta\cdot\eta
  \alternative \der{\theta}
\]
\end{definition}
Number literals such as 0,1 are allowed as function symbols without arguments that are interpreted as the numbers they denote.
Occasionally, constructions will be simplified by considering $\theta+\eta$ and $\theta\cdot\eta$ as special cases of function symbols $f(\theta,\eta)$, but $+$ and $\cdot$ always denote addition and multiplication.
\emph{Differential-form \dL} allows \emph{differentials} \(\der{\theta}\) of terms $\theta$ as terms for the purpose of axiomatically internalizing reasoning about differential equations.
The differential \(\der{\theta}\) describes how the value of $\theta$ changes locally depending on how the values of its variables $x$ change, i.e.\ as a function of the values of the corresponding differential symbols $\D{x}$.
Differentials will make it possible to reduce reasoning about \emph{differential equations} to reasoning about \emph{equations of differentials}, which, quite unlike differential equations, have a local semantics in isolated states and are, thus, amenable to an axiomatic treatment.

Formulas and hybrid programs (\HPs) of \dL are defined by simultaneous induction, because formulas can occur in programs and programs can occur in formulas.
Similar simultaneous inductions are, thus, used throughout the proofs in this article.
\begin{definition}[\dL formula]
The \emph{formulas of (differential-form) differential dynamic logic} ({\dL}) are defined by the grammar
(with \dL formulas $\phi,\psi$, terms $\theta,\eta,\theta_1,\dots,\theta_k$, predicate symbol $p$, quantifier symbol $C$, variable $x$, \HP $\alpha$):
  \[
  \phi,\psi ~\bebecomes~
  \theta\geq\eta \alternative
  p(\theta_1,\dots,\theta_k) \alternative
  \contextapp{C}{\phi} \alternative
  \lnot \phi \alternative
  \phi \land \psi \alternative
  \lforall{x}{\phi} \alternative 
  \lexists{x}{\phi} \alternative
  \dbox{\alpha}{\phi}
  \alternative \ddiamond{\alpha}{\phi}
  \]
\end{definition}
Operators $>,\leq,<,\lor,\limply,\lbisubjunct$ are definable, e.g., \(\phi\limply\psi\) as \(\lnot(\phi\land\lnot\psi)\).
Also \(\dbox{\alpha}{\phi}\) is equivalent to \(\lnot\ddiamond{\alpha}{\lnot\phi}\) and \(\lforall{x}{\phi}\) equivalent to \(\lnot\lexists{x}{\lnot\phi}\).
The modal formula \(\dbox{\alpha}{\phi}\) expresses that $\phi$ holds after all runs of $\alpha$, while the dual \(\ddiamond{\alpha}{\phi}\) expresses that $\phi$ holds after some run of $\alpha$.
\emph{Quantifier symbols} $C$ (with formula $\phi$ as argument), i.e.\ higher-order predicate symbols that bind all variables of $\phi$, are unnecessary but included for convenience since they internalize contextual congruence reasoning efficiently with uniform substitutions.
The concrete quantifier chain in \(\lforall{x}{\lexists{y}{\phi}}\) evaluates the formula $\phi$ at multiple $x$ and $y$ values to determine whether the whole formula is true.
Similarly, an abstract quantifier symbol $C$ can evaluate its formula argument $\phi$ for different variable values to determine whether $\contextapp{C}{\phi}$ is true.
Whether $\contextapp{C}{\phi}$ is true, and where exactly $C$ evaluates its argument $\phi$ to find out, depends on the interpretation of $C$.

\begin{definition}[Hybrid program]
\emph{Hybrid programs} (\HPs) are defined by the following grammar (with $\alpha,\beta$ as \HPs, program constant $a$, variable $x$, term $\theta$ possibly containing $x$,
and with \dL formula\footnote{%
Quantifier-free formulas of first-order logic of real arithmetic are enough for most purposes.
} $\ivr$):
\[
  \alpha,\beta ~\bebecomes~
  a\alternative
  \pupdate{\pumod{x}{\theta}}
  \alternative
  \ptest{\ivr}
  \alternative
  \pevolvein{\D{x}=\genDE{x}}{\ivr}
  \alternative
  \alpha\cup\beta
  \alternative
  \alpha;\beta
  \alternative
  \prepeat{\alpha}
\]
\end{definition}
\emph{Assignments} \m{\pupdate{\pumod{x}{\theta}}} of $\theta$ to variable $x$, \emph{tests} \m{\ptest{\ivr}} of the formula $\ivr$ in the current state, \emph{differential equations} \(\pevolvein{\D{x}=\genDE{x}}{\ivr}\) restricted to the evolution domain $\ivr$, \emph{nondeterministic choices} \(\pchoice{\alpha}{\beta}\), \emph{sequential compositions} \(\alpha;\beta\), and \emph{nondeterministic repetition} \(\prepeat{\alpha}\) are as usual in \dL \cite{DBLP:journals/jar/Platzer08,DBLP:conf/lics/Platzer12b}.
The assignment \m{\pupdate{\pumod{x}{\theta}}} instantaneously changes the value of $x$ to that of $\theta$.
The test $\ptest{\ivr}$ checks whether $\ivr$ is true in the current state and discards the program execution otherwise.
The continuous evolution \(\pevolvein{\D{x}=\genDE{x}}{\ivr}\) will follow the differential equation \(\D{x}=\genDE{x}\) for any nondeterministic amount of time, but cannot leave the region where the evolution domain constraint $\ivr$ holds.
For example, \(\pevolvein{\D{x}=v\syssep\D{v}=a}{v\geq0}\) follows the differential equation where position $x$ changes with time-derivative $v$ while the velocity $v$ changes with time-derivative $a$ for any arbitrary amount of time, but without ever allowing a negative velocity $v$ (which would, otherwise, ultimately happen for negative accelerations $a<0$).
Usually, the value of differential symbol $\D{x}$ is unrelated to the value of variable $x$.
But along a differential equation \(\pevolve{\D{x}=\genDE{x}}\), differential symbol $\D{x}$ has the value of the time-derivative of the value of $x$ (and is, furthermore, equal to $\genDE{x}$).
Differential equations \(\pevolvein{\D{x}=\genDE{x}}{\ivr}\) have to be in explicit form, so $\D{y}$ and $\der{\eta}$ cannot occur in $\genDE{x}$ and $x\not\in\D{\allvars}$.
The nondeterministic choice \(\pchoice{\alpha}{\beta}\) either executes subprogram $\alpha$ or $\beta$, nondeterministically.
The sequential composition \(\alpha;\beta\) first executes $\alpha$ and then, upon completion of $\alpha$, runs $\beta$.
The nondeterministic repetition $\prepeat{\alpha}$ repeats $\alpha$ any number of times, nondeterministically.

The effect of an assignment \m{\Dupdate{\Dumod{\D{x}}{\theta}}} to differential symbol $\D{x}\in\allvars$, called \emph{differential assignment}, is like the effect of an assignment \m{\pupdate{\pumod{x}{\theta}}} to variable $x$, except that it changes the value of the differential symbol $\D{x}$ instead of the value of $x$.
It is not to be confused with the differential equation \(\pevolve{\D{x}=\genDE{x}}\), which will follow said differential equation continuously for an arbitrary amount of time.
The differential assignment \m{\Dupdate{\Dumod{\D{x}}{\theta}}}, instead, only assigns the value of $\theta$ to the differential symbol $\D{x}$ discretely once at an instant of time.
Program constants $a$ are uninterpreted, i.e.\ their behavior depends on the interpretation in the same way that the values of function symbols $f$, predicate symbols $p$, and quantifier symbols $C$ depend on their interpretation.

\begin{example}[Simple car]
The \dL formula
\begin{equation}
v\geq2 \land b>0 \limply \dbox{\prepeat{((\pchoice{\pupdate{\pumod{a}{-b}}}{\pupdate{\pumod{a}{5}}});~ \pevolvein{\D{x}=v\syssep\D{v}=a}{v\geq0})}}{\,v\geq0}
\label{eq:no-backwards}
\end{equation}
expresses that a car starting with velocity $v\geq2$ and braking constant $b>0$ will always have nonnegative velocity $v\geq0$ when following a \HP that repeatedly provides a nondeterministic control choice between putting the acceleration $a$ to braking (\m{\pupdate{\pumod{a}{-b}}}) or to a positive constant (\m{\pupdate{\pumod{a}{5}}}) 
before following the differential equation system
\(\pevolve{\D{x}=v\syssep\D{v}=a}\) restricted to the evolution domain constraint \(v\geq0\) for any amount of time.
The formula in \rref{eq:no-backwards} is true, because the car never moves backward in the \HP.
But similar questions quickly become challenging, e.g., about safe distances to other cars or for models with more detailed physical dynamics.
\end{example}

\subsection{Dynamic Semantics}

The (denotational) dynamic semantics of \dL defines, depending on the values of the symbols, what real value terms evaluate to, what truth-value formulas have, and from what initial states which final states are reachable by running its \HPs.
Since the values of variables and differential symbols can change over time, they receive their value by the state.
A \emph{state} is a mapping \(\iget[state]{\I}:\allvars\to\reals\) from variables $\allvars$ including differential symbols $\D{\allvars}\subseteq\allvars$ to $\reals$.
The set of states is denoted \(\linterpretations{\Sigma}{V}\).
The set \(\scomplement{X} = \linterpretations{\Sigma}{V}\setminus X\) is the complement of a set $X\subseteq\linterpretations{\Sigma}{V}$.
Let
\m{\iget[state]{\imodif[state]{\I}{x}{r}}} denote the state that agrees with state~$\iget[state]{\I}$ except for the value of variable~\m{x}, which is changed to~\m{r \in \reals}.
The interpretation of a function symbol $f$ of arity $n$ (i.e.\ with $n$ arguments) in \emph{interpretation} $\iget[const]{\I}$ is a (smooth, i.e.\ with derivatives of any order) function \(\iget[const]{\I}(f):\reals^n\to\reals\) of $n$ arguments (continuously differentiable suffices).
The set of interpretations is denoted $\linterpretationsconst{\Sigma}{V}$.
The semantics of a term $\theta$ is a mapping \(\ivaluation{}{\theta} : \linterpretationsconst{\Sigma}{V} \to (\linterpretations{\Sigma}{V} \to \reals)\) from interpretations and states to a real number.

\begin{definition}[Semantics of terms] \label{def:dL-valuationTerm}
The \emph{semantics of a term} $\theta$ in interpretation $\iget[const]{\I}$ and state $\iget[state]{\I}\in\linterpretations{\Sigma}{V}$ is its value \m{\ivaluation{\I}{\theta}} in $\reals$.
It is defined inductively as follows
\begin{compactenum}
\item \m{\ivaluation{\I}{x} = \iget[state]{\I}(x)} for variable $x\in\allvars$
\item \(\ivaluation{\I}{f(\theta_1,\dots,\theta_k)} = \iget[const]{\I}(f)\big(\ivaluation{\I}{\theta_1},\dots,\ivaluation{\I}{\theta_k}\big)\) for function symbol $f$
\item \m{\ivaluation{\I}{\theta+\eta} = \ivaluation{\I}{\theta} + \ivaluation{\I}{\eta}}
\item \m{\ivaluation{\I}{\theta\cdot\eta} = \ivaluation{\I}{\theta} \cdot \ivaluation{\I}{\eta}}
\item
\m{
\ivaluation{\I}{\der{\theta}}
\displaystyle
=
\sum_{x\in\allvars} \iget[state]{\I}(\D{x}) \Dp[x]{\ivaluation{\Idot}{\theta}}(\iget[state]{\I})
= \sum_{x\in\allvars} \iget[state]{\I}(\D{x}) \Dp[x]{\ivaluation{\I}{\theta}}
}
\end{compactenum}
\end{definition}

\noindent
Time-derivatives are undefined in an isolated state $\iget[state]{\I}$.
The clou is that differentials can, nevertheless, be given a local semantics in a single state:
\m{\ivaluation{\I}{\der{\theta}}} is the sum of all (analytic) spatial partial derivatives at $\iget[state]{\I}$ of the value of $\theta$ by all variables $x$ multiplied by the corresponding direction described by the value $\iget[state]{\I}(\D{x})$ of differential symbol $\D{x}$.
That sum over all variables $x\in\allvars$ is finite, because $\theta$ only mentions finitely many variables $x$ and the partial derivative by variables $x$ that do not occur in $\theta$ is 0.
As usual, $\Dp[x]{g}(\iget[state]{\I})$ is the partial derivative of function $g$ at point $\iget[state]{\I}$ by variable $x$, which is sometimes also just denoted $\Dp[x]{g(\iget[state]{\I})}$.
Hence, the partial derivative \(\Dp[x]{\ivaluation{\Idot}{\theta}}(\iget[state]{\I}) = \Dp[x]{\ivaluation{\I}{\theta}}\)
is the derivative of the one-dimensional function
\({\def\Im{\imodif[state]{\I}{x}{X}}%
\ivaluation{\Im}{\theta}}\) of $X$ at $\iget[state]{\I}(x)$.
The spatial partial derivatives exist since $\ivaluation{\I}{\theta}$ is a composition of smooth functions, so is itself smooth.
Thus, the semantics of \m{\ivaluation{\I}{\der{\theta}}} is the \emph{differential}\footnote{%
The usual point-free abuse of notation aligns \rref{def:dL-valuationTerm} with its mathematical counterparts by rewriting the differential as 
\(\ivalues{\I}{\der{\theta}} = d(\ivalues{\I}{\theta}) = \sum_{i=1}^n \Dp[x^i]{\ivalues{\I}{\theta}} dx^i\)
when $x^1,\dots,x^n$ are the variables in $\theta$ and their differentials \(dx^i\) form the basis of the cotangent space, which, when evaluated at a point $\iget[state]{\I}$ whose values \(\iget[state]{\I}(\D{x})\) determine the actual tangent vector alias vector field.
}
of (the value of) $\theta$, hence a differential one-form giving a real value for each tangent vector (i.e.\ point of a vector field) described by the values \(\iget[state]{\I}(\D{x})\).
The values \m{\iget[state]{\I}(\D{x})} of the differential symbols $\D{x}$ select the direction in which $x$ changes, locally.
The partial derivatives of \(\ivaluation{\I}{\theta}\) by $x$ describe how the value of $\theta$ changes with a change of $x$.
Along the solution of (the vector field corresponding to) a differential equation, the value of differential \m{\der{\theta}} coincides with the analytic time-derivative of $\theta$ (\rref{lem:differentialLemma}).

The semantics of a formula $\phi$ is a mapping \(\ivaluation{}{\phi} : \linterpretationsconst{\Sigma}{V} \to \powerset{\linterpretations{\Sigma}{V}}\) 
from interpretations to the set of all states in which $\phi$ is true, where \(\powerset{\linterpretations{\Sigma}{V}}\) is the powerset of \(\linterpretations{\Sigma}{V}\).
The semantics of an \HP $\alpha$ is a mapping \(\ivaluation{}{\alpha} : \linterpretationsconst{\Sigma}{V} \to \powerset{\linterpretations{\Sigma}{V}\times\linterpretations{\Sigma}{V}}\) from interpretations to a reachability relation on states.
The set of states $\imodel{\I}{\phi} \subseteq \linterpretations{\Sigma}{V}$ in which formula $\phi$ is true
and the relation \m{\iaccess[\alpha]{\I}\subseteq\linterpretations{\Sigma}{V}\times\linterpretations{\Sigma}{V}} of \HP $\alpha$ are defined by simultaneous induction as their syntax is simultaneously inductive.
The interpretation of predicate symbol $p$ with arity $n$ is an $n$-ary relation \(\iget[const]{\I}(p)\subseteq\reals^n\).
The interpretation of quantifier symbol $C$ is a functional \(\iget[const]{\I}(C) : \powerset{\linterpretations{\Sigma}{V}} \to \powerset{\linterpretations{\Sigma}{V}}\) mapping sets \(M\subseteq\linterpretations{\Sigma}{V}\) of states where its argument is true to sets \(\iget[const]{\I}(C)(M)\subseteq\linterpretations{\Sigma}{V}\) of states where $C$ applied to that argument is then true.

\begin{definition}[\dL semantics] \label{def:dL-valuation}
The \emph{semantics of a \dL formula} $\phi$, for each interpretation $\iget[const]{\I}$ with a corresponding set of states $\linterpretations{\Sigma}{V}$, is the subset \m{\imodel{\I}{\phi}\subseteq\linterpretations{\Sigma}{V}} of states in which $\phi$ is true.
It is defined inductively as follows
\begin{compactenum}
\item \(\imodel{\I}{\theta\geq\eta} = \{\iget[state]{\I} \in \linterpretations{\Sigma}{V} \with \ivaluation{\I}{\theta}\geq\ivaluation{\I}{\eta}\}\)
\item \(\imodel{\I}{p(\theta_1,\dots,\theta_k)} = \{\iget[state]{\I} \in \linterpretations{\Sigma}{V} \with (\ivaluation{\I}{\theta_1},\dots,\ivaluation{\I}{\theta_k})\in\iget[const]{\I}(p)\}\)
\item \(\imodel{\I}{\contextapp{C}{\phi}} = \iget[const]{\I}(C)\big(\imodel{\I}{\phi}\big)\) for quantifier symbol $C$
\item \(\imodel{\I}{\lnot\phi} = \scomplement{(\imodel{\I}{\phi})}
= \linterpretations{\Sigma}{V}\setminus(\imodel{\I}{\phi})\) 
\item \(\imodel{\I}{\phi\land\psi} = \imodel{\I}{\phi} \cap \imodel{\I}{\psi}\)

\item
\(\def\Im{\imodif[state]{\I}{x}{r}}%
\imodel{\I}{\lexists{x}{\phi}} =  \{\iget[state]{\I} \in \linterpretations{\Sigma}{V} \with \iget[state]{\Im} \in \imodel{\I}{\phi} ~\text{for some}~r\in\reals\}\)

\item
\(\def\Im{\imodif[state]{\I}{x}{r}}%
\imodel{\I}{\lforall{x}{\phi}} =  \{\iget[state]{\I} \in \linterpretations{\Sigma}{V} \with \iget[state]{\Im} \in \imodel{\I}{\phi} ~\text{for all}~r\in\reals\}\)

\item \label{case:diamond-semantics}
\(\imodel{\I}{\ddiamond{\alpha}{\phi}} = \iaccess[\alpha]{\I}\compose\imodel{\I}{\phi}\)
\(=\{\iget[state]{\I} \with \imodels{\It}{\phi} ~\text{for some}~\iget[state]{\It}~\text{such that}~\iaccessible[\alpha]{\I}{\It}\}\)

\item \(\imodel{\I}{\dbox{\alpha}{\phi}} = \imodel{\I}{\lnot\ddiamond{\alpha}{\lnot\phi}}\)
\(=\{\iget[state]{\I} \with \imodels{\It}{\phi} ~\text{for all}~\iget[state]{\It}~\text{such that}~\iaccessible[\alpha]{\I}{\It}\}\)

\end{compactenum}
A \dL formula $\phi$ is true at state $\iget[state]{\I}$ in $\iget[const]{\I}$, also written \(\iportray{\I} \models \phi\) iff \(\imodels{\I}{\phi}\).
A \dL formula $\phi$ is \emph{valid in $\iget[const]{\I}$}, written \m{\iget[const]{\I}\models{\phi}}, iff \m{\imodel{\I}{\phi}=\linterpretations{\Sigma}{V}}, i.e.\ \m{\imodels{\I}{\phi}} for all states $\iget[state]{\I}$.
Formula $\phi$ is \emph{valid}, written \m{\entails\phi}, iff \m{\iget[const]{\I}\models{\phi}} for all interpretations $\iget[const]{\I}$.
\end{definition}

\noindent
The relation composition operator $\circ$ in \rref{case:diamond-semantics} is also used for sets which are unary relations.
The interpretation of program constant $a$ is a state-transition relation \(\iget[const]{\I}(a)\subseteq\linterpretations{\Sigma}{V}\times\linterpretations{\Sigma}{V}\),
where \(\related{\iget[const]{\I}(a)}{\iget[state]{\I}}{\iget[state]{\It}}\) iff \HP $a$ can run from initial state $\iget[state]{\I}$ to final state $\iget[state]{\It}$.

\begin{definition}[Transition semantics of \HPs] \label{def:HP-transition}
For each interpretation $\iget[const]{\I}$, each \HP $\alpha$ is interpreted semantically as a binary transition relation \m{\iaccess[\alpha]{\I}\subseteq\linterpretations{\Sigma}{V}\times\linterpretations{\Sigma}{V}} on states, defined inductively by
\begin{compactenum}
\item \m{\iaccess[a]{\I} = \iget[const]{\I}(a)} for program constants $a$
\item
\(\def\Im{\imodif[state]{\I}{x}{r}}%
\iaccess[\pupdate{\pumod{x}{\theta}}]{\I} = \{(\iget[state]{\I},\iget[state]{\Im}) \with r=\ivaluation{\I}{\theta}\}
= \{(\iget[state]{\I},\iget[state]{\It}) \with 
 \iget[state]{\It}=\iget[state]{\I}~\text{except}~%
 \iget[state]{\It}(x)=\ivaluation{\I}{\theta}\}
\)

\item \m{\iaccess[\ptest{\ivr}]{\I} = \{(\iget[state]{\I},\iget[state]{\I}) \with \imodels{\I}{\ivr}\}}
\item
  \m{\iaccess[\pevolvein{\D{x}=\genDE{x}}{\ivr}]{\I} = \{(\iget[state]{\I},\iget[state]{\It}) \with
  \iget[state]{\I}=\iget[state]{\Iff[0]}} on $\scomplement{\{\D{x}\}}$ and \(\iget[state]{\It}=\iget[state]{\Iff[r]}\)
  for some function \m{\iget[flow]{\If}:[0,r]\to\linterpretations{\Sigma}{V}} of some duration $r$
  satisfying \m{\imodels{\If}{\D{x}=\genDE{x}\land\ivr}}$\}$
  \\
  where \m{\imodels{\If}{\D{x}=\genDE{x}\land\ivr}}
  iff
  \(\imodels{\Iff[\zeta]}{\D{x}=\genDE{x}\land\ivr}\)
  and
  \(\iget[state]{\Iff[0]}=\iget[state]{\Iff[\zeta]}\) on $\scomplement{\{x,\D{x}\}}$ 
  for all \(0\leq \zeta\leq r\)
  and if
  \(\D[t]{\iget[state]{\Iff[t]}(x)}(\zeta)\) exists and is equal to \(\iget[state]{\Iff[\zeta]}(\D{x})\) for all \(0\leq \zeta\leq r\).

\item \m{\iaccess[\pchoice{\alpha}{\beta}]{\I} = \iaccess[\alpha]{\I} \cup \iaccess[\beta]{\I}}

\item
\newcommand{\Iz}{\iconcat[state=\mu]{\I}}
\m{\iaccess[\alpha;\beta]{\I} = \iaccess[\alpha]{\I} \compose\iaccess[\beta]{\I}}
\(= \{(\iget[state]{\I},\iget[state]{\It}) : (\iget[state]{\I},\iget[state]{\Iz}) \in \iaccess[\alpha]{\I},  (\iget[state]{\Iz},\iget[state]{\It}) \in \iaccess[\beta]{\I} ~\text{for some}~\iget[state]{\Iz}\}\)

\item \m{\iaccess[\prepeat{\alpha}]{\I} = \displaystyle
\closureTransitive{\big(\iaccess[\alpha]{\I}\big)}
=
\cupfold_{n\in\naturals}\iaccess[{\prepeat[n]{\alpha}}]{\I}} 
with \m{\prepeat[n+1]{\alpha} \mequiv \prepeat[n]{\alpha};\alpha} and \m{\prepeat[0]{\alpha}\mequiv\,\ptest{\ltrue}}
\end{compactenum}
where $\closureTransitive{\rho}$ denotes the reflexive transitive closure of relation $\rho$.
\end{definition}
The equality in \m{\iaccess[\prepeat{\alpha}]{\I}} follows from the Scott-continuity of \HPs \cite[Lemma\,3.7]{DBLP:journals/tocl/Platzer15}.
The case \(\iaccess[\pevolvein{\D{x}=\genDE{x}}{\ivr}]{\I}\) expresses that $\iget[flow]{\If}$ solves the differential equation
   and satisfies $\ivr$ at all times.
   In case $r=0$, the only condition is that \(\iget[state]{\I}=\iget[state]{\It}\) on $\scomplement{\{\D{x}\}}$ and \(\iget[state]{\It}(\D{x})=\ivaluation{\It}{\genDE{x}}\) and \(\imodels{\It}{\ivr}\).
Since \(\iget[state]{\I}\) and \(\iget[state]{\Iff[0]}\) are only assumed to agree on the complement $\scomplement{\{\D{x}\}}$ of the set $\{\D{x}\}$, the initial values \(\iget[state]{\I}(\D{x})\) of differential symbols $\D{x}$ do \emph{not} influence the behavior of \(\iaccessible[\pevolvein{\D{x}=\genDE{x}}{\ivr}]{\I}{\It}\), because they may not be compatible with the time-derivatives for the differential equation, e.g. in \m{\Dupdate{\Dumod{\D{x}}{1}};\pevolve{\D{x}=2}}  with a discontinuity in $\D{x}$.
The final values \(\iget[state]{\It}(\D{x})\) after \(\pevolvein{\D{x}=\genDE{x}}{\ivr}\) will coincide with the derivatives at the final state, though, even for evolutions of duration zero.

\subsection{Static Semantics} \label{sec:static-semantics}

The dynamic semantics gives a precise meaning to \dL formulas and \HPs but is inaccessible for effective reasoning purposes.
By contrast, the static semantics of \dL and \HPs defines only simple aspects of the dynamics concerning the variable usage that follows more directly from the syntactic structure without running the programs or evaluating their dynamical effects.
The correctness of uniform substitutions depends only on the static semantics, which identifies free variables ($\freevarsdef{\theta},\freevarsdef{\phi},\freevarsdef{\alpha}$ of terms $\theta$, formulas $\phi$ and programs $\alpha$) and bound variables ($\boundvarsdef{\alpha}$).
The static semantics first characterizes free and bound variables semantically from the dynamic semantics and subsequently shows algorithms for computing them conservatively.
\begin{definition}[Static semantics] \label{def:static-semantics}
The \emph{static semantics} defines the \emph{free variables}, which are all variables that the value of an expression depends on,
as well as \emph{bound variables}, which can change their value during the evaluation of an expression, as follows:
\begin{align*}
  \freevarsdef{\theta} &= 
\renewcommand{\Ialt}{\vdLint[const=I,state=\tilde{\nu}]}%
  \cupfold \{x \in \allvars \,:~ \text{there are}~\iget[const]{\I}~\text{and}~\iget[state]{\I}=\iget[state]{\Ialt} ~\text{on}~\scomplement{\{x\}} ~ \text{such that} ~ \ivaluation{\I}{\theta}\neq\ivaluation{\Ialt}{\theta}\}\\
  \freevarsdef{\phi} &= \cupfold \{x \in \allvars \,:~ \text{there are}~\iget[const]{\I}~\text{and}~\iget[state]{\I}=\iget[state]{\Ialt} ~\text{on}~\scomplement{\{x\}} ~ \text{such that} ~ \imodels{\I}{\theta}\not\ni\iget[state]{\Ialt}\}\\
  \freevarsdef{\alpha} &= \cupfold \{x \in \allvars \,:~ \text{there are}~\iget[const]{\I},\iget[state]{\I},\iget[state]{\Ialt},\iget[state]{\It} ~\text{such that}~\iget[state]{\I}=\iget[state]{\Ialt} ~\text{on}~\scomplement{\{x\}}~ \text{and}~\iaccessible[\alpha]{\I}{\It}\\
  &
\renewcommand{\Ialt}{\vdLint[const=I,state=\tilde{\nu}]}%
  \phantom{= \cupfold \{x \in \allvars \,:~}~
  \text{but there is no}~\iget[state]{\Italt} ~\text{with}~\iget[state]{\It}=\iget[state]{\Italt} ~\text{on}~\scomplement{\{x\}} ~\text{such that}~\iaccessible[\alpha]{\Ialt}{\Italt}\}\\
  \boundvarsdef{\alpha} &= \cupfold \{x \in \allvars \,:~ \text{there are}~\iget[const]{\I}~\text{and}~\iaccessible[\alpha]{\I}{\It} ~ \text{such that} ~ \iget[state]{\I}(x)\neq\iget[state]{\It}(x)\}
\end{align*}
The \dfn{signature}, i.e.\ set of function, predicate, quantifier symbols, and program constants in $\phi$ is denoted \(\intsigns{\phi}\); accordingly \(\intsigns{\theta}\) for term\,$\theta$ and \(\intsigns{\alpha}\) for program\,$\alpha$.
\end{definition}

For example, only $\{v,b,x\}$ are free variables of the formula \rref{eq:no-backwards}, yet $\{a,x,\D{x},v,\D{v}\}$ are the bound variables of its program. 
Acceleration $a$ is not a free variable of \rref{eq:no-backwards},
because $a$ is never actually read, as $a$ must have been written on \emph{every} execution path before being read anywhere.
No execution of the program in \rref{eq:no-backwards} depends on the initial value of $a$, so $a$ is not free since $a$ is not free after the loop or in the postcondition.

The static semantics provides uniform substitutions with all they need to know to determine what changes during substitutions go unnoticed (only changes to free variables have an impact on the value of an expression \rref{lem:coincidence-term}--\ref{lem:coincidence-HP}) and what state-change an expression may cause itself (only bound variables can change their value during the evaluation of an expression \rref{lem:bound}).
Whether a uniform substitution preserves truth in a proof depends on the interaction of the free and bound variables.
If it introduces a free variable into a context where that variable is bound, then the possible change in value of that bound variable may affect the overall truth-value.

The first property that uniform substitutions depend on is that \HPs have bounded effect: only bound variables of \HP $\alpha$ are modified during runs of $\alpha$.

\begin{lemma}[Bound effect] \label{lem:bound}
  The set $\boundvarsdef{\alpha}$ is the smallest set with the bound effect property: 
  If \(\iaccessible[\alpha]{\I}{\It}\), then \(\iget[state]{\I}=\iget[state]{\It}\) on $\scomplement{\boundvarsdef{\alpha}}$.
\end{lemma}
\begin{proofatend}
First prove that $\boundvarsdef{\alpha}$ has the bound effect property.
Consider any \(\iaccessible[\alpha]{\I}{\It}\) and $x\not\in\boundvarsdef{\alpha}$ then \(\iget[state]{\I}(x)=\iget[state]{\It}(x)\) by \rref{def:static-semantics}.

Suppose there was a set $V\not\supseteq\boundvarsdef{\alpha}$ satisfying the bound effect property for $\alpha$.
Then there is a variable $x\in\boundvarsdef{\alpha}\setminus V$, which implies that
there are $\iget[const]{\I}$ and there are \(\iaccessible[\alpha]{\I}{\It}\) such that \(\iget[state]{\I}(x)\neq\iget[state]{\It}(x)\).
But, then $V$ does not have the bound effect property, as \(\iaccessible[\alpha]{\I}{\It}\) but it is not the case that \(\iget[state]{\I}=\iget[state]{\It}\) on $\scomplement{V}$, since $\scomplement{V}\supseteq\{x\}$ yet \(\iget[state]{\I}(x)\neq\iget[state]{\It}(x)\).
\end{proofatend}
With a small-step operational semantics, a corresponding notion of bound variables of a formula could be defined as those that change their value during the evaluation of formulas, but that is not needed here.

The value of a term only depends on the values of its free variables.
When evaluating a term $\theta$ in two different states $\iget[state]{\I}$, $\iget[state]{\Ialt}$ that agree on its free variables $\freevarsdef{\theta}$, the values of $\theta$ in both states coincide.
Accordingly, the value of a term will agree for different interpretations \(\iget[const]{\I},\iget[const]{\Ialt}\) that agree on the symbols $\intsigns{\theta}$ that occur in $\theta$.

\begin{lemma}[Coincidence for terms] \label{lem:coincidence-term}
  The set $\freevarsdef{\theta}$ is the smallest set with the coincidence property for $\theta$:
  If \(\iget[state]{\I}=\iget[state]{\Ialt}\) on $\freevarsdef{\theta}$
  and \(\iget[const]{\I}=\iget[const]{\Ialt}\) on $\intsigns{\theta}$, then
  \m{\ivaluation{\I}{\theta}=\ivaluation{\Ialt}{\theta}}.
\end{lemma}
\begin{proofatend}
\renewcommand{\Ialt}{\vdLint[const=I,state=\tilde{\nu}]}%
To prove that $\freevarsdef{\theta}$ has the coincidence property, it is enough to show by induction that, for any set of variables $S\subseteq\scomplement{\freevarsdef{\theta}}$, the state $\iget[state]{\Imid}$ in between $\iget[state]{\I}$ and $\iget[state]{\Ialt}$ that is defined as \(\iget[state]{\Imid}=\iget[state]{\I}\) on $S$ and as \(\iget[state]{\Imid}=\iget[state]{\Ialt}$ on $\scomplement{S}$ agrees with $\iget[state]{\Ialt}$ in the value \(\ivaluation{\Imid}{\theta}=\ivaluation{\Ialt}{\theta}\).

\begin{compactenum}
\addtocounter{enumi}{-1}
\item For $S=\emptyset$, there is nothing to show as \(\iget[state]{\Imid}=\iget[state]{\Ialt}\).

\item For $S\cup\{z\}$ with a variable $z\not\in\freevarsdef{\theta}$, abbreviate the modified state $\modif{\iget[state]{\Imid}}{z}{\iget[state]{\I}(z)}$ by $\iget[state]{\Imidstep}$, which satisfies
\(\iget[state]{\Imidstep}=\iget[state]{\Imid}\) on $\scomplement{\{z\}}$,
so \(\ivaluation{\Imidstep}{\theta}=\ivaluation{\Imid}{\theta}=\ivaluation{\Ialt}{\theta}\), because $z\not\in\freevarsdef{\theta}$ and by induction hypothesis.
\end{compactenum}
When $S$ is the set of all variables where $\iget[state]{\I}$ and $\iget[state]{\Ialt}$ differ, which is $S\subseteq\scomplement{\freevarsdef{\theta}}$ by assumption, this implies $\iget[state]{\Imid}=\iget[state]{\I}$ so
\(\ivaluation{\I}{\theta}=\ivaluation{\Ialt}{\theta}\).
{%
\renewcommand{\I}{\vdLint[const=I,state=\tilde{\nu}]}%
\renewcommand{\Ialt}{\vdLint[const=J,state=\tilde{\nu}]}%
Finally, if \(\iget[const]{\I}=\iget[const]{\Ialt}\) on $\intsigns{\theta}$ then also \(\ivaluation{\I}{\theta}=\ivaluation{\Ialt}{\theta}\) by an induction.}

Suppose there was a set $V\not\supseteq\freevarsdef{\theta}$ satisfying the coincidence property for $\theta$.
Then there is a variable $x\in\freevarsdef{\theta}\setminus V$,
which implies that there are $\iget[const]{\I}$, $\iget[state]{\I}=\iget[state]{\Ialt}$ on $\scomplement{\{x\}}$ such that \(\ivaluation{\I}{\theta}\neq\ivaluation{\Ialt}{\theta}\).
Then $V$ does not have the coincidence property, as $\iget[state]{\I}=\iget[state]{\Ialt}$ on $V$ but \(\ivaluation{\I}{\theta}\neq\ivaluation{\Ialt}{\theta}\).
\end{proofatend}

\noindent
In particular, the semantics of differentials is a sum over just the free variables:
\[
\ivaluation{\I}{\der{\theta}}
=
\sum_{x\in\freevarsdef{\theta}} \iget[state]{\I}(\D{x}) \Dp[x]{\ivaluation{\Idot}{\theta}}(\iget[state]{\I})
= \sum_{x\in\freevarsdef{\theta}} \iget[state]{\I}(\D{x}) \Dp[x]{\ivaluation{\I}{\theta}}
\]

When evaluating a \dL formula $\phi$ in two different states $\iget[state]{\I}$, $\iget[state]{\Ialt}$ that agree on its free variables $\freevarsdef{\phi}$ in \(\iget[const]{\I}=\iget[const]{\Ialt}\) on $\intsigns{\phi}$, the truth-values of $\phi$ in both states coincide.

\begin{lemma}[Coincidence for formulas] \label{lem:coincidence}
  The set $\freevarsdef{\phi}$ is the smallest set with the coincidence property for $\phi$:
  If \(\iget[state]{\I}=\iget[state]{\Ialt}\) on $\freevarsdef{\phi}$
  and \(\iget[const]{\I}=\iget[const]{\Ialt}\) on $\intsigns{\phi}$, then
  \m{\imodels{\I}{\phi}} iff \m{\imodels{\Ialt}{\phi}}.
\end{lemma}
\begin{proofatend}
\renewcommand{\Ialt}{\vdLint[const=I,state=\tilde{\nu}]}%
To prove that $\freevarsdef{\phi}$ has the coincidence property, it is enough to show by induction that, for any set of variables $S\subseteq\scomplement{\freevarsdef{\phi}}$, the state $\iget[state]{\Imid}$ in between $\iget[state]{\I}$ and $\iget[state]{\Ialt}$ that is defined as \(\iget[state]{\Imid}=\iget[state]{\I}\) on $S$ and as \(\iget[state]{\Imid}=\iget[state]{\Ialt}$ on $\scomplement{S}$ agrees with $\iget[state]{\Ialt}$ in the truth-value \(\imodels{\Imid}{\phi}\) iff \(\imodels{\Ialt}{\phi}\).

\begin{compactenum}
\addtocounter{enumi}{-1}
\item For $S=\emptyset$, there is nothing to show as \(\iget[state]{\Imid}=\iget[state]{\Ialt}\).

\item For $S\cup\{z\}$ with a variable $z\not\in\freevarsdef{\phi}$, abbreviate the modified state $\modif{\iget[state]{\Imid}}{z}{\iget[state]{\I}(z)}$ by $\iget[state]{\Imidstep}$, which satisfies
\(\iget[state]{\Imidstep}=\iget[state]{\Imid}\) on $\scomplement{\{z\}}$,
so \(\imodels{\Imidstep}{\phi}\) iff \(\imodels{\Imid}{\phi}\), because $z\not\in\freevarsdef{\phi}$, iff \(\imodels{\Ialt}{\phi}\) by induction hypothesis.
\end{compactenum}
When $S$ is the set of all variables where $\iget[state]{\I}$ and $\iget[state]{\Ialt}$ differ, which is $S\subseteq\scomplement{\freevarsdef{\phi}}$ by assumption, this implies $\iget[state]{\Imid}=\iget[state]{\I}$ so
\(\imodels{\I}{\phi}\) iff \(\imodels{\Ialt}{\phi}\).
{%
\renewcommand{\I}{\vdLint[const=I,state=\tilde{\nu}]}%
\renewcommand{\Ialt}{\vdLint[const=J,state=\tilde{\nu}]}%
Finally, if \(\iget[const]{\I}=\iget[const]{\Ialt}\) on $\intsigns{\phi}$ then also \(\imodels{\I}{\phi}\) iff \(\imodels{\Ialt}{\phi}\) by an induction.}

Suppose there was a set $V\not\supseteq\freevarsdef{\phi}$ satisfying the coincidence property for $\phi$.
Then there is a variable $x\in\freevarsdef{\phi}\setminus V$,
which implies that there are $\iget[const]{\I}$, $\iget[state]{\I}=\iget[state]{\Ialt}$ on $\scomplement{\{x\}}$ such that \(\imodels{\I}{\phi}\not\ni\iget[state]{\Ialt}\).
Then $V$ does not have the coincidence property, as $\iget[state]{\I}=\iget[state]{\Ialt}$ on $V$ but \(\imodels{\I}{\phi}\not\ni\iget[state]{\Ialt}\).
\end{proofatend}

The runs of an \HP $\alpha$ only depend on the values of its free variables, because its behavior cannot depend on the values of variables that it never reads.
If \(\iget[state]{\I}=\iget[state]{\Ialt}\) on $\freevarsdef{\alpha}$ and \(\iget[const]{\I}=\iget[const]{\Ialt}\) on $\intsigns{\phi}$
  and \(\iaccessible[\alpha]{\I}{\It}\),
  then there is an $\iget[state]{\Italt}$ such that
  \(\iaccessible[\alpha]{\Ialt}{\Italt}\)
  and $\iget[state]{\It}$ and $\iget[state]{\Italt}$ agree on $\freevarsdef{\alpha}$.
In fact, the final states $\iget[state]{\It},\iget[state]{\Italt}$ continue to agree on any set $V\supseteq\freevarsdef{\alpha}$ that the initial states $\iget[state]{\I},\iget[state]{\Ialt}$ agreed on.
The respective pairs of initial and final states of a run of \HP $\alpha$ already agree on the complement $\scomplement{\boundvarsdef{\alpha}}$ by \rref{lem:bound}.

\begin{lemma}[Coincidence for programs] \label{lem:coincidence-HP}
  The set $\freevarsdef{\alpha}$ is the smallest set with the coincidence property for $\alpha$:
  If \(\iget[state]{\I}=\iget[state]{\Ialt}\) on $V\supseteq\freevarsdef{\alpha}$,
  \(\iget[const]{\I}=\iget[const]{\Ialt}\) on $\intsigns{\alpha}$
  and \(\iaccessible[\alpha]{\I}{\It}\),
  then there is a $\iget[state]{\Italt}$ such that
  \(\iaccessible[\alpha]{\Ialt}{\Italt}\)
  and 
  \(\iget[state]{\It}=\iget[state]{\Italt}\) on $V$.
\begin{center}
\begin{tikzpicture}
  \matrix (m) [matrix of math nodes,row sep=3em,column sep=4em,minimum width=2em]
  {
     \iget[state]{\I} & \iget[state]{\It} \\
     \iget[state]{\Ialt} & \iget[state]{\Italt} \\};
  \path
    (m-1-1) edge [similar state] node [left,align=right] {on $V\supseteq\freevarsdef{\alpha}$} (m-2-1)
            edge [transition] node [below] {$\alpha$} (m-1-2)
    (m-2-1) edge [transition,transition exists] node [above] {$\alpha$}
            (m-2-2)
    (m-1-2) edge [similar state] node [right,align=right] {on $V$} (m-2-2);
  \path
    (m-1-1) edge[similar state,bend left=30] node[above] {on $\scomplement{\boundvarsdef{\alpha}}$} (m-1-2)
    (m-2-1) edge[similar state,bend right=30] node[below] {on $\scomplement{\boundvarsdef{\alpha}}$} (m-2-2);
\end{tikzpicture}
\end{center}
\end{lemma}
\begin{proofatend}
\renewcommand{\Ialt}{\vdLint[const=I,state=\tilde{\nu}]}%
\renewcommand{\Italt}{\vdLint[const=I,state=\tilde{\omega}]}%

To prove that $\freevarsdef{\alpha}$ has the coincidence property, it is enough to show by induction for all $S\subseteq\scomplement{\freevarsdef{\alpha}}$ that the state $\iget[state]{\Imid}$ in between $\iget[state]{\Ialt}$ and $\iget[state]{\I}$ that is defined as $\iget[state]{\Imid}=\iget[state]{\Ialt}$ on $S$ and as $\iget[state]{\Imid}=\iget[state]{\I}$ on $\scomplement{S}$ has a state $\iget[state]{\Itmid}$ that agrees with $\iget[state]{\Itmid}=\iget[state]{\It}$ on $\scomplement{S}$ such that \(\iaccessible[\alpha]{\Imid}{\Itmid}\) as well.

\begin{enumerate}
\addtocounter{enumi}{-1}
\item For $S=\emptyset$, there is nothing to show for $\iget[state]{\Imid}=\iget[state]{\I}$ and $\iget[state]{\Itmid}=\iget[state]{\It}$, because \(\iaccessible[\alpha]{\I}{\It}\).

\item For $S\cup\{z\}$ with a variable $z\not\in\freevarsdef{\alpha}$,
let $\iget[state]{\Imidstep}$ denote modified state $\modif{\iget[state]{\Imid}}{z}{\iget[state]{\Ialt}(z)}$.
Then $\iget[state]{\Imidstep}=\iget[state]{\Imid}$ on $\scomplement{\{z\}}$ and, by induction hypothesis, \(\iaccessible[\alpha]{\Imid}{\Itmid}\) for some $\iget[state]{\Itmid}$ with \(\iget[state]{\Itmid}=\iget[state]{\It}$ on $\scomplement{S}$.
Since $z\not\in\freevarsdef{\alpha}$, this implies there is a state $\iget[state]{\Itmidstep}$ such that $\iget[state]{\Itmidstep}=\iget[state]{\Itmid}$ on $\scomplement{\{z\}}$ and \(\iaccessible[\alpha]{\Imidstep}{\Itmidstep}\).
Thus, \(\iget[state]{\Itmidstep}=\iget[state]{\Itmid}=\iget[state]{\It}\) on $\scomplement{(S\cup\{z\})}$.
\end{enumerate}
Finally the state $\iget[state]{\Imid}$ resulting for $S=\scomplement{V}$ satisfies \(\iget[state]{\Imid}=\iget[state]{\Ialt}\) because \(\iget[state]{\Imid}=\iget[state]{\Ialt}\) on $\scomplement{V}$ and \(\iget[state]{\Imid}=\iget[state]{\I}=\iget[state]{\Ialt}\) on $\scomplement{(\scomplement{V})}=V$ already
and the state $\iget[state]{\Italt}$ defined as $\iget[state]{\Itmid}$ satisfies $\iget[state]{\Itmid}=\iget[state]{\It}$ on $V$ and \(\iaccessible[\alpha]{\Ialt}{\Italt}\).
{%
\renewcommand{\I}{\vdLint[const=I,state=\tilde{\nu}]}%
\renewcommand{\It}{\vdLint[const=I,state=\tilde{\nu}]}%
\renewcommand{\Ialt}{\vdLint[const=J,state=\tilde{\nu}]}%
\renewcommand{\Italt}{\vdLint[const=J,state=\tilde{\omega}]}%
Finally, if \(\iget[const]{\I}=\iget[const]{\Ialt}\) on $\intsigns{\alpha}$ then also \(\iaccessible[\alpha]{\Ialt}{\Italt}\).
}

Suppose there was a set $V\not\supseteq\freevarsdef{\alpha}$ satisfying the coincidence property for $\alpha$.
Then there is a variable $x\in\freevarsdef{\alpha}\setminus V$,
which implies that there are $\iget[const]{\I},\iget[state]{\I},\iget[state]{\Ialt},\iget[state]{\It}$ such that \(\iget[state]{\I}=\iget[state]{\Ialt}\) on $\scomplement{\{x\}}$ and \(\iaccessible[\alpha]{\I}{\It}\),
but there is no $\iget[state]{\Italt}$ with \(\iget[state]{\It}=\iget[state]{\Italt}\) on $\scomplement{\{x\}}$  such that \(\iaccessible[\alpha]{\Ialt}{\Italt}\).
Then $V$ does not have the coincidence property, because \(\iget[state]{\I}=\iget[state]{\Ialt}\) on $\scomplement{\{x\}}\supseteq V$ and \(\iaccessible[\alpha]{\I}{\It}\),
but there is no $\iget[state]{\Italt}$ with \(\iget[state]{\It}=\iget[state]{\Italt}\) on $\scomplement{\{x\}}$  such that \(\iaccessible[\alpha]{\Ialt}{\Italt}\).
\end{proofatend}

\subsection{Correct Static Semantics Computations} \label{sec:static-semantics-computation}

Lemmas~\ref{lem:bound}--\ref{lem:coincidence-HP} hold for any superset of $\boundvarsdef{\alpha},\freevarsdef{\theta},\freevarsdef{\phi},\freevarsdef{\alpha}$, respectively.
Supersets of the static semantics can be computed easily from the syntactic structure and provide the sole input that uniform substitutions depend on, which, in turn, are the only part of the calculus where the static semantics is relevant.
Only variables that are read in a formula or program can be free variables.
And only variables that have quantifiers or are written to or have differential equations can be bound variables.

Bound variables $x$ of a formula are those that are bound by $\lforall{x}{}$or $\lexists{x}{}$, but also those that are bound by modalities such as \(\dbox{\pupdate{\pumod{x}{5y}}}{}\)
or \(\ddiamond{\pevolve{\D{x}=1}}{}\)
or \(\dbox{\pchoice{\pumod{x}{1}}{\pevolve{\D{x}=1}}}{}\)
or \(\dbox{\pchoice{\pumod{x}{1}}{\ptest{\ltrue}}}{}\) because of the assignment to $x$ or differential equation for $x$ they contain.
The scope of the bound variable $x$ is limited to the quantified formula or to the postcondition and remaining program of a modality.
\begin{definition}[Bound variable] \label{def:boundvars}
  The set $\boundvars{\phi}\subseteq\allvars$ of (syntactically) \emph{bound variables} of \dL formula $\phi$ is defined inductively as:
  \begin{align*}
  \boundvars{p(\theta_1,\dots,\theta_k)} &= \emptyset
  &&\text{where $p$ can also be $\geq$}
  \\
  \boundvars{\contextapp{C}{\phi}} &= \allvars\\ %
  \boundvars{\lnot\phi} &= \boundvars{\phi}\\
  \boundvars{\phi\land\psi} &= \boundvars{\phi}\cup\boundvars{\psi}\\
  \boundvars{\lforall{x}{\phi}} = \boundvars{\lexists{x}{\phi}} &= \{x\}\cup\boundvars{\phi}\\
  \boundvars{\dbox{\alpha}{\phi}} = \boundvars{\ddiamond{\alpha}{\phi}} &= \boundvars{\alpha}\cup\boundvars{\phi}
  \end{align*}
  The set $\boundvars{\alpha}\subseteq\allvars$ of (syntactically) \emph{bound variables} of \HP $\alpha$, i.e.\ all those that may potentially be written to, is defined inductively as:
  \begin{align*}
  \boundvars{a} &= \allvars &&\text{for program constant $a$}\\
  \boundvars{\pupdate{\pumod{x}{\theta}}} &= \{x\}
  \\
  \boundvars{\ptest{\ivr}} &= \emptyset
  \\
  \boundvars{\pevolvein{\D{x}=\genDE{x}}{\ivr}} &= \{x,\D{x}\}
  \\
  \boundvars{\alpha\cup\beta} = \boundvars{\alpha;\beta} &= \boundvars{\alpha}\cup\boundvars{\beta}
  \\
  \boundvars{\prepeat{\alpha}} &= \boundvars{\alpha}
  \end{align*}
\end{definition}
\noindent
Both $x$ and $\D{x}$ are bound by a differential equation \m{\pevolve{\D{x}=\genDE{x}}}, as both may change their value.
All variables $\allvars$ is the only option for program constants and quantifier symbols $C$, since, depending on their interpretation, both may change the value of any $x\in\allvars$.

The free variables of a quantified formula are defined by removing its bound variable as \(\freevars{\lforall{x}{\phi}} = \freevars{\phi}\setminus\{x\}\), since all occurrences of $x$ in $\phi$ are bound by $\forall{x}$.
The bound variables of a program in a modality act in a similar way, except that the program itself may read variables during the computation, so its free variables need to be taken into account.
By analogy to the quantifier case, it is often suspected that
\(\freevars{\dbox{\alpha}{\phi}}\) could be defined as \(\freevars{\alpha}\cup(\freevars{\phi}\setminus\boundvars{\alpha})\).
But that would be unsound, because
\(\dbox{\pchoice{\pupdate{\pumod{x}{1}}}{\pupdate{\pumod{y}{2}}}}{\,x\geq1}\) would have no free variables then,
contradicting the fact that its truth-value depends on the initial value of $x$.
The reason is that $x$ is a bound variable of that program, but only written to on some but not on all paths. So the initial value of $x$ may be needed to evaluate the truth of the postcondition $x\geq1$ on some execution paths.
If a variable is must-bound, so written to on all paths of the program, however, it can safely be removed from the free variables of the postcondition.
The static semantics defines the subset of variables that are must-bound ($\mustboundvars{\alpha}$), so must be written to on all execution paths of $\alpha$.
This complication does not happen for ordinary quantifiers or strictly nested languages like pure $\lambda$-calculi.

\begin{definition}[Must-bound variable] \label{def:mustboundvar}
  The set $\mustboundvars{\alpha}\subseteq\boundvars{\alpha}\subseteq\allvars$ of (syntactically) \emph{must-bound variables} of \HP $\alpha$, i.e.\ all those that must be written to on all paths of $\alpha$, is defined inductively as:
  \begin{align*}
  \mustboundvars{a} &= \emptyset &&\text{for program constant $a$}\\
    \mustboundvars{\alpha} &= \boundvars{\alpha} &&\text{for atomic \HPs $\alpha$ except program constants}
    \\
  \mustboundvars{\alpha\cup\beta} &= \mustboundvars{\alpha}\cap\mustboundvars{\beta}
  \\
  \mustboundvars{\alpha;\beta} &= \mustboundvars{\alpha}\cup\mustboundvars{\beta}
  \\
  \mustboundvars{\prepeat{\alpha}} &= \emptyset
  \end{align*}
\end{definition}

Finally, the static semantics also defines which variables are free so may be read.
The definition of free variables is simultaneously inductive for formulas $(\freevars{\phi})$ and programs $(\freevars{\alpha})$ owing to their mutually recursive syntactic structure.

\begin{definition}[Free variable] \label{def:freevars}
  The set $\freevars{\theta}\subseteq\allvars$ of (syntactically) \emph{free variables} of term $\theta$, i.e.\ those that occur in $\theta$ directly or indirectly, is defined inductively as:
  \begin{align*}
  \freevars{x} &= \{x\} &&\text{hence \(\freevars{\D{x}} = \{\D{x}\}\)}\\
  \freevars{f(\theta_1,\dots,\theta_k)} &= \freevars{\theta_1}\cup\dots\cup\freevars{\theta_k}
  &&\text{where $f$ can also be $+$ or $\cdot$}
  \\
  \freevars{\der{\theta}} &= \freevars{\theta} \cup \D{\freevars{\theta}}
  \end{align*}
  The set $\freevars{\phi}$ of (syntactically) \emph{free variables} of \dL formula $\phi$, i.e.\ all that occur in $\phi$ outside the scope of quantifiers or modalities binding it, is defined inductively as:
  \begin{align*}
  \freevars{p(\theta_1,\dots,\theta_k)} &= \freevars{\theta_1}\cup\dots\cup\freevars{\theta_k}
  &&\text{where $p$ can also be $\geq$}
  \\
  \freevars{\contextapp{C}{\phi}} &= \allvars\\ %
  \freevars{\lnot\phi} &= \freevars{\phi}\\
  \freevars{\phi\land\psi} &= \freevars{\phi}\cup\freevars{\psi}\\
  \freevars{\lforall{x}{\phi}} = \freevars{\lexists{x}{\phi}} &= \freevars{\phi}\setminus\{x\}\\
  \freevars{\dbox{\alpha}{\phi}} = \freevars{\ddiamond{\alpha}{\phi}} &= \freevars{\alpha}\cup(\freevars{\phi}\setminus\mustboundvars{\alpha})
  \end{align*}
  The set $\freevars{\alpha}\subseteq\allvars$ of (syntactically) \emph{free variables} of \HP $\alpha$, i.e.\ all those that may potentially be read, is defined inductively as:
  \begin{align*}
  \freevars{a} &= \allvars &&\hspace{-10pt}\text{for program constant $a$}\\
  \freevars{\pupdate{\pumod{x}{\theta}}} &= \freevars{\theta}
  \\
  \freevars{\ptest{\ivr}} &= \freevars{\ivr}
  \\
  \freevars{\pevolvein{\D{x}=\genDE{x}}{\ivr}} &= \{x\}\cup\freevars{\genDE{x}}\cup\freevars{\ivr}
  \\
  \freevars{\pchoice{\alpha}{\beta}} &= \freevars{\alpha}\cup\freevars{\beta}
  \\
  \freevars{\alpha;\beta} &= \freevars{\alpha}\cup(\freevars{\beta}\setminus\mustboundvars{\alpha})
  \\
  \freevars{\prepeat{\alpha}} &= \freevars{\alpha}
  \end{align*}
The \emph{variables} of \dL formula $\phi$, whether free or bound, are \(\vars{\phi}=\freevars{\phi}\cup\boundvars{\phi}\).
The \emph{variables} of \HP $\alpha$, whether free or bound, are \(\vars{\alpha}=\freevars{\alpha}\cup\boundvars{\alpha}\).
\end{definition}
Soundness requires \(\freevars{\der{\theta}}\) to be the union of $\freevars{\theta}$ and its differential closure \(\D{\freevars{\theta}}\) of all differential symbols corresponding to the variables in $\freevars{\theta}$,
because
the value of \(\der{x y}\) depends on $\freevarsdef{\der{x y}}=\{x,\D{x},y,\D{y}\}$ so the current and differential symbol values. Indeed, \(\der{x y}\) will turn out to equal \(\D{x} y + x \D{y}\) (\rref{lem:derivationLemma}), which has the same set of free variables $\{x,\D{x},y,\D{y}\}$ for more obvious reasons.
Both $x$ and $\D{x}$ are bound in \(\pevolvein{\D{x}=\genDE{x}}{\ivr}\) since both change their value, but only $x$ is added to the free variables, because the behavior can only depend on the initial value of $x$, not of that of $\D{x}$.
  All variables $\allvars$ are free and bound variables for program constants $a$, because their effect depends on the interpretation $\iget[const]{\I}$, so they may read and write any variable in \(\freevars{a}=\boundvars{a}=\allvars\) but possibly not on all paths, so \(\mustboundvars{a}=\emptyset\).

For example, $\freevars{\phi}=\freevarsdef{\phi}=\{v,b,x\}$ are the free variables of the formula $\phi$ in \rref{eq:no-backwards}, while $\boundvars{\alpha}=\boundvarsdef{\alpha}=\{a,x,\D{x},v,\D{v}\}$ are the bound variables (and must-bound variables) of its program $\alpha$. 
This would have been different for the less precise definition
\(
\freevars{\alpha;\beta} = \freevars{\alpha}\cup\freevars{\beta}
\).
Of course \cite{DBLP:journals/ams/Rice53}, syntactic computations may give bigger sets, e.g.,
\(\freevars{x^2-x^2}=\{x\}\neq\freevarsdef{x^2-x^2}=\emptyset\) or 
\(\boundvars{\pupdate{\pumod{x}{x}}}=\{x\}\neq\boundvarsdef{\pupdate{\pumod{x}{x}}}=\emptyset\), or similarly when some differential equation can never be executed.

Since uniform substitutions depend on the static semantics, soundness of uniform substitutions requires the static semantics to be computed correctly.
Correctness of the static semantics is easy to prove by straightforward structural induction with some attention for differential cases.
There is a subtlety in the soundness proof for the free variables of programs and formulas, though.
The states $\iget[state]{\It}$ and $\iget[state]{\Italt}$ resulting from \rref{lem:coincidence-HP} continue to agree on $\freevars{\alpha}$ and the variables that are bound on the particular path that $\alpha$ ran for the transition \(\iaccessible[\alpha]{\I}{\It}\).
They may disagree on variables $z$ that are neither free (so the initial states $\iget[state]{\I}$ and $\iget[state]{\Ialt}$ have not been assumed to coincide) nor bound on the particular path that $\alpha$ took, because $z$ has not been written to.
\begin{example}[Bound variables may not agree after an \HP]
Let \(\iaccessible[\alpha]{\I}{\It}\).
It is not enough to assume \(\iget[state]{\I}=\iget[state]{\Ialt}\) only on $\freevars{\alpha}$ in order to guarantee \(\iget[state]{\It}=\iget[state]{\Italt}\) on $\vars{\alpha}$ for some $\iget[state]{\Italt}$ such that
  \(\iaccessible[\alpha]{\Ialt}{\Italt}\), because
\(
\alpha \mdefequiv \pchoice{\pupdate{\pumod{x}{1}}}{\pupdate{\pumod{y}{2}}}
\)
will force the final states to agree only on either $x$ or on $y$, whichever one was assigned to during the respective run of $\alpha$, not on both \(\boundvars{\alpha}=\{x,y\}\),
even though any initial states \(\iget[state]{\I},\iget[state]{\Ialt}\) agree on $\freevars{\alpha}=\emptyset$.
This can only happen because \(\emptyset=\mustboundvars{\alpha}\neq\boundvars{\alpha}=\{x,y\}\).
\end{example}

Yet, the respective resulting states $\iget[state]{\It}$ and $\iget[state]{\Italt}$ still agree on the must-bound variables that are bound on \emph{all} paths of $\alpha$, rather than just somewhere in $\alpha$.
If initial states agree on (at least) all free variables $\freevars{\alpha}$ that \HP $\alpha$ may read, then the final states continue to agree on those (even if overwritten since) as well as on all variables that $\alpha$ must write on all paths, i.e.\ $\mustboundvars{\alpha}$.
This is crucial for the soundness proof of the syntactic static semantics,
because, e.g., free occurrences in $\phi$ of must-bound variables of $\alpha$ will not be free in \(\dbox{\alpha}{\phi}\), so the initial states will not have been assumed to agree initially.
It is of similar significance that the resulting states continue to agree on any superset $V\supseteq\freevars{\alpha}$ of the free variables that the initial states agreed on.

\begin{lemma}[Soundness of static semantics] \label{lem:sound-static-semantics}
The static semantics correctly computes supersets and, thus, Lemmas~\ref{lem:bound}--\ref{lem:coincidence-HP} hold for $\boundvars{\alpha},\freevars{\theta},\freevars{\phi},\freevars{\alpha}$:
\[
\boundvars{\alpha} \supseteq \boundvarsdef{\alpha}
\quad
\freevars{\theta} \supseteq \freevarsdef{\theta}
\quad
\freevars{\phi} \supseteq \freevarsdef{\phi}
\quad
\freevars{\alpha} \supseteq \freevarsdef{\alpha}
\]
\end{lemma}
\begin{proofatend}
First prove \(\boundvars{\alpha} \supseteq \boundvarsdef{\alpha}\) by a straightforward structural induction on $\alpha$.
\begin{compactenum}
\item For program constant $a$, the statement is obvious, since \(\boundvars{a}=\allvars\).

\item \m{\iaccessible[\pupdate{\pumod{x}{\theta}}]{\I}{\It}}
\def\Im{\imodif[state]{\I}{x}{r}}%
iff 
\(\iget[state]{\It}=\iget[state]{\Im}\) with \(r=\ivaluation{\I}{\theta}\)
so \(\iget[state]{\I}=\iget[state]{\It}\) except for \(\{x\}=\boundvars{\pupdate{\pumod{x}{\theta}}} \supseteq \boundvarsdef{\pupdate{\pumod{x}{\theta}}}\).

\item \m{\iaccessible[\ptest{\ivr}]{\I}{\It} = \{(\iget[state]{\I},\iget[state]{\I}) \with \imodels{\I}{\ivr}\}}
implies \(\iget[state]{\I}=\iget[state]{\It}\) so
\(\boundvarsdef{\ptest{\ivr}}=\boundvars{\ptest{\ivr}}=\emptyset\).

\item
  \m{\iaccessible[\pevolvein{\D{x}=\genDE{x}}{\ivr}]{\I}{\It}}
  implies $\iget[state]{\I}=\iget[state]{\Iff[0]}$ on $\scomplement{\{\D{x}\}}$ and $\iget[state]{\It}=\iget[state]{\Iff[r]}$ for some $\iget[flow]{\If}$ with
  \(\imodels{\If}{\D{x}=\genDE{x}\land\ivr}\),
  so \(\iget[state]{\Iff[\zeta]}=\iget[state]{\I}\) on $\scomplement{\{x,\D{x}\}}$ for all $\zeta$.
  Thus \(\iget[state]{\I}=\iget[state]{\It}\) on \(\scomplement{\{x,\D{x}\}}\) by \rref{def:HP-transition}.
  Hence, \(\boundvarsdef{\pevolvein{\D{x}=\genDE{x}}{\ivr}} \subseteq \{x,\D{x}\} = \boundvars{\pevolvein{\D{x}=\genDE{x}}{\ivr}}\).

\item \m{\iaccessible[\pchoice{\alpha}{\beta}]{\I}{\It} = \iaccess[\alpha]{\I} \cup \iaccess[\beta]{\I}}
implies \(\iaccessible[\alpha]{\I}{\It}\) or \(\iaccessible[\beta]{\I}{\It}\),
By induction hypothesis,
\(\boundvarsdef{\alpha}\subseteq\boundvars{\alpha}\) and
\(\boundvarsdef{\beta}\subseteq\boundvars{\beta}\).
Either way, \(\iget[state]{\I}=\iget[state]{\It}\) on $\scomplement{(\boundvarsdef{\alpha}\cup\boundvarsdef{\beta})}$.
So, \(\boundvarsdef{\pchoice{\alpha}{\beta}} \subseteq \boundvarsdef{\alpha}\cup\boundvarsdef{\beta} \subseteq  \boundvars{\alpha} \cup \boundvars{\beta} =\boundvars{\pchoice{\alpha}{\beta}}\).

\item
\newcommand{\Iz}{\dLint[state=\mu]}%
\m{\iaccessible[\alpha;\beta]{\I}{\It} = \iaccess[\alpha]{\I} \compose\iaccess[\beta]{\I}},
i.e.\ there is a $\iget[state]{\Iz}$ such that \(\iaccessible[\alpha]{\I}{\Iz}\) as well as \(\iaccessible[\beta]{\Iz}{\It}\).
By induction hypothesis, \(\boundvarsdef{\alpha}\subseteq\boundvars{\alpha}\) and
\(\boundvarsdef{\beta}\subseteq\boundvars{\beta}\).
Thus, \(\iget[state]{\I}=\iget[state]{\Iz}=\iget[state]{\It}\) on $\scomplement{(\boundvarsdef{\alpha}\cup\boundvarsdef{\beta})}$.
So \(\boundvarsdef{\alpha;\beta} \subseteq \boundvarsdef{\alpha} \cup \boundvarsdef{\beta} \subseteq \boundvars{\alpha} \cup \boundvars{\beta} = \boundvars{\alpha;\beta}\).

\item The case
\m{\iaccessible[\prepeat{\alpha}]{\I}{\It} = \displaystyle\cupfold_{n\in\naturals}\iaccess[{\prepeat[n]{\alpha}}]{\I}}
follows by induction on $n$.
\end{compactenum}

\noindent
The second part proves \(\freevars{\theta} \supseteq \freevarsdef{\theta}\).
{%
\renewcommand{\Ialt}{\vdLint[const=I,state=\tilde{\nu}]}%
\renewcommand{\Italt}{\vdLint[const=I,state=\tilde{\omega}]}%
Let $y\in\freevarsdef{\theta}$, which implies that there are
$\iget[const]{\I}$ and \(\iget[state]{\I}=\iget[state]{\Ialt}\) on $\scomplement{\{x\}}$ such that \(\ivaluation{\I}{\theta}\neq\ivaluation{\Ialt}{\theta}\).
The proof is a structural induction on $\theta$.
IH is short for induction hypothesis.
\begin{compactenum}
\item For any variable including $y$, \(\freevarsdef{y}=\freevars{y}\) since $y$ is the only variable that its value depends on.

\item \m{\ivaluation{\I}{f(\theta_1,\dots,\theta_k)} \neq \ivaluation{\Ialt}{f(\theta_1,\dots,\theta_k)}}
implies for some $i$ that 
\m{\ivaluation{\I}{\theta_i} \neq \ivaluation{\Ialt}{\theta_i}}.
Hence, $y\in\freevarsdef{\theta_i}$ since $y$ is the only variable that $\iget[state]{\I}$ and $\iget[state]{\Ialt}$ differ at.
By IH, this implies $y\in\freevars{\theta_i}\subseteq\freevars{f(\theta_1,\dots,\theta_k)}\).
This includes the case where $f$ is $+$ or $\cdot$ as well.

\item
\m{\ivaluation{\I}{\der{\theta}}
\displaystyle
= \sum_x \iget[state]{\I}(\D{x}) \itimes \Dp[x]{\ivaluation{\I}{\theta}}
\neq \sum_x \iget[state]{\Ialt}(\D{x}) \itimes \Dp[x]{\ivaluation{\Ialt}{\theta}}
= \ivaluation{\Ialt}{\der{\theta}}
}.
Since $\iget[state]{\I}$ and $\iget[state]{\Ialt}$ only differ in $y$, it is the case that
\begin{inparaenum}[\it i)]
\item \label{case:partialy}
for some $x$, \(\Dp[x]{\ivaluation{\I}{\theta}} \neq \Dp[x]{\ivaluation{\Ialt}{\theta}}\), or
\item \label{case:yisxprime}
$y$ is some $\D{x}$ and \(\iget[state]{\I}(\D{x})\neq\iget[state]{\Ialt}(\D{x})\) and
\(\Dp[x]{\ivaluation{\I}{\theta}}\neq0\) or \(\Dp[x]{\ivaluation{\Ialt}{\theta}}\neq0\).
\end{inparaenum}
In \rref{case:partialy}, there are states $\iget[state]{\It}=\iget[state]{\Italt}$ that agree on $\scomplement{\{y\}}$ such that \(\ivaluation{\It}{\theta} \neq \ivaluation{\Italt}{\theta}\) as otherwise their partial derivatives by $x$ would agree.
Thus, $y\in\freevarsdef{\theta} \overset{\text{IH}}{\subseteq} \freevars{\theta}\subseteq\freevars{\der{\theta}}$ by IH.
In \rref{case:yisxprime}, $x\in\freevarsdef{\theta}$ as, otherwise, \(\Dp[x]{\ivaluation{\I}{\theta}}=0\) for all states $\iget[state]{\I}$ since there would not be any state $\iget[state]{\It}$ agreeing \(\iget[state]{\It}=\iget[state]{\I}\) on $\scomplement{\{x\}}$ with a different value \(\ivaluation{\It}{\theta} \neq \ivaluation{\I}{\theta}\) then, so that their partial derivatives by $x$ would both be 0.
By IH, $x\in\freevarsdef{\theta}$ implies $x\in\freevars{\theta}$, which implies $\D{x}\in\freevars{\der{\theta}}$. 
\end{compactenum}
}%

\noindent
The proof of soundness of $\freevars{\phi}$ and $\freevars{\alpha}$ is indirect by a simultaneous inductive prove that both satisfy their respective coincidence property and are, thus, sound \(\freevars{\phi}\supseteq\freevarsdef{\phi}\) and \(\freevars{\alpha}\supseteq\freevarsdef{\alpha}\) by \rref{lem:coincidence} and~\ref{lem:coincidence-HP}, since $\freevarsdef{\phi}$ and $\freevarsdef{\alpha}$ are the smallest sets satisfying coincidence.
In fact, the inductive proof for programs shows a stronger coincidence property augmented with must-bound variables.

The proof that $\freevars{\phi}$ satisfies the coincidence property, and thus \(\freevars{\alpha}\supseteq\freevarsdef{\alpha}\), is by structural induction on $\phi$, simultaneously with the coincidence property for programs.
To simplify the proof, doubly negated existential quantifiers are considered structurally smaller than universal quantifiers and doubly negated diamond modalities smaller than box modalities.
\begin{enumerate}
\item \(\imodels{\I}{p(\theta_1,\dots,\theta_k)}\) iff \((\ivaluation{\I}{\theta_1},\dots,\ivaluation{\I}{\theta_k})\in\iget[const]{\I}(p)\)
iff \((\ivaluation{\Ialt}{\theta_1},\dots,\ivaluation{\Ialt}{\theta_k})\in\iget[const]{\Ialt}(p)\) iff \(\imodels{\Ialt}{p(\theta_1,\dots,\theta_k)}\)
by \rref{lem:coincidence-term} and \(\freevars{\theta}\supseteq\freevarsdef{\theta}\) since $\freevars{\theta_i}\subseteq\freevars{p(\theta_1,\dots,\theta_k)}$
and $\iget[const]{\I}$ and $\iget[const]{\Ialt}$ were assumed to agree on the function symbol $p$ that occurs in the formula.
This includes the case where $p$ is $\geq$ so that $\iget[const]{\I}$ and $\iget[const]{\Ialt}$ agree by definition.

\item \(\imodels{\I}{\contextapp{C}{\phi}} = \iget[const]{\I}(C)\big(\imodel{\I}{\phi}\big)\)
iff, by IH, \(\imodels{\Ialt}{\contextapp{C}{\phi}} = \iget[const]{\Ialt}(C)\big(\imodel{\Ialt}{\phi}\big)\)
since \(\iget[state]{\I}=\iget[state]{\Ialt}\) on \(\freevars{\contextapp{C}{\phi}}=\allvars\), so \(\iget[state]{\I}=\iget[state]{\Ialt}\),
and \(\iget[const]{\I}=\iget[const]{\Ialt}\) on \(\intsigns{\contextapp{C}{\phi}}=\{C\}\cup\intsigns{\phi}\), so \(\iget[const]{\I}(C)=\iget[const]{\Ialt}(C)\) and, by induction hypothesis, implies \(\imodel{\I}{\phi}=\imodel{\Ialt}{\phi}\)
using \(\iget[const]{\I}=\iget[const]{\Ialt}\) on \(\intsigns{\phi}\subseteq\intsigns{\contextapp{C}{\phi}}\).

\item \(\imodels{\I}{\lnot\phi}\) iff \(\inonmodels{\I}{\phi}\)
iff, by IH, \(\inonmodels{\Ialt}{\phi}\) iff \(\imodels{\Ialt}{\lnot\phi}\)
using $\freevars{\lnot\phi}=\freevars{\phi}$.

\item \(\imodels{\I}{\phi\land\psi}\) iff \(\imodels{\I}{\phi}\cap\imodel{\I}{\psi}\)
iff, by IH, \(\imodels{\Ialt}{\phi}\cap\imodel{\Ialt}{\psi}\) iff \(\imodels{\Ialt}{\phi\land\psi}\)
using $\freevars{\phi\land\psi}=\freevars{\phi}\cup\freevars{\psi}$.

\item
{\def\Im{\imodif[state]{\I}{x}{r}}%
\def\Imalt{\dLint[const=I,state=\tilde{\nu}_x^r]}%
\(\imodels{\I}{\lexists{x}{\phi}}\) iff \(\iget[state]{\Im} \in \imodel{\I}{\phi} ~\text{for some}~r\in\reals\)
iff \(\iget[state]{\Imalt} \in \imodel{\I}{\phi} ~\text{for some}~r\in\reals\) iff \(\imodels{\Ialt}{\lexists{x}{\phi}}\)
for the same $r$
by induction hypothesis using that \(\iget[state]{\Im}=\iget[state]{\Imalt}\) on $\freevars{\phi}\subseteq\{x\}\cup\freevars{\lexists{x}{\phi}}$.
}

\item The case \(\lforall{x}{\phi}\) follows from the equivalence \(\lforall{x}{\phi} \mequiv \lnot\lexists{x}{\lnot\phi}\) using \(\freevars{\lnot\lexists{x}{\lnot\phi}} = \freevars{\lforall{x}{\phi}}\).

\item \(\imodels{\I}{\ddiamond{\alpha}{\phi}}\) iff there is a $\iget[state]{\It}$ such that \(\iaccessible[\alpha]{\I}{\It}\) and \(\imodels{\It}{\phi}\).
Since \(\iget[state]{\I}=\iget[state]{\Ialt}\) on $\freevars{\ddiamond{\alpha}{\phi}}\supseteq\freevars{\alpha}$
and \(\iaccessible[\alpha]{\I}{\It}\),
the simultaneous induction hypothesis implies with \(\iget[const]{\I}=\iget[const]{\Ialt}\) on $\intsigns{\alpha}\subseteq\intsigns{\ddiamond{\alpha}{\phi}}$ that there is an $\iget[state]{\Italt}$ such that
  \(\iaccessible[\alpha]{\Ialt}{\Italt}\)
  and \(\iget[state]{\It}=\iget[state]{\Italt}\) on $\freevars{\ddiamond{\alpha}{\phi}}\cup\mustboundvars{\alpha}
  = \freevars{\alpha}\cup(\freevars{\phi}\setminus\mustboundvars{\alpha})\cup\mustboundvars{\alpha}
  = \freevars{\alpha}\cup\freevars{\phi}\cup\mustboundvars{\alpha}
  \supseteq \freevars{\phi}$.

\centerline{
\includegraphics[scale=0.8]{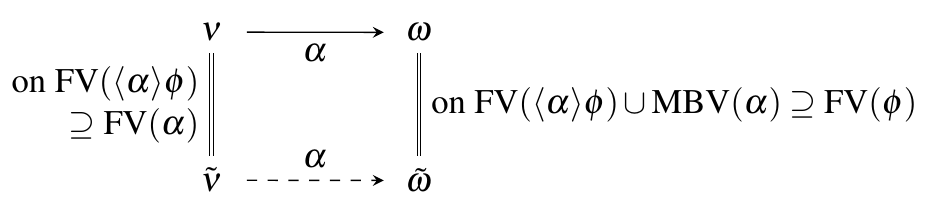}
}

Since, \(\iget[state]{\It}=\iget[state]{\Italt}\) on $\freevars{\phi}$ and \(\iget[const]{\I}=\iget[const]{\Ialt}\) on $\intsigns{\phi} \subseteq \intsigns{\ddiamond{\alpha}{\phi}}$,
the induction hypothesis implies that
\(\imodels{\Italt}{\phi}\) since \(\imodels{\It}{\phi}\).
Since \(\iaccessible[\alpha]{\Ialt}{\Italt}\), this implies
\m{\imodels{\Ialt}{\ddiamond{\alpha}{\phi}}}.

\item \(\imodels{\I}{\dbox{\alpha}{\phi}} = \imodel{\I}{\lnot\ddiamond{\alpha}{\lnot\phi}}\)
iff \(\inonmodels{\I}{\ddiamond{\alpha}{\lnot\phi}}\),
so by IH, iff \(\inonmodels{\Ialt}{\ddiamond{\alpha}{\lnot\phi}}\) iff \(\imodels{\Ialt}{\dbox{\alpha}{\phi}}\)
using that $\freevars{\ddiamond{\alpha}{\lnot\phi}}=\freevars{\dbox{\alpha}{\phi}}$.
\end{enumerate}

\noindent
The proof that $\freevars{\alpha}$ satisfies the coincidence property, and thus \(\freevars{\phi}\supseteq\freevarsdef{\phi}\), shows a stronger property.
  If \(\iget[state]{\I}=\iget[state]{\Ialt}\) on $V\supseteq\freevars{\alpha}$,
  \(\iget[const]{\I}=\iget[const]{\Ialt}\) on $\intsigns{\alpha}$
  and \(\iaccessible[\alpha]{\I}{\It}\),
  then there is a $\iget[state]{\Italt}$ such that
  \(\iaccessible[\alpha]{\Ialt}{\Italt}\)
  and 
  \(\iget[state]{\It}=\iget[state]{\Italt}\) on $V\cup\mustboundvars{\alpha}$.
The proof is by (simultaneous) induction on the structural complexity of $\alpha$, where $\prepeat{\alpha}$ is considered to be structurally more complex than \HPs of any length but with less nested repetitions, which induces a well-founded order on \HPs.
For atomic programs $\alpha$ for which \(\boundvars{\alpha}=\mustboundvars{\alpha}\), it is enough to show agreement on $\vars{\alpha}= \freevars{\alpha}\cup\boundvars{\alpha}=\freevars{\alpha}\cup\mustboundvars{\alpha}$, because any variable in $V\setminus\vars{\alpha}$ is in \(\scomplement{\boundvars{\alpha}}\), which remain unchanged by $\alpha$ according to \rref{lem:bound} and \(\boundvars{\alpha}\supseteq\boundvarsdef{\alpha}\).

\begin{compactenum}

\item Since \(\freevars{a} = \allvars\) so $\iget[state]{\I}=\iget[state]{\Ialt}$, the statement is vacuously true for program constant $a$.

\item \m{\iaccessible[\pupdate{\pumod{x}{\theta}}]{\I}{\It} = \{(\iget[state]{\I},\iget[state]{\It}) \with \iget[state]{\It}=\iget[state]{\I}~\text{except that}~\ivaluation{\It}{x}=\ivaluation{\I}{\theta}\}}
then there is \(\iaccessible[\pupdate{\pumod{x}{\theta}}]{\Ialt}{\Italt}\)
and \(\iget[state]{\Italt}(x)=\ivaluation{\Italt}{x}=\ivaluation{\Ialt}{\theta}
=\ivaluation{\I}{\theta}=\ivaluation{\It}{x} = \iget[state]{\I}(x)\)
by \rref{lem:coincidence-term},
since \(\iget[state]{\I}=\iget[state]{\Ialt}\) on $\freevars{\pupdate{\pumod{x}{\theta}}}=\freevars{\theta}$ and \(\iget[const]{\I}=\iget[const]{\Ialt}\) on $\intsigns{\theta}$.
So, \(\iget[state]{\It}=\iget[state]{\Italt}\) on $\boundvars{\pupdate{\pumod{x}{\theta}}}=\{x\}$.
Also, \(\iget[state]{\I}=\iget[state]{\It}\) on $\scomplement{\boundvars{\pupdate{\pumod{x}{\theta}}}}$
and \(\iget[state]{\Ialt}=\iget[state]{\Italt}\) on $\scomplement{\boundvars{\pupdate{\pumod{x}{\theta}}}}$ by \rref{lem:bound}.
Since \(\iget[state]{\I}=\iget[state]{\Ialt}\) on $\freevars{\pupdate{\pumod{x}{\theta}}}$, these imply
\(\iget[state]{\It}=\iget[state]{\Italt}\) on $\freevars{\pupdate{\pumod{x}{\theta}}}\setminus\boundvars{\pupdate{\pumod{x}{\theta}}}$.
Since \(\iget[state]{\It}=\iget[state]{\Italt}\) on $\boundvars{\pupdate{\pumod{x}{\theta}}}$ had been shown already, this implies
\(\iget[state]{\It}=\iget[state]{\Italt}\) on $\vars{\pupdate{\pumod{x}{\theta}}}$.

\item \m{\iaccessible[\ptest{\ivr}]{\I}{\It} = \{(\iget[state]{\I},\iget[state]{\I}) \with \imodels{\I}{\ivr}\}}
then \(\iget[state]{\It}=\iget[state]{\I}\) by \rref{def:HP-transition}.
Since, \(\imodels{\I}{\ivr}\) and \(\iget[state]{\I}=\iget[state]{\Ialt}\) on $\freevars{\ptest{\ivr}}$ and \(\iget[const]{\I}=\iget[const]{\Ialt}\) on $\intsigns{\ivr}=\intsigns{\ptest{\ivr}}$,
\rref{lem:coincidence} and the simultaneous induction for \(\freevars{\psi}\supseteq\freevarsdef{\psi}\) implies that \(\imodels{\Ialt}{\ivr}\), so \(\iaccessible[\ptest{\ivr}]{\Ialt}{\Ialt}\).
So \(\iget[state]{\I}=\iget[state]{\Ialt}\) on $\vars{\ptest{\ivr}}=\freevars{\ptest{\ivr}}$ since \(\boundvars{\ptest{\ivr}}=\emptyset\).

\item
\newcommand{\Ifalt}{\DALint[const=J,flow=\tilde{\rule{0pt}{5pt}\varphi}]}%
\newcommand*{\Iffalt}[1][\zeta]{\vdLint[const=J,state=\tilde{\rule{0pt}{5pt}\varphi}(#1)]}%
  \m{\iaccessible[\pevolvein{\D{x}=\genDE{x}}{\ivr}]{\I}{\It}} 
  implies that there is an $\iget[state]{\Italt}$ reached from $\iget[state]{\Ialt}$ by following the differential equation for the same amount it took to reach $\iget[state]{\It}$ from $\iget[state]{\I}$.
  That is, $\iget[state]{\I}=\iget[state]{\Iff[0]}$ on $\scomplement{\{\D{x}\}}$ and \(\iget[state]{\It}=\iget[state]{\Iff[r]}\)
  for some function \m{\iget[flow]{\If}:[0,r]\to\linterpretations{\Sigma}{V}}
  satisfying \m{\imodels{\If}{\D{x}=\genDE{x}\land\ivr}},
  especially
  \(\imodels{\Iff[\zeta]}{\D{x}=\genDE{x}\land\ivr}\) for all $0\leq\zeta\leq r$.
  Define \m{\iget[flow]{\Ifalt}:[0,r]\to\linterpretations{\Sigma}{V}}
  at $\zeta$ as \(\iget[state]{\Iffalt[\zeta]}=\iget[state]{\Iff[\zeta]}\) on $\{x,\D{x}\}$
  and as \(\iget[state]{\Iffalt[\zeta]}=\iget[state]{\Ialt}\) on $\scomplement{\{x,\D{x}\}}$.
  Fix any $0\leq\zeta\leq r$.
  Then it only remains to show that
  \(\imodels{\Iffalt[\zeta]}{\D{x}=\genDE{x}\land\ivr}\),
  i.e.\
  \(\imodels{\Iffalt[\zeta]}{\D{x}=\genDE{x}}\cap\imodel{\Iffalt[\zeta]}{\ivr}\),
  which follows from
  \(\imodels{\Iff[\zeta]}{\D{x}=\genDE{x}}\cap\imodel{\Iff[\zeta]}{\ivr}\)
  by \rref{lem:coincidence} and the simultaneous induction hypothesis,
  since \(\iget[const]{\Iff[\zeta]}=\iget[const]{\Iffalt[\zeta]}\)
  on $\intsigns{\D{x}=\genDE{x}}\cup\intsigns{\ivr} = \intsigns{\pevolvein{\D{x}=\genDE{x}}{\ivr}}$
  and \(\iget[state]{\Iff[\zeta]}=\iget[state]{\Iffalt[\zeta]}\)
  on $\freevars{\D{x}=\genDE{x}}\cup\freevars{\ivr} = \freevars{\pevolvein{\D{x}=\genDE{x}}{\ivr}} \cup \{\D{x}\}$.
  Here, \(\iget[state]{\Iff[\zeta]}=\iget[state]{\Iffalt[\zeta]}\) agree on $\{x,\D{x}\}$ by construction of $\iget[flow]{\Ifalt}$.
  Agreement of $\iget[state]{\Iff[\zeta]}$ and $ \iget[state]{\Iffalt[\zeta]}$ for the other free variables follows from the assumption that $\iget[state]{\I}=\iget[state]{\Ialt}$ on $\freevars{\pevolvein{\D{x}=\genDE{x}}{\ivr}}$
  since \(\iget[state]{\I}=\iget[state]{\Iff[\zeta]}\) on $\scomplement{\{x,\D{x}\}}$ by \rref{def:HP-transition}
  and since \(\iget[state]{\Ialt}=\iget[state]{\Iffalt[\zeta]}\) on $\scomplement{\{x,\D{x}\}}$ by construction.
    
\item \m{\iaccessible[\pchoice{\alpha}{\beta}]{\I}{\It} = \iaccess[\alpha]{\I} \cup \iaccess[\beta]{\I}}
implies \(\iaccessible[\alpha]{\I}{\It}\) or \(\iaccessible[\beta]{\I}{\It}\),
which since \(V\supseteq\freevars{\pchoice{\alpha}{\beta}}\supseteq\freevars{\alpha}\)
and \(V\supseteq\freevars{\pchoice{\alpha}{\beta}}\supseteq\freevars{\beta}\) implies, by induction hypothesis,
that there is an $\iget[state]{\Italt}$ such that
  \(\iaccessible[\alpha]{\Ialt}{\Italt}\)
  and \(\iget[state]{\It}=\iget[state]{\Italt}\) on $V\cup\mustboundvars{\alpha}$
or that there is an $\iget[state]{\Italt}$ such that
  \(\iaccessible[\beta]{\Ialt}{\Italt}\)
  and \(\iget[state]{\It}=\iget[state]{\Italt}\) on $V\cup\mustboundvars{\beta}$, respectively.
  In either case, there is a $\iget[state]{\Italt}$ such that
  \(\iaccessible[\pchoice{\alpha}{\beta}]{\Ialt}{\Italt}\)
  and \(\iget[state]{\It}=\iget[state]{\Italt}\) on $V\cup\mustboundvars{\pchoice{\alpha}{\beta}}$,
  because \(\iaccess[\alpha]{\Ialt}\subseteq\iaccess[\pchoice{\alpha}{\beta}]{\Ialt}\)
  and \(\iaccess[\beta]{\Ialt}\subseteq\iaccess[\pchoice{\alpha}{\beta}]{\Ialt}\) 
  and \(\mustboundvars{\pchoice{\alpha}{\beta}}=\mustboundvars{\alpha}\cap\mustboundvars{\beta}\).

\centerline{
\includegraphics[scale=0.8]{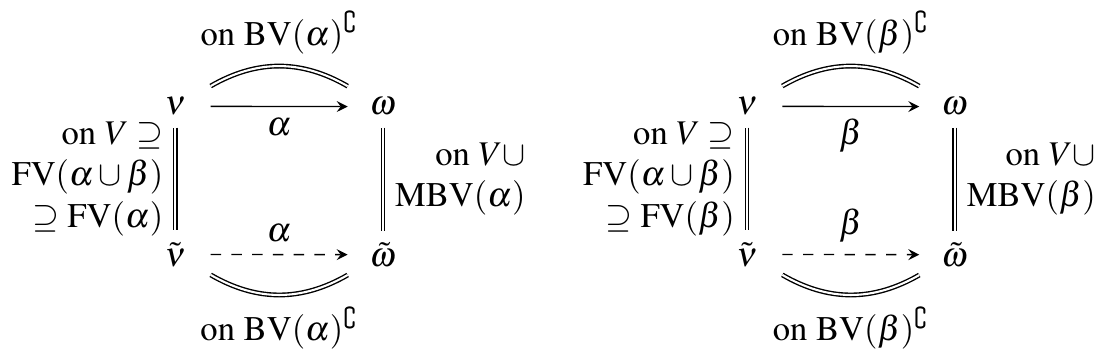}
}

\item \label{case:coincidence-compose}
\newcommand{\Iz}{\iconcat[state=\mu]{\I}}%
\newcommand{\Izalt}{\vdLint[const=J,state=\tilde{\mu}]}%
\m{\iaccessible[\alpha;\beta]{\I}{\It} = \iaccess[\alpha]{\I} \compose \iaccess[\beta]{\I}},
i.e.\ there is a state $\iget[state]{\Iz}$ such that \(\iaccessible[\alpha]{\I}{\Iz}\) and \(\iaccessible[\beta]{\Iz}{\It}\).
Since \(V\supseteq\freevars{\alpha;\beta}\supseteq\freevars{\alpha}\), by induction hypothesis, there is a $\iget[state]{\Izalt}$ such that
  \(\iaccessible[\alpha]{\Ialt}{\Izalt}\)
  and \(\iget[state]{\Iz}=\iget[state]{\Izalt}\) on $V\cup\mustboundvars{\alpha}$.
  Since \(V\supseteq\freevars{\alpha;\beta}\),
  so
  \(V\cup\mustboundvars{\alpha} \supseteq \freevars{\alpha;\beta} \cup \mustboundvars{\alpha}
  = \freevars{\alpha}\cup(\freevars{\beta}\setminus\mustboundvars{\alpha}) \cup \mustboundvars{\alpha}
  = \freevars{\alpha}\cup\freevars{\beta}\cup\mustboundvars{\alpha}
  \supseteq \freevars{\beta}\)
  by \rref{def:freevars},
  and since \(\iaccessible[\beta]{\Iz}{\It}\), the induction hypothesis implies that
  there is an $\iget[state]{\Italt}$ such that
  \(\iaccessible[\beta]{\Izalt}{\Italt}\)
  and \(\iget[state]{\It}=\iget[state]{\Italt}\) on $(V\cup\mustboundvars{\alpha})\cup\mustboundvars{\beta} = V\cup\mustboundvars{\alpha;\beta}$.

  \centerline{
  \includegraphics[scale=0.8]{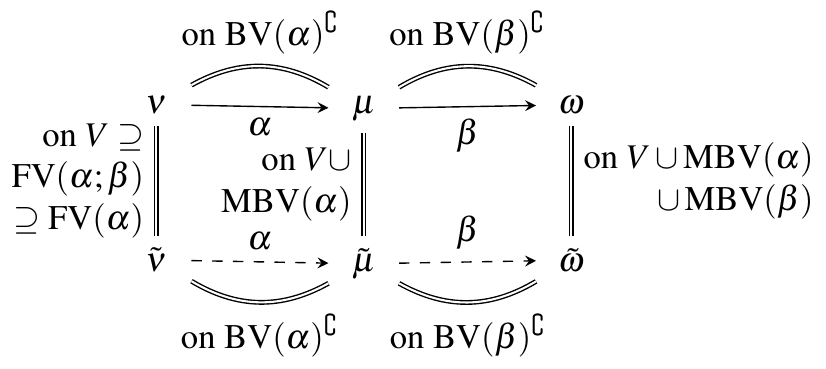}
}

\item
\renewcommand{\Iz}[1][]{\dLint[state=\nu_{#1}]}%
\m{\iaccessible[\prepeat{\alpha}]{\I}{\It} = \displaystyle\cupfold_{n\in\naturals}\iaccess[{\prepeat[n]{\alpha}}]{\I}}
iff there is an $n\in\naturals$ such that \(\iaccessible[\alpha^n]{\I}{\It}\).
The case $n=0$ follows from the assumption \(\iget[state]{\I}=\iget[state]{\Ialt}\) on $V\supseteq\freevars{\alpha}$, since \(\iget[state]{\It}=\iget[state]{\I}\) holds in that case and $\mustboundvars{\prepeat{\alpha}}=\emptyset$.
The case $n>0$ proceeds as follows.
Since \(\freevars{\prepeat[n]{\alpha}}=\freevars{\prepeat{\alpha}}=\freevars{\alpha}\), the induction hypothesis applied to the structurally simpler \HP $\prepeat[n]{\alpha} \mequiv \prepeat[n-1]{\alpha};\alpha$ with less loops (so using \rref{case:coincidence-compose}) implies
that there is an $\iget[state]{\Italt}$ such that
  \(\iaccessible[\alpha^n]{\Ialt}{\Italt}\)
  and \(\iget[state]{\It}=\iget[state]{\Italt}\) on $V\cup\mustboundvars{\prepeat[n]{\alpha}} \supseteq V = V\cup\mustboundvars{\prepeat{\alpha}}$,
  since $\mustboundvars{\prepeat{\alpha}}=\emptyset$.
  Since \(\iaccess[{\prepeat[n]{\alpha}}]{\Ialt}\subseteq\iaccess[\prepeat{\alpha}]{\Ialt}\) by \rref{def:HP-transition}, this concludes the proof.
\qedhere
\end{compactenum}
\end{proofatend}
\noindent
In particular, the final states $\iget[state]{\It}$ and $\iget[state]{\Italt}$ agree on $\vars{\alpha}$ if the initial states $\iget[state]{\I}$ and $\iget[state]{\Ialt}$ agree on $\vars{\alpha}$ and even if the initial states only agree on $\vars{\alpha}\setminus\mustboundvars{\alpha}$.

This concludes the static semantics of \dL, which computes syntactically what kind of state change formulas $\phi$ and \HPs $\alpha$ may cause (captured in $\boundvars{\phi},\boundvars{\alpha}$) and what part of the state their values and behavior depends on ($\freevars{\phi},\freevars{\alpha}$).
\rref{lem:sound-static-semantics} will be used implicitly in the sequel when referring to Lemmas~\ref{lem:bound}--\ref{lem:coincidence-HP}.

\section{Uniform Substitutions} \label{sec:usubst}

The uniform substitution rule \irref{US0} from first-order logic \cite[\S35,40]{Church_1956} substitutes \emph{all} occurrences of predicate $p(\usarg)$ by a formula $\mapply{\psi}{\usarg}$, i.e.\ it replaces all occurrences of $p(\theta)$, for any (vectorial) argument term $\theta$, by the corresponding $\mapply{\psi}{\theta}$ simultaneously:
\[
      \cinferenceRule[US0|US$_1$]{uniform substitution}
      {\linferenceRule[formula]
        {\preusubst[\phi]{p}}
        {\usubst[\phi]{p}{\psi}}
      }{}%
      \qquad\qquad
      \cinferenceRule[US|US]{uniform substitution}
      {\linferenceRule[formula]
        {\phi}
        {\applyusubst{\sigma}{\phi}}
      }{}%
\]
Soundness of rule \irref{US0} \cite{DBLP:journals/tocl/Platzer15} requires all relevant substitutions of $\mapply{\psi}{\theta}$ for $p(\theta)$ to be \emph{admissible}, i.e.\ that no $p(\theta)$ occurs in the scope of a quantifier or modality binding a variable of $\mapply{\psi}{\theta}$ other than the occurrences in $\theta$; see \cite[\S35,40]{Church_1956}.
A precise definition of admissibility is the key ingredient and will be developed from the static semantics.

This section develops rule \irref{US} as a more general and constructive definition with a precise substitution algorithm and precise admissibility conditions that allow symbols from more syntactic categories to be substituted.
The \dL calculus uses uniform substitutions that affect terms, formulas, and programs.
A \dfn{uniform substitution} $\sigma$ is a mapping
from expressions of the
form \(f(\usarg)\) to terms $\applysubst{\sigma}{f(\usarg)}$,
from \(p(\usarg)\) to formulas $\applysubst{\sigma}{p(\usarg)}$,
from \(\contextapp{C}{\uscarg}\) to formulas $\applysubst{\sigma}{\contextapp{C}{\uscarg}}$,
and from program constants \(a\) to \HPs $\applysubst{\sigma}{a}$.
Vectorial extensions are accordingly for uniform substitutions of other arities $k\geq0$.
Here $\usarg$ is a reserved function symbol of arity zero and $\uscarg$ a reserved quantifier symbol of arity zero, which mark the positions where the respective argument, e.g., argument $\theta$ to $p(\usarg)$ in the formula $p(\theta)$, will end up in the replacement $\applysubst{\sigma}{p(\usarg)}$ used for $p(\theta)$.

\begin{example}[Uniform substitutions with or without clashes] \label{ex:usubst-assign1}
The uniform substitution
\(\sigma=\usubstlist{\usubstmod{f}{x+1},\usubstmod{p(\usarg)}{(\usarg\neq x)}}\)
substitutes all occurrences of function symbol $f$ (of arity 0) by $x+1$ and simultaneously substitutes all occurrences of $p(\theta)$ with predicate symbol $p$ of any argument $\theta$ by the corresponding $(\theta\neq x)$.
Whether that uniform substitution is sound depends on admissibility of $\sigma$ for the formula $\phi$ in \irref{US} as will be defined in \rref{def:usubst-admissible}.
It will turn out to be admissible (and thus sound) for
\[
\linfer[US]
{\dbox{\pupdate{\umod{y}{f}}}{p(2y)} \lbisubjunct \dbox{\pupdate{\umod{y}{f}}}{p(2f)}}
{\dbox{\pupdate{\umod{y}{x+1}}}{\,2y\neq x} \lbisubjunct \dbox{\pupdate{\umod{y}{x+1}}}{\,2(x+1)\neq x} }
\qquad
\sigma=\usubstlist{\usubstmod{f}{x+1},\usubstmod{p(\usarg)}{(\usarg\neq x)}}
\]
but will turn out to be in-admissible (and, in fact, would be unsound) for:
\[
\newcommand*{\dangerocc}[1]{\textcolor{vred}{#1}}
\linfer[clash]
{\dbox{\pupdate{\umod{x}{f}}}{p(x)} \lbisubjunct p(f)}
{\dbox{\pupdate{\umod{\dangerocc{x}}{x+1}}}{\,x\neq x} \lbisubjunct x+1\neq x}
\qquad
\sigma=\usubstlist{\usubstmod{f}{x+1},\usubstmod{p(\usarg)}{(\usarg\neq \dangerocc{x})}}
\]
Here, $\sigma$ is not admissible, because $\sigma$ has a free variable $x$ in its replacement for $p(\usarg)$ that it introduces into a context where $x$ is bound by the modality $\dbox{\pupdate{\umod{x}{\dots}}}{}$, so the $x$ in replacement $\usarg\neq x$ for $p(\usarg)$ would refer to different values in the occurrences of $p$.
\end{example}

Figure~\ref{fig:usubst} defines the result $\applyusubst{\sigma}{\phi}$ of applying to a \dL formula~$\phi$ the \dfn{uniform substitution} $\sigma$ that uniformly replaces all occurrences of a function~$f$ by a term (instantiated with its respective argument of $f$) and all occurrences of a predicate~$p$ or a quantifier~$C$ symbol by a formula (instantiated with its argument)
as well as of a program constant $a$ by a program.
A uniform substitution can replace any number of such function, predicate, and quantifier symbols or program constants simultaneously.
The notation $\applysubst{\sigma}{f(\usarg)}$ denotes the replacement for $f(\usarg)$ according to $\sigma$, i.e.\ the value $\applysubst{\sigma}{f(\usarg)}$ of function $\sigma$ at $f(\usarg)$.
By contrast, $\applyusubst{\sigma}{\phi}$ denotes the result of applying $\sigma$ to $\phi$ according to \rref{fig:usubst} (likewise for $\applyusubst{\sigma}{\theta}$ and $\applyusubst{\sigma}{\alpha}$).
The notation $f\in\replacees{\sigma}$ signifies that $\sigma$ replaces $f$, i.e.\ \(\applysubst{\sigma}{f(\usarg)} \neq f(\usarg)\).
Finally, $\sigma$ is a total function when augmented with \(\applysubst{\sigma}{g(\usarg)}=g(\usarg)\) for all $g\not\in\replacees{\sigma}$, so that the case $g\not\in\replacees{\sigma}$ in \rref{fig:usubst} is subsumed by case $f\in\replacees{\sigma}$.
Corresponding notation is used for predicate symbols, quantifier symbols, and program constants.
The cases $g\not\in\replacees{\sigma}$, $p\not\in\replacees{\sigma}$, $C\not\in\replacees{\sigma}$, $b\not\in\replacees{\sigma}$ follow from the other cases but are listed explicitly for clarity.
Arguments are put in for the placeholder $\usarg$ recursively by uniform substitution \(\{\usarg\mapsto\applyusubst{\sigma}{\theta}\}\) in \rref{fig:usubst}, which is defined since it replaces the function symbol $\usarg$ of arity 0 by $\applyusubst{\sigma}{\theta}$, or accordingly for quantifier symbol $\uscarg$ of arity 0.

\begin{figure}[tb]
  \newcommand*{\implied}[1]{\textcolor{darkgray}{#1}}%
  \begin{displaymath}
    \begin{array}{@{}rcll@{}}
    \applyusubst{\sigma}{x} &=& x & \text{for variable $x\in\allvars$}\\
    \applyusubst{\sigma}{f(\theta)} &=& (\applyusubst{\sigma}{f})(\applyusubst{\sigma}{\theta})
  \mdefeq \applyusubst{\{\usarg\mapsto\applyusubst{\sigma}{\theta}\}}{\applysubst{\sigma}{f(\usarg)}} &
  \text{for function symbol}~f\in\replacees{\sigma}
    \\
  \implied{\applyusubst{\sigma}{g(\theta)}} &\implied{=}& \implied{g(\applyusubst{\sigma}{\theta})} &\implied{\text{for function symbol}~g\not\in\replacees{\sigma}}
  \\
  \applyusubst{\sigma}{\theta+\eta} &=& \applyusubst{\sigma}{\theta} + \applyusubst{\sigma}{\eta}
  \\
  \applyusubst{\sigma}{\theta\cdot\eta} &=& \applyusubst{\sigma}{\theta} \cdot \applyusubst{\sigma}{\eta}
  \\
  \applyusubst{\sigma}{\der{\theta}} &=& \der{\applyusubst{\sigma}{\theta}} &\text{if $\sigma$ is $\allvars$-admissible for $\theta$}
  \\
  \hline
  \applyusubst{\sigma}{\theta\geq\eta} &\mequiv& \applyusubst{\sigma}{\theta} \geq \applyusubst{\sigma}{\eta}\\
    \applyusubst{\sigma}{p(\theta)} &\mequiv& (\applyusubst{\sigma}{p})(\applyusubst{\sigma}{\theta})
  \mdefequiv \applyusubst{\{\usarg\mapsto\applyusubst{\sigma}{\theta}\}}{\applysubst{\sigma}{p(\usarg)}} &
  \text{for predicate symbol}~p\in\replacees{\sigma}\\
  \implied{\applyusubst{\sigma}{q(\theta)}} &\implied{\mequiv}& \implied{q(\applyusubst{\sigma}{\theta})} &\implied{\text{for predicate symbol}~q\not\in\replacees{\sigma}}\\
  \applyusubst{\sigma}{\contextapp{C}{\phi}} &\mequiv& \contextapp{\applyusubst{\sigma}{C}}{\applyusubst{\sigma}{\phi}}
  \mdefequiv \applyusubst{\{\uscarg\mapsto\applyusubst{\sigma}{\phi}\}}{\applysubst{\sigma}{\contextapp{C}{\uscarg}}} &
  \text{if $\sigma$ is $\allvars$-admissible for $\phi$, $C\in\replacees{\sigma}$}
  \\
  \implied{\applyusubst{\sigma}{\contextapp{C}{\phi}}} &\implied{\mequiv}& \implied{\contextapp{C}{\applyusubst{\sigma}{\phi}}} &
  \implied{\text{if $\sigma$ is $\allvars$-admissible for $\phi$, $C\not\in\replacees{\sigma}$}}
  \\
    \applyusubst{\sigma}{\lnot\phi} &\mequiv& \lnot\applyusubst{\sigma}{\phi}\\
    \applyusubst{\sigma}{\phi\land\psi} &\mequiv& \applyusubst{\sigma}{\phi} \land \applyusubst{\sigma}{\psi}\\
    \applyusubst{\sigma}{\lforall{x}{\phi}} &=& \lforall{x}{\applyusubst{\sigma}{\phi}} & \text{if $\sigma$ is $\{x\}$-admissible for $\phi$}\\
    \applyusubst{\sigma}{\lexists{x}{\phi}} &=& \lexists{x}{\applyusubst{\sigma}{\phi}} & \text{if $\sigma$ is $\{x\}$-admissible for $\phi$}\\
    \applyusubst{\sigma}{\dbox{\alpha}{\phi}} &=& \dbox{\applyusubst{\sigma}{\alpha}}{\applyusubst{\sigma}{\phi}} & \text{if $\sigma$ is $\boundvarsdef{\applyusubst{\sigma}{\alpha}}$-admissible for $\phi$}\\
    \applyusubst{\sigma}{\ddiamond{\alpha}{\phi}} &=& \ddiamond{\applyusubst{\sigma}{\alpha}}{\applyusubst{\sigma}{\phi}} & \text{if $\sigma$ is $\boundvarsdef{\applyusubst{\sigma}{\alpha}}$-admissible for $\phi$}
    \\
  \hline
    \applyusubst{\sigma}{a} &\mequiv& \applysubst{\sigma}{a} &\text{for program constant $a\in\replacees{\sigma}$}\\
    \implied{\applyusubst{\sigma}{b}} &\implied{\mequiv}& \implied{b} &\implied{\text{for program constant $b\not\in\replacees{\sigma}$}}\\
    \applyusubst{\sigma}{\pupdate{\umod{x}{\theta}}} &\mequiv& \pupdate{\umod{x}{\applyusubst{\sigma}{\theta}}}\\
    \applyusubst{\sigma}{\pevolvein{\D{x}=\genDE{x}}{\ivr}} &\mequiv&
    \hevolvein{\D{x}=\applyusubst{\sigma}{\genDE{x}}}{\applyusubst{\sigma}{\ivr}} & \text{if $\sigma$ is $\{x,\D{x}\}$-admissible for $\genDE{x},\ivr$}\\
    \applyusubst{\sigma}{\ptest{\ivr}} &\mequiv& \ptest{\applyusubst{\sigma}{\ivr}}\\
    \applyusubst{\sigma}{\pchoice{\alpha}{\beta}} &\mequiv& \pchoice{\applyusubst{\sigma}{\alpha}} {\applyusubst{\sigma}{\beta}}\\
    \applyusubst{\sigma}{\alpha;\beta} &\mequiv& \applyusubst{\sigma}{\alpha}; \applyusubst{\sigma}{\beta} &\text{if $\sigma$ is $\boundvarsdef{\applyusubst{\sigma}{\alpha}}$-admissible for $\beta$}\\
    \applyusubst{\sigma}{\prepeat{\alpha}} &\mequiv& \prepeat{(\applyusubst{\sigma}{\alpha})} &\text{if $\sigma$ is $\boundvarsdef{\applyusubst{\sigma}{\alpha}}$-admissible for $\alpha$}
    \end{array}%
  \end{displaymath}%
  \caption{Recursive application of uniform substitution~$\sigma$}%
  \index{substitution!uniform|textbf}%
  \label{fig:usubst}
\end{figure}%
\begin{definition}[Admissible uniform substitution] \label{def:usubst-admissible}
  \index{admissible|see{substitution, admissible}}
  A uniform substitution~$\sigma$ is \emph{$U$-\admissible for $\phi$} (or $\theta$ or $\alpha$, respectively) with respect to the variables $U\subseteq\allvars$ iff
  \(\freevarsdef{\restrict{\sigma}{\intsigns{\phi}}}\cap U=\emptyset\),
  where \({\restrict{\sigma}{\intsigns{\phi}}}\) is the restriction of $\sigma$ that only replaces symbols that occur in $\phi$,
  and
  \(\freevarsdef{\sigma}=\cupfold_{f\in\replacees{\sigma}} \freevarsdef{\applysubst{\sigma}{f(\usarg)}} \cup \cupfold_{p\in\replacees{\sigma}} \freevarsdef{\applysubst{\sigma}{p(\usarg)}}\)
  are the \emph{free variables} that $\sigma$ introduces. 
  A uniform substitution~$\sigma$ is \emph{\admissible for $\phi$} (or $\theta$ or $\alpha$, respectively) iff the bound variables $U$ of each operator of $\phi$ are not free in the substitution on its arguments, i.e.\ $\sigma$ is $U$-admissible.
  These admissibility conditions are listed explicitly in \rref{fig:usubst}, which defines the result $\applyusubst{\sigma}{\phi}$ of applying $\sigma$ to $\phi$.
\end{definition}

The substitution $\sigma$ is said to \emph{clash} and its result $\applyusubst{\sigma}{\phi}$ (or $\applyusubst{\sigma}{\theta}$ or $\applyusubst{\sigma}{\alpha}$) is not defined if $\sigma$ is not admissible, in which case rule \irref{US} is not applicable either.
All subsequent applications of uniform substitutions are required to be defined (no clash).
If a uniform substitution is admissible using the syntactic $\freevars{\phi}$ and $\boundvars{\alpha}$ from \rref{sec:static-semantics-computation}, then it is also admissible using the static semantics $\freevarsdef{\phi}$ and $\boundvarsdef{\alpha}$ from \rref{sec:static-semantics} by \rref{lem:sound-static-semantics}, so that the syntactic computations can be used soundly.

\begin{example}[Admissibility] 
The first use of \irref{US} in \rref{ex:usubst-assign1} is admissible, because no free variable of the substitution is introduced into a context in which that variable is bound.
The second, unsound attempt in \rref{ex:usubst-assign1} clashes, because it is not admissible, since \(x\in\freevars{\sigma}\) but also \(x\in\boundvars{\pupdate{\pumod{x}{x+1}}}\).
Occurrences of such bound variables that result from the arguments of the predicates or functions are exempt:
\[
\linfer[US]
{\dbox{\pupdate{\umod{x}{f}}}{p(x)} \lbisubjunct p(f)}
{\dbox{\pupdate{\umod{x}{x+1}}}{x\neq y} \lbisubjunct x+1\neq y}
\qquad
\sigma=\usubstlist{\usubstmod{f}{x+1},\usubstmod{p(\usarg)}{(\usarg\neq y)}}
\]
\end{example}

\subsection{Uniform Substitution Lemmas}

Soundness of rule \irref{US} requires proving that validity is preserved when replacing symbols with their uniform substitutes.
The key to its soundness proof is to relate this syntactic change to a semantic change of the interpretations such that validity of its premise in all interpretations implies validity of the premise in the semantically modified interpretation, which is then equivalent to validity of its syntactical substitute in the conclusion.
The semantic substitution corresponding to (or adjoint to) $\sigma$ modifies the interpretation of function, predicate and quantifier symbols as well as program constants semantically in the same way that $\sigma$ replaces them syntactically.
When $\sigma$ is admissible, the value of an expression in the adjoint interpretation agrees with the value of its uniform substitute in the original interpretation.
This link to the static semantics proves the following correspondence of syntactic and  semantic substitution.

Let \m{\iget[const]{\imodif[const]{\I}{\,\usarg}{d}}} denote the interpretation that agrees with interpretation~$\iget[const]{\I}$ except for the interpretation of function symbol~$\usarg$ which is changed to~\m{d\in\reals}.
Correspondingly \m{\iget[const]{\imodif[const]{\I}{\uscarg}{R}}} denotes the interpretation that agrees with $\iget[const]{\I}$ except that quantifier symbol $\uscarg$ is $R\subseteq\linterpretations{\Sigma}{V}$.

\begin{definition}[Substitution adjoints] \label{def:adjointUsubst}
\def\Ialta{\iadjointSubst{\sigma}{\Ialt}}%
The \emph{adjoint} to substitution $\sigma$ is the operation that maps $\iportray{\I}$ to the \emph{adjoint} interpretation $\iget[const]{\Ia}$ in which the interpretation of each function symbol $f\in\replacees{\sigma}$, predicate symbol $p\in\replacees{\sigma}$, quantifier symbol $C\in\replacees{\sigma}$, and program constant $a\in\replacees{\sigma}$ is modified according to $\sigma$:
\begin{align*}
  \iget[const]{\Ia}(f) &: \reals\to\reals;\, d\mapsto\ivaluation{\imodif[const]{\I}{\,\usarg}{d}}{\applysubst{\sigma}{f}(\usarg)}\\
  \iget[const]{\Ia}(p) &= \{d\in\reals \with \imodels{\imodif[const]{\I}{\,\usarg}{d}}{\applysubst{\sigma}{p}(\usarg)}\}
  \\
  \iget[const]{\Ia}(C) &: \powerset{\linterpretations{\Sigma}{V}}\to\powerset{\linterpretations{\Sigma}{V}};\, R\mapsto\imodel{\imodif[const]{\I}{\,\uscarg}{R}}{\applysubst{\sigma}{\contextapp{C}{\uscarg}}}
  \\
  \iget[const]{\Ia}(a) &= \iaccess[\applysubst{\sigma}{a}]{\I}
\end{align*}
\end{definition}
\begin{corollary}[Admissible adjoints] \label{cor:adjointUsubst}
\def\Ialta{\iadjointSubst{\sigma}{\Ialt}}%
If \(\iget[state]{\I}=\iget[state]{\It}\) on \(\freevarsdef{\sigma}\),
then \(\iget[const]{\Ia}=\iget[const]{\Ita}\).
If $\sigma$ is $U$-admissible for $\theta$ (or $\phi$ or $\alpha$, respectively) and \(\iget[state]{\I}=\iget[state]{\It}\) on $\scomplement{U}$, then
\begin{align*}
  \ivalues{\Ia}{\theta} &= \ivalues{\Ita}{\theta}
  ~\text{i.e.}~
  \ivaluation{\iconcat[state=\mu]{\Ia}}{\theta} = \ivaluation{\iconcat[state=\mu]{\Ita}}{\theta} ~\text{for all states}~\mu
  \\
  \imodel{\Ia}{\phi} &= \imodel{\Ita}{\phi}\\
  \iaccess[\alpha]{\Ia} &= \iaccess[\alpha]{\Ita}
\end{align*}
\end{corollary}
\begin{proofatend}
$\iget[const]{\Ia}$ is well-defined, as $\iget[const]{\Ia}(f)$ is a smooth function since its substitute term \({\applysubst{\sigma}{f}(\usarg)}\) has smooth values.
First, \(\iget[const]{\Ia}(a) = \iaccess[\applysubst{\sigma}{a}]{\I} = \iget[const]{\Ita}(a)\) holds because the adjoint to $\sigma$ for $\iportray{\I}$ in the case of programs is independent of $\iget[state]{\Ia}$ (programs have access to their initial state at runtime).
Likewise \(\iget[const]{\Ia}(C) = \iget[const]{\Ita}(C)\) for quantifier symbols, because the adjoint is independent of $\iget[state]{\Ia}$ for quantifier symbols.
By \rref{lem:coincidence-term},
\(\ivaluation{\imodif[const]{\I}{\,\usarg}{d}}{\applysubst{\sigma}{f}(\usarg)}
= \ivaluation{\imodif[const]{\It}{\,\usarg}{d}}{\applysubst{\sigma}{f}(\usarg)}\) when \(\iget[state]{\I}=\iget[state]{\It}\) on \(\freevarsdef{\applysubst{\sigma}{f}(\usarg)} \subseteq \freevarsdef{\sigma}\).
Also \(\imodels{\imodif[const]{\I}{\,\usarg}{d}}{\applysubst{\sigma}{p}(\usarg)}\)
iff \(\imodels{\imodif[const]{\It}{\,\usarg}{d}}{\applysubst{\sigma}{p}(\usarg)}\)
by \rref{lem:coincidence} when \(\iget[state]{\I}=\iget[state]{\It}\) on \(\freevarsdef{\applysubst{\sigma}{p}(\usarg)} \subseteq \freevarsdef{\sigma}\).
Thus, \(\iget[const]{\Ia}=\iget[const]{\Ita}\) when $\iget[state]{\I}=\iget[state]{\It}$ on $\freevarsdef{\sigma}$.

\newcommand{\Iaz}{\iconcat[state=\mu]{\Ia}}%
\newcommand{\Itaz}{\iconcat[state=\mu]{\Ita}}%
If $\sigma$ is $U$-admissible for $\phi$ (or $\theta$ or $\alpha$), then
  \(\freevarsdef{\applysubst{\sigma}{f(\usarg)}}\cap U=\emptyset\), i.e.\
  \(\freevarsdef{\applysubst{\sigma}{f(\usarg)}}\subseteq\scomplement{U}\)
  for every function symbol $f\in\intsigns{\phi}$ (or $\theta$ or $\alpha$) and likewise for predicate symbols $p\in\intsigns{\phi}$.
  Since \(\iget[state]{\I}=\iget[state]{\It}\) on $\scomplement{U}$ was assumed,
  \(\iget[const]{\Ita}=\iget[const]{\Ia}\) on the function and predicate symbols in $\intsigns{\phi}$ (or $\theta$ or $\alpha$).
  Finally \(\iget[const]{\Ita}=\iget[const]{\Ia}\) on $\intsigns{\phi}$ (or $\intsigns{\theta}$ or $\intsigns{\alpha}$, respectively) implies that
  \(\imodel{\Ita}{\phi}=\imodel{\Ia}{\phi}\) by \rref{lem:coincidence}
  (since \(\imodels{\Itaz}{\phi}\) iff \(\imodels{\Iaz}{\phi}\) holds for all $\iget[state]{\Itaz}$)
  and that \(\ivalues{\Ia}{\theta} = \ivalues{\Ita}{\theta}\) by \rref{lem:coincidence-term}
  and that \(\iaccess[\alpha]{\Ita} = \iaccess[\alpha]{\Ia}\) by \rref{lem:coincidence-HP}, respectively.
\qedhere
\end{proofatend}

Substituting equals for equals is sound by the compositional semantics of \dL.
The more general uniform substitutions are still sound, because the semantics of uniform substitutes of expressions agrees with the semantics of the expressions themselves in the adjoint interpretations.
The semantic modification of adjoint interpretations has the same effect as the syntactic uniform substitution.

\begin{lemma}[Uniform substitution for terms] \label{lem:usubst-term}
The uniform substitution $\sigma$ and its adjoint interpretation $\iportray{\Ia}$ for $\iportray{\I}$ have the same semantics for all \emph{terms} $\theta$:
\[\ivaluation{\I}{\applyusubst{\sigma}{\theta}} = \ivaluation{\Ia}{\theta}\]
\end{lemma}
\begin{proofatend}
\def\Im{\vdLint[const=\usebox{\ImnI},state=\nu\oplus(x\mapsto\usebox{\Imnx})]}%

The proof is by structural induction on $\theta$ and the structure of $\sigma$.
\begin{compactenum}
\item
  \(\ivaluation{\I}{\applyusubst{\sigma}{x}} 
  = \ivaluation{\I}{x} = \iget[state]{\I}(x) \)
  = \(\ivaluation{\Ia}{x}\)
  \(\text{since}~ x\not\in\replacees{\sigma}\)
  for variable $x\in\allvars$

\item
  Let \(f\in\replacees{\sigma}\).
  \(\ivaluation{\I}{\applyusubst{\sigma}{f(\theta)}}
  = \ivaluation{\I}{(\applyusubst{\sigma}{f})\big(\applyusubst{\sigma}{\theta}\big)}
  = \ivaluation{\I}{\applyusubst{\{\usarg\mapsto\applyusubst{\sigma}{\theta}\}}{\applysubst{\sigma}{f(\usarg)}}}\)
  \(\overset{\text{IH}}{=} \ivaluation{\Iminner}{\applysubst{\sigma}{f(\usarg)}}\)
    \(= (\iget[const]{\Ia}(f))(d)\)
    \\
  \(= (\iget[const]{\Ia}(f))(\ivaluation{\Ia}{\theta})
  = \ivaluation{\Ia}{f(\theta)}\)
  with
  \(d\mdefeq\ivaluation{\I}{\applyusubst{\sigma}{\theta}} 
  \overset{\text{IH}}{=} \ivaluation{\Ia}{\theta}\)
  by using the induction hypothesis twice,
  once for \(\applyusubst{\sigma}{\theta}\) on the smaller $\theta$ 
  and once for
  \(\applyusubst{\{\usarg\mapsto\applyusubst{\sigma}{\theta}\}}{\applysubst{\sigma}{f(\usarg)}}\)
  on the possibly bigger term
  \({\applysubst{\sigma}{f(\usarg)}}\)
  but the structurally simpler uniform substitution
  \(\applyusubst{\{\usarg\mapsto\applyusubst{\sigma}{\theta}\}}{\dots}\)
  that is a substitution on the symbol $\usarg$ of arity zero, not a substitution of functions with arguments.
  For well-foundedness of the induction note that the $\usarg$ substitution only happens for function symbols $f$ with at least one argument $\theta$
  so not for $\usarg$ itself.
  
\item
  \(\ivaluation{\I}{\applyusubst{\sigma}{g(\theta)}}
  = \ivaluation{\I}{g(\applyusubst{\sigma}{\theta})}\)
  = \(\iget[const]{\I}(g)\big(\ivaluation{\I}{\applyusubst{\sigma}{\theta}}\big)\)
  \(\overset{\text{IH}}{=} \iget[const]{\I}(g)\big(\ivaluation{\Ia}{\theta}\big)\)
  \(= \iget[const]{\Ia}(g)\big(\ivaluation{\Ia}{\theta}\big)
  = \ivaluation{\Ia}{g(\theta)}\)
  by induction hypothesis and since \(\iget[const]{\I}(g)=\iget[const]{\Ia}(g)\) as the interpretation of $g$ does not change in $\iget[const]{\Ia}$
  \(\text{when}~g\not\in\replacees{\sigma}\).

\item
  \(\ivaluation{\I}{\applyusubst{\sigma}{\theta+\eta}}
  = \ivaluation{\I}{\applyusubst{\sigma}{\theta} + \applyusubst{\sigma}{\eta}}
  = \ivaluation{\I}{\applyusubst{\sigma}{\theta}} + \ivaluation{\I}{\applyusubst{\sigma}{\eta}}\)
  \(\overset{\text{IH}}{=} \ivaluation{\Ia}{\theta} + \ivaluation{\Ia}{\eta}
  = \ivaluation{\Ia}{\theta+\eta}\)
  by induction hypothesis.

\item
  \(\ivaluation{\I}{\applyusubst{\sigma}{\theta\cdot\eta}}
  = \ivaluation{\I}{\applyusubst{\sigma}{\theta} \cdot \applyusubst{\sigma}{\eta}}
  = \ivaluation{\I}{\applyusubst{\sigma}{\theta}} \cdot \ivaluation{\I}{\applyusubst{\sigma}{\eta}}\)
  \(\overset{\text{IH}}{=} \ivaluation{\Ia}{\theta} \cdot \ivaluation{\Ia}{\eta}
  = \ivaluation{\Ia}{\theta\cdot\eta}\)
  by induction hypothesis.

\item
\(
\ivaluation{\I}{\applyusubst{\sigma}{\der{\theta}}}
= \ivaluation{\I}{\der{\applyusubst{\sigma}{\theta}}}
= \sum_x \iget[state]{\I}(\D{x}) \itimes \Dp[x]{\ivaluation{\I}{\applyusubst{\sigma}{\theta}}}
\overset{\text{IH}}{=} \sum_x \iget[state]{\I}(\D{x}) \itimes \Dp[x]{\ivaluation{\Ia}{\theta}}
= \ivaluation{\Ia}{\der{\theta}}
\)
by induction hypothesis,
provided $\sigma$ is $\allvars$-admissible for $\theta$, i.e. does not introduce any variables or differential symbols, 
so that \rref{cor:adjointUsubst} implies \(\iget[const]{\Ia}=\iget[const]{\Ita}\) for all $\iget[state]{\I},\iget[state]{\It}$ (that agree on $\scomplement{\allvars}=\emptyset$, which imposes no condition on $\iget[state]{\I},\iget[state]{\It}$).
In particular, the adjoint interpretation $\iget[const]{\Ia}$ is the same for all ways of changing the value of variable $x$ in state $\iget[state]{\I}$ when forming the partial derivative.
\qedhere
\end{compactenum}
\end{proofatend}

\noindent
The uniform substitute of a formula is true at $\iget[state]{\I}$ in an interpretation iff the formula itself is true at $\iget[state]{\I}$ in its adjoint interpretation.
Uniform substitution lemmas are proved by simultaneous induction for formulas and programs, as they are mutually recursive.

\begin{lemma}[Uniform substitution for formulas] \label{lem:usubst}
The uniform substitution $\sigma$ and its adjoint interpretation $\iportray{\Ia}$ for $\iportray{\I}$ have the same semantics for all \emph{formulas} $\phi$:
\[\imodels{\I}{\applyusubst{\sigma}{\phi}} ~\text{iff}~ \imodels{\Ia}{\phi}\]
\end{lemma}
\begin{proofatend}
The proof is by structural induction on $\phi$ and the structure on $\sigma$, simultaneously with \rref{lem:usubst-HP}.
\begin{compactenum}
\item
  \(\imodels{\I}{\applyusubst{\sigma}{\theta\geq\eta}}\)
  iff \(\imodels{\I}{\applyusubst{\sigma}{\theta} \geq \applyusubst{\sigma}{\eta}}\)
  iff \(\ivaluation{\I}{\applyusubst{\sigma}{\theta}} \geq \ivaluation{\I}{\applyusubst{\sigma}{\eta}}\),
  by \rref{lem:usubst-term},
  iff \(\ivaluation{\Ia}{\theta} \geq \ivaluation{\Ia}{\eta}\)
  iff \(\ivaluation{\Ia}{\theta\geq\eta}\).

\item
  Let \(p\in\replacees{\sigma}\).
  Then
  \(\imodels{\I}{\applyusubst{\sigma}{p(\theta)}}\)
  iff \(\imodels{\I}{(\applyusubst{\sigma}{p})\big(\applyusubst{\sigma}{\theta}\big)}\)
  iff \(\imodels{\I}{\applyusubst{\{\usarg\mapsto\applyusubst{\sigma}{\theta}\}}{\applysubst{\sigma}{p(\usarg)}}}\)
  iff, by IH, \(\imodels{\Iminner}{\applysubst{\sigma}{p(\usarg)}}\)
  iff \(d \in \iget[const]{\Ia}(p)\)
  iff \((\ivaluation{\Ia}{\theta}) \in \iget[const]{\Ia}(p)\)
  iff \(\imodels{\Ia}{p(\theta)}\)
  with \(d\mdefeq\ivaluation{\I}{\applyusubst{\sigma}{\theta}} = \ivaluation{\Ia}{\theta}\)
  by using \rref{lem:usubst-term} for \(\applyusubst{\sigma}{\theta}\)
  and by using the induction hypothesis
  for \(\applyusubst{\{\usarg\mapsto\applyusubst{\sigma}{\theta}\}}{\applysubst{\sigma}{p(\usarg)}}\)
  on the possibly bigger formula \({\applysubst{\sigma}{p(\usarg)}}\) but the structurally simpler uniform substitution \(\applyusubst{\{\usarg\mapsto\applyusubst{\sigma}{\theta}\}}{\dots}\) that is a mere substitution on function symbol $\usarg$ of arity zero, not a substitution of predicates.

\item
  Let \(q\not\in\replacees{\sigma}\). Then
  \(\imodels{\I}{\applyusubst{\sigma}{q(\theta)}}\)
  iff \(\imodels{\I}{q(\applyusubst{\sigma}{\theta})}\)
  iff \(\big(\ivaluation{\I}{\applyusubst{\sigma}{\theta}}\big)\in \iget[const]{\I}(q)\)
  so, by \rref{lem:usubst-term}, that holds iff \(\big(\ivaluation{\Ia}{\theta}\big) \in \iget[const]{\I}(q)\)
  iff \(\big(\ivaluation{\Ia}{\theta}\big) \in \iget[const]{\Ia}(q)\)
  iff \(\imodels{\Ia}{q(\theta)}\)
  since \(\iget[const]{\I}(q)=\iget[const]{\Ia}(q)\) as the interpretation of $q$ does not change in $\iget[const]{\Ia}$ when \(q\not\in\replacees{\sigma}\).

\item
\def\ImM{\imodif[const]{\I}{\uscarg}{R}}%
\let\ImN\ImM%
\def\IaM{\imodif[const]{\Ia}{\uscarg}{R}}%
  For the case \({\applyusubst{\sigma}{\contextapp{C}{\phi}}}\), first show 
  \(\imodel{\I}{\applyusubst{\sigma}{\phi}} = \imodel{\Ia}{\phi}\).
  By induction hypothesis for the smaller $\phi$:
  \(\imodels{\It}{\applyusubst{\sigma}{\phi}}\)
  iff
  \(\imodels{\Ita}{\phi}\),
  where 
  \(\imodel{\Ita}{\phi}=\imodel{\Ia}{\phi}\)
  by \rref{cor:adjointUsubst}
  for all $\iget[state]{\Ia},\iget[state]{\Ita}$
  (that agree on $\scomplement{\allvars}=\emptyset$, which imposes no condition on $\iget[state]{\I},\iget[state]{\It}$)
  since $\sigma$ is $\allvars$-admissible for $\phi$.
  The proof proceeds:

  \(\imodels{\I}{\applyusubst{\sigma}{\contextapp{C}{\phi}}}\)
  \(=\imodel{\I}{\contextapp{\applyusubst{\sigma}{C}}{\applyusubst{\sigma}{\phi}}}\)
  \(= \imodel{\I}{\applyusubst{\{\uscarg\mapsto\applyusubst{\sigma}{\phi}\}}{\applysubst{\sigma}{\contextapp{C}{\uscarg}}}}\),
  so, by induction hypothesis for the structurally simpler uniform substitution ${\{\uscarg\mapsto\applyusubst{\sigma}{\phi}\}}$ that is a mere substitution on quantifier symbol $\uscarg$ of arity zero, iff
  \(\imodels{\ImM}{\applysubst{\sigma}{\contextapp{C}{\uscarg}}}\)
  since the adjoint to \(\{\uscarg\mapsto\applyusubst{\sigma}{\phi}\}\) is $\iget[const]{\ImM}$ with \(R\mdefeq\imodel{\I}{\applyusubst{\sigma}{\phi}}\) by definition.
  
  Also
  \(\imodels{\Ia}{\contextapp{C}{\phi}}\)
  \(= \iget[const]{\Ia}(C)\big(\imodel{\Ia}{\phi}\big)\)
  \(= \imodel{\ImN}{\applysubst{\sigma}{\contextapp{C}{\uscarg}}}\)
  for \(R=\imodel{\Ia}{\phi}=\imodel{\I}{\applyusubst{\sigma}{\phi}}\) by induction hypothesis.
  Both sides are, thus, equivalent.
  
\item
  The case \({\applyusubst{\sigma}{\contextapp{C}{\phi}}}\) for $C\not\in\replacees{\sigma}$ again first shows
  \(\imodel{\I}{\applyusubst{\sigma}{\phi}} = \imodel{\Ia}{\phi}\)
  for all $\iget[state]{\I}$ using that $\sigma$ is $\allvars$-admissible for $\phi$.
  Then
  \(\imodels{\I}{\applyusubst{\sigma}{\contextapp{C}{\phi}}}\)
  \(= \imodel{\I}{\contextapp{C}{\applyusubst{\sigma}{\phi}}}\)
  \(= \iget[const]{\I}(C)\big(\imodel{\I}{\applyusubst{\sigma}{\phi}}\big)\)
  \(= \iget[const]{\I}(C)\big(\imodel{\Ia}{\phi}\big)\)
  \(= \iget[const]{\Ia}(C)\big(\imodel{\Ia}{\phi}\big)\)
  \(= \imodel{\Ia}{\contextapp{C}{\phi}}\)
  iff \(\imodels{\Ia}{\contextapp{C}{\phi}}\)

\item
  \(\imodels{\I}{\applyusubst{\sigma}{\lnot\phi}}\)
  iff \(\imodels{\I}{\lnot\applyusubst{\sigma}{\phi}}\)
  iff \(\inonmodels{\I}{\applyusubst{\sigma}{\phi}}\),
  so by IH,
  iff \(\inonmodels{\Ia}{\phi}\)
  iff \(\imodels{\Ia}{\lnot\phi}\)

\item
  \(\imodels{\I}{\applyusubst{\sigma}{\phi\land\psi}}\)
  iff \(\imodels{\I}{\applyusubst{\sigma}{\phi} \land \applyusubst{\sigma}{\psi}}\)
  iff \(\imodels{\I}{\applyusubst{\sigma}{\phi}}\) and \(\imodels{\I}{\applyusubst{\sigma}{\psi}}\),
  by induction hypothesis,
  iff \(\imodels{\Ia}{\phi}\) and \(\imodels{\Ia}{\psi}\)
  iff \(\imodels{\Ia}{\phi\land\psi}\)

\item
\def\Imd{\imodif[state]{\I}{x}{d}}%
\def\Iamd{\imodif[state]{\Ia}{x}{d}}%
\def\Imda{\iadjointSubst{\sigma}{\Imd}}%
  \(\imodels{\I}{\applyusubst{\sigma}{\lexists{x}{\phi}}}\)
  iff \(\imodels{\I}{\lexists{x}{\applyusubst{\sigma}{\phi}}}\)
  (provided that $\sigma$ is $\{x\}$-admissible for $\phi$)
  iff \(\imodels{\Imd}{\applyusubst{\sigma}{\phi}}\) for some $d$,
  so, by induction hypothesis,
  iff \(\imodels{\Imda}{\phi}\) for some $d$,
  which is equivalent to
  \(\imodels{\Iamd}{\phi}\) by \rref{cor:adjointUsubst} as $\sigma$ is $\{x\}$-admissible for $\phi$ and $\iget[state]{\I}=\iget[state]{\Imd}$ on $\scomplement{\{x\}}$.
  Thus, this is equivalent to
  \(\imodels{\Ia}{\lexists{x}{\phi}}\).
  
\item The case
  \(\imodels{\I}{\applyusubst{\sigma}{\lforall{x}{\phi}}}\)
  follows by duality \(\lforall{x}{\phi} \mequiv \lnot\lexists{x}{\lnot\phi}\), which is respected in the definition of uniform substitutions.

\item
  \newcommand{\Iat}{\iconcat[state=\omega]{\Ia}}%
  \(\imodels{\I}{\applyusubst{\sigma}{\ddiamond{\alpha}{\phi}}}\)
  iff \(\imodels{\I}{\ddiamond{\applyusubst{\sigma}{\alpha}}{\applyusubst{\sigma}{\phi}}}\)
  (provided $\sigma$ is $\boundvarsdef{\applyusubst{\sigma}{\alpha}}$-admissible for $\phi$)
  iff there is a $\iget[state]{\It}$ such that
  \(\iaccessible[\applyusubst{\sigma}{\alpha}]{\I}{\It}\) and \(\imodels{\It}{\applyusubst{\sigma}{\phi}}\),
  which, by \rref{lem:usubst-HP} and induction hypothesis, respectively, is equivalent to:
  there is a $\iget[state]{\Ita}$ such that
  \(\iaccessible[\alpha]{\Ia}{\Ita}\) and \(\imodels{\Ita}{\phi}\),
  which is equivalent to
  \(\imodels{\Ia}{\ddiamond{\alpha}{\phi}}\),
  because \(\imodels{\Ita}{\phi}\) is equivalent to \(\imodels{\Iat}{\phi}\) by \rref{cor:adjointUsubst}
  as $\sigma$ is $\boundvarsdef{\applyusubst{\sigma}{\alpha}}$-admissible for $\phi$ and \(\iget[state]{\Ia}=\iget[state]{\Iat}\) on $\scomplement{\boundvarsdef{\applyusubst{\sigma}{\alpha}}}$ by \rref{lem:bound} since
  \(\iaccessible[\applyusubst{\sigma}{\alpha}]{\I}{\It}\).

\item The case
  \(\imodels{\I}{\applyusubst{\sigma}{\dbox{\alpha}{\phi}}}\)
  follows by duality \(\dbox{\alpha}{\phi} \mequiv \lnot\ddiamond{\alpha}{\lnot\phi}\), which is respected in the definition of uniform substitutions.
\qedhere
\end{compactenum}
\end{proofatend}

\noindent
The uniform substitute of a program has a run from $\iget[state]{\I}$ to $\iget[state]{\It}$ in an interpretation iff the program itself has a run from $\iget[state]{\I}$ to $\iget[state]{\It}$ in its adjoint interpretation.

\begin{lemma}[Uniform substitution for programs] \label{lem:usubst-HP}
The uniform substitution $\sigma$ and its adjoint interpretation $\iportray{\Ia}$ for $\iportray{\I}$ have the same semantics for all \emph{programs} $\alpha$:
\[
\iaccessible[{\applyusubst{\sigma}{\alpha}}]{\I}{\It}
~\text{iff}~
\iaccessible[\alpha]{\Ia}{\Ita}
\]
\end{lemma}
\begin{proofatend}
The proof is by structural induction on $\alpha$, simultaneously with \rref{lem:usubst}.
\begin{compactenum}
\item
  \(\iaccessible[\applyusubst{\sigma}{a}]{\I}{\It} = \iaccess[\applysubst{\sigma}{a}]{\I} = \iget[const]{\Ia}(a) = \iaccess[a]{\Ia}\)
  for program constant $a\in\replacees{\sigma}$
  (the proof is accordingly for $a\not\in\replacees{\sigma}$).

\item 
  \(\iaccessible[\applyusubst{\sigma}{\pumod{x}{\theta}}]{\I}{\It}
  = \iaccess[\pumod{x}{\applyusubst{\sigma}{\theta}}]{\I}\)
  iff \(\iget[state]{\It} = \modif{\iget[state]{\I}}{x}{\ivaluation{\I}{\applyusubst{\sigma}{\theta}}}\)
  = \(\modif{\iget[state]{\I}}{x}{\ivaluation{\Ia}{\theta}}\)
  by \rref{lem:usubst}, which is, thus, equivalent to
  \(\iaccessible[\pumod{x}{\theta}]{\Ia}{\Ita}\).

\item
  \(\iaccessible[\applyusubst{\sigma}{\ptest{\ivr}}]{\I}{\It}
  = \iaccess[\ptest{\applyusubst{\sigma}{\ivr}}]{\I}\)
  iff \(\iget[state]{\It}=\iget[state]{\I}\)
  and \(\imodels{\I}{\applyusubst{\sigma}{\ivr}}\),
  iff, by \rref{lem:usubst},
  \(\iget[state]{\Ita}=\iget[state]{\Ia}\)     and
  \(\imodels{\Ia}{\ivr}\),
  which is equivalent to
  \(\iaccessible[\ptest{\ivr}]{\Ia}{\Ita}\).

\item
\newcommand{\Izeta}{\iconcat[state=\varphi(t)]{\I}}
\def\Izetaa{\iadjointSubst{\sigma}{\Izeta}}%
\newcommand{\Iazeta}{\iconcat[state=\varphi(t)]{\Ia}}
  \(\iaccessible[\applyusubst{\sigma}{\pevolvein{\D{x}=\genDE{x}}{\ivr}}]{\I}{\It}
  = \iaccess[\pevolvein{\D{x}=\applyusubst{\sigma}{\genDE{x}}}
  {\applyusubst{\sigma}{\ivr}}]{\I}\)
  (provided that $\sigma$ is $\{x,\D{x}\}$-ad\-mis\-si\-ble for $\genDE{x},\ivr$)
  iff \(\mexists{\varphi:[0,T]\to\linterpretations{\Sigma}{V}}\)
  with \(\varphi(0)=\iget[state]{\I}\) on $\scomplement{\{\D{x}\}}$, \(\varphi(T)=\iget[state]{\It}\) and for all $t\geq0$:
  \(\D{\varphi}(t) = \ivaluation{\Izeta}{\applyusubst{\sigma}{\genDE{x}}}
  = \ivaluation{\Izetaa}{\genDE{x}}\) by \rref{lem:usubst-term}
  and
  \(\imodels{\Izeta}{\applyusubst{\sigma}{\ivr}}\),
  which, by \rref{lem:usubst}, holds iff
  \(\imodels{\Izetaa}{\ivr}\).
  
  Also
  \(\iaccessible[\pevolvein{\D{x}=\genDE{x}}{\ivr}]{\Ia}{\Ita}\)
  iff \(\mexists{\varphi:[0,T]\to\linterpretations{\Sigma}{V}}\)
  with \(\varphi(0)=\iget[state]{\I}\) on $\scomplement{\{\D{x}\}}$ and \(\varphi(T)=\iget[state]{\It}\) and for all $t\geq0$:
  \(\D{\varphi}(t) = \ivaluation{\Iazeta}{\genDE{x}}\)
  and
  \(\imodels{\Iazeta}{\ivr}\).
  Finally,
  \(\ivalues{\Iazeta}{\genDE{x}}=\ivalues{\Izetaa}{\genDE{x}}\) and
  \(\imodel{\Izetaa}{\ivr}=\imodel{\Iazeta}{\ivr}\)
  by \rref{cor:adjointUsubst}
  since $\sigma$ is $\{x,\D{x}\}$-admissible for $\genDE{x}$ and $\ivr$ and \(\iget[state]{\I}=\iget[state]{\Iazeta}\) on $\scomplement{\boundvarsdef{\pevolvein{\D{x}=\genDE{x}}{\ivr}}}\supseteq\scomplement{\{x,\D{x}\}}$ by \rref{lem:bound}.
  
\item  
  \(\iaccessible[\applyusubst{\sigma}{\pchoice{\alpha}{\beta}}]{\I}{\It}
  = \iaccess[\pchoice{\applyusubst{\sigma}{\alpha}}{\applyusubst{\sigma}{\beta}}]{\I}\)
  = \(\iaccess[\applyusubst{\sigma}{\alpha}]{\I} \cup \iaccess[\applyusubst{\sigma}{\beta}]{\I}\),
  which, by induction hypothesis, is equivalent to
  \(\iaccessible[\alpha]{\Ia}{\Ita}\) or \(\iaccessible[\beta]{\Ia}{\Ita}\),
  which is 
  \(\iaccessible[\alpha]{\Ia}{\Ita} \cup \iaccess[\beta]{\Ia} = \iaccess[\pchoice{\alpha}{\beta}]{\Ia}\).

\item
{\newcommand{\Iz}{\iconcat[state=\mu]{\I}}%
\newcommand{\Iza}{\iadjointSubst{\sigma}{\Iz}}%
\newcommand{\Iaz}{\iconcat[state=\mu]{\Ia}}%
  \(\iaccessible[\applyusubst{\sigma}{\alpha;\beta}]{\I}{\It}
  = \iaccess[\applyusubst{\sigma}{\alpha}; \applyusubst{\sigma}{\beta}]{\I}\)
  = \(\iaccess[\applyusubst{\sigma}{\alpha}]{\I} \compose \iaccess[\applyusubst{\sigma}{\beta}]{\I}\)
  (provided $\sigma$ is $\boundvarsdef{\applyusubst{\sigma}{\alpha}}$-admissible for $\beta$)
  iff there is a $\iget[state]{\Iz}$ such that
  \(\iaccessible[\applyusubst{\sigma}{\alpha}]{\I}{\Iz}\) and \(\iaccessible[\applyusubst{\sigma}{\beta}]{\Iz}{\It}\),
  which, by induction hypothesis, is equivalent to
  \(\iaccessible[\alpha]{\Ia}{\Iza}\) and \(\iaccessible[\beta]{\Iza}{\Ita}\).
  Yet, \(\iaccess[\beta]{\Iza}=\iaccess[\beta]{\Ia}\)
  by \rref{cor:adjointUsubst}, because $\sigma$ is $\boundvarsdef{\applyusubst{\sigma}{\alpha}}$-admissible for $\beta$ and \(\iget[state]{\I}=\iget[state]{\It}\) on $\scomplement{\boundvarsdef{\applyusubst{\sigma}{\alpha}}}$ by \rref{lem:bound} since \(\iaccessible[\applyusubst{\sigma}{\alpha}]{\I}{\Iz}\).
  Finally,
  \(\iaccessible[\alpha]{\Ia}{\Iaz}\) and \(\iaccessible[\beta]{\Iaz}{\Ita}\) for some $\iget[state]{\Iaz}$
  is equivalent to \(\iaccessible[\alpha;\beta]{\Ia}{\Ita}\).

}

\item
\newcommand{\Iz}[1][]{\iconcat[state=\nu_{#1}]{\I}}%
\newcommand{\Iza}[1][]{\iadjointSubst{\sigma}{\Iz[#1]}}%
\newcommand{\Iaz}[1][]{\iconcat[state=\nu_{#1}]{\Ia}}%
  \(\iaccessible[\applyusubst{\sigma}{\prepeat{\alpha}}]{\I}{\It}
  = \iaccess[\prepeat{(\applyusubst{\sigma}{\alpha})}]{\I}
  = \closureTransitive{\big(\iaccess[\applyusubst{\sigma}{\alpha}]{\I}\big)}
  = \cupfold_{n\in\naturals} (\iaccess[\applyusubst{\sigma}{\alpha}]{\I})^n
  \)
  (provided that $\sigma$ is $\boundvarsdef{\applyusubst{\sigma}{\alpha}}$-admissible for $\alpha$)
  iff there are $n\in\naturals$ and \(\iget[state]{\Iz[0]}=\iget[state]{\Ia},\iget[state]{\Iz[1]},\dots,\iget[state]{\Iz[n]}=\iget[state]{\Ita}\) such that
  \(\iaccessible[\applyusubst{\sigma}{\alpha}]{\Iz[i]}{\Iz[i+1]}\) for all $i<n$.
  By $n$ uses of the induction hypothesis, this is equivalent to
  \(\iaccessible[\alpha]{\Iza[i]}{\Iza[i+1]}\) for all $i<n$,
  which is equivalent to
  \(\iaccessible[\alpha]{\Iaz[i]}{\Iza[i+1]}\) 
  by \rref{cor:adjointUsubst}
  since $\sigma$ is $\boundvarsdef{\applyusubst{\sigma}{\alpha}}$-admissible for $\alpha$
  and \(\iget[state]{\Iza[i+1]}=\iget[state]{\Iza[i]}\) on $\scomplement{\boundvarsdef{\applyusubst{\sigma}{\alpha}}}$ by \rref{lem:bound} as \(\iaccessible[\applyusubst{\sigma}{\alpha}]{\Iz[i]}{\Iz[i+1]}\) for all $i<n$.
  Thus, this is equivalent to
  \(\iaccessible[\prepeat{\alpha}]{\Ia}{\Ita}
  = \closureTransitive{\big(\iaccess[\alpha]{\Ia}\big)}\).
\qedhere
\end{compactenum}
\end{proofatend}

\subsection{Soundness}

The uniform substitution lemmas are the key insights for the soundness of proof rule \irref{US}, which is only applicable if its uniform substitution is defined.
A proof rule is \emph{sound} iff validity of all its premises implies validity of its conclusion.

\begin{theorem}[Soundness of uniform substitution] \label{thm:usubst-sound}
  The proof rule \irref{US} is sound.
  \[
  \cinferenceRuleQuote{US}
  \]
\end{theorem}
\begin{proof}
\def\Ia{\iadjointSubst{\sigma}{\I}}%
Let the premise $\phi$ of \irref{US} be valid, i.e.\ \m{\imodels{\I}{\phi}} for all interpretations and states $\iportray{\I}$.
To show that the conclusion is valid, consider any interpretation and state $\iportray{\I}$ and show \(\imodels{\I}{\applyusubst{\sigma}{\phi}}\).
By \rref{lem:usubst}, \(\imodels{\I}{\applyusubst{\sigma}{\phi}}\) iff \(\imodels{\Ia}{\phi}\).
Now \(\imodels{\Ia}{\phi}\) holds, because \(\imodels{\I}{\phi}\) for all $\iportray{\I}$, including for $\iportray{\Ia}$, by premise.
The rule \irref{US0} is the special case of \irref{US} where $\sigma$ only substitutes predicate symbol $p$.
\qedhere
\end{proof}
Uniform substitutions can also be used to soundly instantiate locally sound proof rules or whole proofs just like proof rule \irref{US} soundly instantiates axioms or other valid formulas (\rref{thm:usubst-sound}).
An inference or proof rule is \emph{locally sound} iff its conclusion is valid in any interpretation $\iget[const]{\I}$ in which all its premises are valid.
All locally sound proof rules are sound.
The use of \rref{thm:usubst-rule} in a proof is marked \irref{USR}.

\begin{theorem}[Soundness of uniform substitution of rules] \label{thm:usubst-rule}
  All uniform substitution instances (with $\freevarsdef{\sigma}=\emptyset$) of locally sound inferences are \emph{locally sound}:
  \[
\linfer
{\phi_1 \quad \dots \quad \phi_n}
{\psi}
~~\text{locally sound}\qquad\text{implies}\qquad
\linfer%
{\applyusubst{\sigma}{\phi_1} \quad \dots \quad \applyusubst{\sigma}{\phi_n}}
{\applyusubst{\sigma}{\psi}}
~~\text{locally sound}
\irlabel{USR|USR}
  \]
\end{theorem}
\begin{proof}
\def\locproof{\mathcal{D}}%
Let $\locproof$ be the inference on the left and let $\applyusubst{\sigma}{\locproof}$ be the substituted inference on the right.
Assume $\locproof$ to be locally sound.
To show that $\applyusubst{\sigma}{\locproof}$ is locally sound, consider any $\iget[const]{\I}$ in which all premises of $\applyusubst{\sigma}{\locproof}$ are valid, i.e.\
\(\iget[const]{\I}\models{\applyusubst{\sigma}{\phi_j}}\) for all $j$.
That is, \(\imodels{\I}{\applyusubst{\sigma}{\phi_j}}\) for all $\iget[state]{\I}$ and all $j$.
By \rref{lem:usubst}, \(\imodels{\I}{\applyusubst{\sigma}{\phi_j}}\) is equivalent to
\(\imodels{\Ia}{\phi_j}\),
which, thus, also holds for all $\iget[state]{\I}$ and all $j$.
By \rref{cor:adjointUsubst}, \(\imodel{\Ia}{\phi_j}=\imodel{\Ita}{\phi_j}\) for any $\iget[state]{\Ita}$, since $\freevarsdef{\sigma}=\emptyset$.
Fix an arbitrary state $\iget[state]{\Ita}$.
\newcommand{\Itar}{\iconcat[state=\nu]{\Ita}}%
Then \(\imodels{\Itar}{\applyusubst{\sigma}{\phi_j}}\)  holds for all $\iget[state]{\Itar}$ and all $j$ for the same (arbitrary) $\iget[state]{\Ita}$ that determines $\iget[const]{\Ita}$.

Consequently, all premises of $\locproof$ are valid in the same $\iget[const]{\Ita}$, i.e. \(\iget[const]{\Ita}\models{\phi_j}\) for all $j$.
Thus, \(\iget[const]{\Ita}\models{\psi}\) by local soundness of $\locproof$.
That is, \(\imodels{\Ia}{\psi}=\imodel{\Ita}{\psi}\) by \rref{cor:adjointUsubst} for all $\iget[state]{\Ia}$.
By \rref{lem:usubst}, \(\imodels{\Ia}{\psi}\) is equivalent to \(\imodels{\I}{\applyusubst{\sigma}{\psi}}\),
which continues to hold for all $\iget[state]{\I}$.
Thus, \(\iget[const]{\I}\models{\applyusubst{\sigma}{\psi}}\), i.e.\ the conclusion of $\applyusubst{\sigma}{\locproof}$ is valid in $\iget[const]{\I}$, hence $\applyusubst{\sigma}{\locproof}$ is locally sound.
Consequently, all uniform substitution instances $\applyusubst{\sigma}{\locproof}$ of locally sound inferences $\locproof$ with $\freevarsdef{\sigma}=\emptyset$ are locally sound.
\qedhere
\end{proof}
If $\psi$ has a proof (i.e.\ $n=0$), \irref{USR} preserves local soundness even if $\freevarsdef{\sigma}\neq\emptyset$, because \irref{US} proves \(\applyusubst{\sigma}{\psi}\) from the provable $\psi$, which makes this inference locally sound, since local soundness is equivalent to soundness for $n=0$ premises.
If $\psi$ has a proof, uniform substitution of rules \irref{USR} for $n=0$ premises is identical to rule \irref{US}.

\begin{example}[Uniform substitutions are only globally sound]
Rule \irref{US} itself is only sound but not locally sound, so it cannot have been used on any unproved premises at any point during a proof that is to be instantiated by proof rule \irref{USR} from \rref{thm:usubst-rule}.
The following sound proof on the left with a modus ponens (marked \irref{MP}) has an unproved premise on which \irref{US} has been used at some point during the proof:

\noindent
\begin{minipage}{5.5cm}
\begin{sequentdeduction}[array]
\linfer[MP]
{\linfer
  {\lclose}
  {\lsequent{1=0} {\dbox{\pevolve{\D{x}=2}}{\,x<5}}}
 !
  \linfer[US]
  {\lsequent{} {f(x)=0}}
  {\lsequent{} {1=0}}
}%
{\lsequent{} {\dbox{\pevolve{\D{x}=2}}{\,x<5}}}
\end{sequentdeduction}\end{minipage}%
\qquad\text{\irref{USR} not applicable}\qquad\quad%
\begin{minipage}{2cm}
\[
\linfer[clash]
{\lsequent{} {0=0}}
{\lsequent{} {\dbox{\pevolve{\D{x}=2}}{\,x<5}}}
\]
\end{minipage}

\noindent
This use of \irref{US}, which substitutes $1$ for $f(\usarg)$, makes the left proof sound but \emph{not} locally sound.
That prevents rule \irref{USR} of \rref{thm:usubst-rule} from (unsoundly) concluding the uniform substitution instance on the right with
\(\sigma = \usubstlist{\usubstmod{f(\usarg)}{0}}\).
Rule \irref{US} assumes that its premise is valid (in all interpretations $\iget[const]{\I}$), but the (clashing) substitution instance on the right only proves one choice for $f$ to satisfy premise \(f(x)=0\).
Rule \irref{US} can still be used in the proof of a premise that proves without endangering local soundness, because proved premises are valid in all interpretations by soundness.
\end{example}

\section{Differential Dynamic Logic Axioms} \label{sec:dL-axioms}

Proof rules and axioms for a Hilbert-type axiomatization of \dL from prior work \cite{DBLP:conf/lics/Platzer12b} are shown in \rref{fig:dL}, except that, thanks to proof rule \irref{US}, axioms and proof rules now reduce to  the finite list of concrete \dL formulas in \rref{fig:dL} as opposed to an infinite collection of axioms from a finite list of axiom schemata along with schema variables, side conditions, and implicit instantiation rules.
Soundness of the axioms follows from soundness of corresponding axiom schemata \cite{Harel_et_al_2000,DBLP:conf/lics/Platzer12b}, but is easier to prove standalone, because it is a finite list of formulas without the need to prove soundness for all their instances.
Soundness of axioms, thus, reduces to validity of one formula as opposed to validity of all formulas that can be generated by the instantiation mechanism complying with the respective side conditions for that axiom schema.
The proof rules in \rref{fig:dL} are \emph{axiomatic rules}, i.e.\ pairs of concrete \dL formulas to be instantiated by \irref{USR}.
Soundness of axiomatic rules reduces to proving that their concrete conclusion formula is a consequence of their premise formula.
Further, $\usall$ is the vector of all relevant variables, which is finite-dimensional, or  considered as a built-in vectorial term. Proofs in the uniform substitution \dL calculus use \irref{US} (and variable renaming such as \(\lforall{x}{p(x)}\) to \(\lforall{y}{p(y)}\)) to instantiate the axioms from \rref{fig:dL} to the required form.

\begin{figure}[tb]
  \renewcommand*{\irrulename}[1]{\text{#1}}%
  \begin{calculuscollections}{\columnwidth}
    \begin{calculus}
      \cinferenceRule[diamond|$\didia{\cdot}$]{diamond axiom}
      {\linferenceRule[equiv]
        {\lnot\dbox{a}{\lnot p(\usall)}}
        {\axkey{\ddiamond{a}{p(\usall)}}}
      }
      {}
      \cinferenceRule[assignb|$\dibox{:=}$]{assignment / substitution axiom}
      {\linferenceRule[equiv]
        {p(f)}
        {\axkey{\dbox{\pupdate{\umod{x}{f}}}{p(x)}}}
      }
      {}%
      \irlabel{Dassignb|$\dibox{:=}$}%
      \cinferenceRule[testb|$\dibox{?}$]{test}
      {\linferenceRule[equiv]
        {(q \limply p)}
        {\axkey{\dbox{\ptest{q}}{p}}}
      }{}
      \cinferenceRule[choiceb|$\dibox{\cup}$]{axiom of nondeterministic choice}
      {\linferenceRule[equiv]
        {\dbox{a}{p(\usall)} \land \dbox{b}{p(\usall)}}
        {\axkey{\dbox{\pchoice{a}{b}}{p(\usall)}}}
      }{}
      \cinferenceRule[composeb|$\dibox{{;}}$]{composition} %
      {\linferenceRule[equiv]
        {\dbox{a}{\dbox{b}{p(\usall)}}}
        {\axkey{\dbox{a;b}{p(\usall)}}}
      }{}
      \cinferenceRule[iterateb|$\dibox{{}^*}$]{iteration/repeat unwind} %
      {\linferenceRule[equiv]
        {p(\usall) \land \dbox{a}{\dbox{\prepeat{a}}{p(\usall)}}}
        {\axkey{\dbox{\prepeat{a}}{p(\usall)}}}
      }{}
      \cinferenceRule[K|K]{K axiom / modal modus ponens} %
      {\linferenceRule[impl]
        {\dbox{a}{(p(\usall)\limply q(\usall))}}
        {(\dbox{a}{p(\usall)}\limply\axkey{\dbox{a}{q(\usall)}})}
      }{}
      \cinferenceRule[I|I]{loop induction}
      {\linferenceRule[impl]
        {\dbox{\prepeat{a}}{(p(\usall)\limply\dbox{a}{p(\usall)})}}
        {(p(\usall)\limply\axkey{\dbox{\prepeat{a}}{p(\usall)}})}
      }{}
      \cinferenceRule[V|V]{vacuous $\dbox{}{}$}
      {\linferenceRule[impl]
        {p}
        {\axkey{\dbox{a}{p}}}
      }{}%
    \end{calculus}
    \qquad
    \begin{calculus}
      \cinferenceRule[G|G]{$\dbox{}{}$ generalisation} %
      {\linferenceRule[formula]
        {p(\usall)}
        {\dbox{a}{p(\usall)}}
      }{}
      \cinferenceRule[gena|$\forall{}$]{$\forall{}$ generalisation}
      {\linferenceRule[formula]
        {p(x)}
        {\lforall{x}{p(x)}}
      }{}%
      \cinferenceRule[MP|MP]{modus ponens}
      {\linferenceRule[formula]
        {p\limply q \quad p}
        {q}
      }{}%
      \cinferenceRule[CQ|CQ]{congequal congruence of equations on formulas (convert term congruence to formula congruence: term congruence on formulas)}
      {\linferenceRule[formula]
        {f(\usall) = g(\usall)}
        {p(f(\usall)) \lbisubjunct p(g(\usall))}
      }{}%
      \cinferenceRule[CE|CE]{congequiv congruence of equivalences on formulas}
      {\linferenceRule[formula]
        {p(\usall) \lbisubjunct q(\usall)}
        {\contextapp{C}{p(\usall)} \lbisubjunct \contextapp{C}{q(\usall)}}
      }{}%
    \end{calculus}%
  \end{calculuscollections}
  \caption{Differential dynamic logic axioms and proof rules}
  \label{fig:dL}
\end{figure}
\emph{Diamond axiom} \irref{diamond} expresses the duality of the $\dbox{\cdot}{}$ and $\ddiamond{\cdot}{}$ modalities.
\emph{Assignment axiom} \irref{assignb} expresses that $p(x)$ holds after the assignment $\pupdate{\pumod{x}{f}}$ iff $p(f)$ holds initially.
\emph{Test axiom} \irref{testb} expresses that $p$ holds after the test $\ptest{q}$ iff $p$ is implied by $q$, because test~$\ptest{q}$ only runs when~$q$ holds.
\emph{Choice axiom} \irref{choiceb} expresses that $p(\usall)$ holds after all runs of $\pchoice{a}{b}$ iff $p(\usall)$ holds after all runs of $a$ and after all runs of $b$.
\emph{Sequential composition axiom} \irref{composeb} expresses that $p(\usall)$ holds after all runs of $a;b$ iff, after all runs of $a$, it is the case that $p(\usall)$ holds after all runs of $b$.
\emph{Iteration axiom} \irref{iterateb} expresses that $p(\usall)$ holds after all repetitions of $a$ iff it holds initially and, after all runs of $a$, it is the case that $p(\usall)$ holds after all repetitions of $a$.
Axiom \irref{K} is the \emph{modal modus ponens} from modal logic \cite{HughesCresswell96}.
\emph{Induction axiom} \irref{I} expresses that
if, no matter how often $a$ repeats, $p(\usall)$ holds after all runs of $a$ if it was true before, then, if $p(\usall)$ holds initially, it holds after all repetitions of $a$.
\emph{Vacuous axiom} \irref{V} expresses that arity 0 predicate symbol $p$ continues to hold after all runs of $a$ if it was true before.

\emph{G\"odel's generalization} rule \irref{G} expresses that $p(\usall)$ holds after all runs of $a$ if $p(\usall)$ is valid.
Accordingly \irref{gena} is the \emph{$\forall$-generalization} rule.
\irref{MP} is \emph{modus ponens}.
\emph{Congruence} rules \irref{CQ}, \irref{CE} are not needed but included to efficiently use axioms in any context.
Congruence rule \irref{CT} derives from \irref{CQ} using \(p(\usarg) \mdefequiv \big(c(\usarg)=c(g(\usall))\big)\) and reflexivity:
\[
      \dinferenceRule[CT|CT]{congterm congruence on terms}
      {\linferenceRule[formula]
        {f(\usall) = g(\usall)}
        {c(f(\usall)) = c(g(\usall))}
      }{}%
\]

\begin{remark} \label{rem:usall}%
The use of variable vector $\usall$ is not essential but simplifies concepts.
An equivalent axiomatization is obtained when considering $p(\usall)$ to be a quantifier symbol of arity 0 in the axiomatization, or as $\contextapp{C}{\ltrue}$ with a quantifier symbol of arity 1.
Neither replacements of quantifier symbols nor (vectorial) placeholders $\usarg$ for the substitutions \(\usubstlist{\usubstmod{p(\usarg)}{\psi}}\) that are used for $p(\usall)$ cause any free variables in the substitution.
The mnemonic notation \(\sigma=\usubstlist{\usubstmod{p(\usall)}{\phi}}\)
adopted for such uniform substitutions reminds that the variables $\usall$ are not free in $\sigma$ even if they occur in the replacement $\phi$.
\end{remark}

Sound axioms are just valid formulas, so true in all states.
For example, in any state where \(\dbox{a}{\dbox{b}{p(\usall})}\) is true, \(\dbox{a;b}{p(\usall)}\) is true, too, by equivalence axiom \irref{composeb}.
Using axiom \irref{composeb} to replace one by the other is a truth-preserving transformation, i.e.\ in any state in which one is true, the other is true, too.
Sound rules are validity-preserving, i.e.\ the conclusion is valid if the premises are valid, which is weaker than truth-preserving transformation.
For proof search, the \dL axioms are meant to be used to reduce the \emph{axiom key} (marked \axkey{blue}) to the structurally simpler remaining conditions (right-hand sides of equivalences and the conditions assumed in implications).

\paragraph{Real Quantifiers.}

Besides (decidable) real arithmetic (whose use is denoted \irref{qear}), complete axioms for first-order logic can be adopted to express
universal instantiation \irref{allinst} (if $p$ is true of all $x$, it is also true of constant function symbol $f$),
distributivity \irref{alldist},
and vacuous quantification \irref{vacuousall} (predicate $p$ of arity zero does not depend on $x$).

  \begin{calculuscollections}{\columnwidth}
    \begin{calculus}
      \cinferenceRule[allinst|$\forall$i]{universal instantiation}
       {(\axkey{\lforall{x}{p(x)}}) \limply p(f)}
       {}
       \cinferenceRule[alldist|$\forall{\limply}$]{$\forall$ distributes over $\limply$}
       {\lforall{x}{(p(x)\limply q(x))} \limply (\lforall{x}{p(x)} \limply \axkey{\lforall{x}{q(x)}})}
       {}
       \cinferenceRule[vacuousall|V$_\forall$]{vacuous universal quantifier}
       {p \limply \axkey{\lforall{x}{p}}}
       {}%
    \end{calculus}
  \end{calculuscollections}

\paragraph{The Significance of Clashes.}

This section illustrates how uniform substitutions tell sound instantiations apart from unsound proof attempts.
Rule \irref{US} clashes exactly when the substitution introduces a free variable into a bound context, which would be unsound.
\rref{ex:usubst-assign1} on p.~\pageref{ex:usubst-assign1} already showed that even an occurrence of $p(x)$ in a context where $x$ is bound does not permit mentioning $x$ in the replacement except in the $\usarg$ places.
\irref{US} can directly handle even nontrivial binding structures, though, e.g. from \irref{assignb}
with the substitution \(\sigma=\usubstlist{\usubstmod{f}{x^2},\usubstmod{p(\usarg)}{\dbox{\prepeat{(z:=\usarg+z)};z:=\usarg+yz}{y\geq\usarg}}}\):
\[
\linfer[US]
{\dbox{\pupdate{\umod{x}{f}}}{p(x)} \lbisubjunct p(f)}
{\dbox{x:=x^2}{\dbox{\prepeat{(z:=x{+}z)};z:=x{+}yz}{\,y{\geq}x}} \lbisubjunct
\dbox{\prepeat{(z:=x^2{+}z)};z:=x^2{+}yz}{\,y{\geq}x^2}
}
\]
It is soundness-critical that \irref{US} clashes when trying to instantiate $p$ in \irref{vacuousall} with a formula that mentions the bound variable $x$:
\[
\linfer[clash]
{p \limply \lforall{x}{p}}
{x\geq0 \limply \lforall{x}{(x\geq0)}}
\qquad
\usubstlist{\usubstmod{p}{x\geq0}}
\]
It is soundness-critical that \irref{US} clashes when substituting $p$ in vacuous program axiom \irref{V} with a formula with a free occurrence of a variable bound by the replacement of $a$:
\[
\linfer[clash]
{p \limply \dbox{a}{p}}
{x\geq0 \limply \dbox{\pevolve{\D{x}=-1}}{\,x\geq0}}
\qquad
\usubstlist{\usubstmod{a}{\pevolve{\D{x}=-1}},\usubstmod{p}{x\geq0}}
\]
Additional free variables are acceptable, though, e.g.\ in replacements for $p$ as long as they are not bound in the particular context into which they will be substituted:
\[
\linfer[US]
{p \limply \dbox{a}{p}}
{y\geq0 \limply \dbox{\pevolve{\D{x}=-1}}{\,y\geq0}}
\qquad
\usubstlist{\usubstmod{a}{\pevolve{\D{x}=-1}},\usubstmod{p}{y\geq0}}
\]
Complex formulas are acceptable as replacements for $p$ if their free variables are not bound in the context, e.g., using \(\sigma = \usubstlist{\usubstmod{a}{\Dupdate{\Dumod{\D{x}}{5x}}},\usubstmod{p}{\dbox{\pevolve{\D{x}=x^2-2x+2}}{x\geq1}}}\):
\[
\linfer[US]
{p \limply \dbox{a}{p}}
{\dbox{\pevolve{\D{x}=x^2-2x+2}}{\,x\geq1} \limply \dbox{\Dupdate{\Dumod{\D{x}}{5x}}}{\dbox{\pevolve{\D{x}=x^2-2x+2}}{\,x\geq1}}}
\]
But it is soundness-critical that \irref{US} clashes when substituting a formula with a free dependence on $\D{x}$ for $p$ into a context where $\D{x}$ will be bound after the substitution:
\[
\linfer[clash]
{p \limply \dbox{a}{p}}
{\der{x-1}\geq0 \limply \dbox{\Dupdate{\Dumod{\D{x}}{5x}}}{\der{x-1}\geq0}}
\qquad
\usubstlist{\usubstmod{a}{\Dupdate{\Dumod{\D{x}}{5x}}},\usubstmod{p}{\der{x-1}\geq0}}
\]
G\"odel's generalization rule \irref{G} uses $p(\usall)$ instead of the $p$ that \irref{V} uses, so its \irref{USR} instance allows all variables $\usall$ to occur in the replacement without causing a clash:
\[
\linfer[G+USR]
{(-x)^2\geq0}
{\dbox{\pevolve{\D{x}=-1}}{(-x)^2\geq0}}
\qquad
\usubstlist{\usubstmod{a}{\pevolve{\D{x}=-1}},\usubstmod{p(\usall)}{(-x)^2\geq0}}
\]
Intuitively, the argument $\usall$ in this uniform substitution instance of \irref{G} was not introduced as part of the substitution but merely put in for the placeholder $\usarg$ instead.
Let \(\usall=(x,y)\), 
\irref{US} 
\(\usubstlist{\usubstmod{a}{x:=x+1},\usubstmod{b}{x:=0;\pevolve{\D{y}=-2}},\usubstmod{p(\usall)}{x\geq y}}\) derives from \irref{choiceb}:%
\begin{sequentdeduction}
  \linfer[US]
    {\dbox{\pchoice{a}{b}}{p(\usall)} \lbisubjunct \dbox{a}{p(\usall)} \land \dbox{b}{p(\usall)}}
  {\dbox{\pchoice{x:=x+1}{(x:=0;\pevolve{\D{y}=-2})}}{\,x\geq y} \lbisubjunct \dbox{x:=x+1}{\,x\geq0} \land \dbox{x:=0;\pevolve{\D{y}=-2}}{\,x\geq y}}
\end{sequentdeduction}

\noindent
With \(\usall=(x,y)\) and
\(\usubstlist{\usubstmod{a}{\pchoice{x:=x+1}{y:=0}},\usubstmod{b}{\pevolve{\D{y}=-1}},\usubstmod{p(\usall)}{x\geq y}}\), \irref{US} yields:
\begin{sequentdeduction}
\hspace{-0.5cm}
  \linfer[US]
    {\dbox{a;b}{p(\usall)} \lbisubjunct \dbox{a}{\dbox{b}{p(\usall)}}}
  {\dbox{(\pchoice{x:=x+1}{y:=0});\pevolve{\D{y}=-1}}{\,x\geq y} \lbisubjunct \dbox{\pchoice{x:=x+1}{y:=0}}{\dbox{\pevolve{\D{y}=-1}}{\,x\geq y}}}
\end{sequentdeduction}

Not all axioms fit to the uniform substitution framework, though.
The Barcan schema was used in a completeness proof for the Hilbert-type calculus for differential dynamic logic \cite{DBLP:conf/lics/Platzer12b} (but not in the completeness proof for its sequent calculus \cite{DBLP:journals/jar/Platzer08}):
\[
      \cinferenceRule[B|B]{Barcan$\dbox{}{}\forall{}$} %
      {\linferenceRule[impl]
        {\lforall{x}{\dbox{\alpha}{p(x)}}}
        {\dbox{\alpha}{\lforall{x}{p(x)}}}
      }{\m{x\not\in\alpha}}
\]
Axiom schema \irref{B} is unsound without the restriction \(x\not\in\alpha\), though, so that the following formula, which cannot enforce $x\not\in a$, would be an unsound axiom
\begin{equation}
{\lforall{x}{\dbox{a}{p(x)}}\limply{\dbox{a}{\lforall{x}{p(x)}}}}
\label{eq:unsound-B-attempt}
\end{equation}
Indeed, the effect of program constant $a$ might depend on the value of $x$ or it might write to $x$.
In \rref{eq:unsound-B-attempt}, $x$ cannot be written by $a$ without violating soundness:
\[
\linfer[unsound]
  {\lforall{x}{\dbox{a}{p(x)}}\limply{\dbox{a}{\lforall{x}{p(x)}}}}
  {\lforall{x}{\dbox{x:=0}{\,x\geq0}}\limply{\dbox{x:=0}{\lforall{x}{(x\geq0)}}}}
\qquad
\usubstlist{\usubstmod{a}{x:=0},\usubstmod{p(\usarg)}{\usarg\geq0}}
\]
nor can $x$ be read by $a$ in \rref{eq:unsound-B-attempt} without violating soundness:
\[
\linfer[unsound]
  {\lforall{x}{\dbox{a}{p(x)}}\limply{\dbox{a}{\lforall{x}{p(x)}}}}
  {\lforall{x}{\dbox{\ptest{(y=x^2)}}{\,y=x^2}}\limply{\dbox{\ptest{(y=x^2)}}{\lforall{x}{y=x^2}}}}
\qquad
\usubstlist{\usubstmod{a}{\ptest{(y=x^2)}},\usubstmod{p(\usarg)}{y=\usarg^2}}
\]

Thus, the completeness proof for differential dynamic logic from prior work \cite{DBLP:conf/lics/Platzer12b} does not carry over.
A more general completeness result for differential game logic \cite{DBLP:journals/tocl/Platzer15} implies, however, that Barcan schema \irref{B} is unnecessary for completeness.

\section{Differential Equations and Differential Axioms} \label{sec:differential}

\providecommand{\sol}{x}%
\providecommand{\solf}{\sol}%
\providecommand{\solutionfor}[2][]{{\solf}_{#1}}%

\rref{sec:dL-axioms} leverages uniform substitutions to obtain a finite list of axioms without side-conditions.
They lack axioms for differential equations, though.
Classical calculi for \dL have axiom schema \irref{evolveb} from p.~\pageref{ir:evolveb} for replacing differential equations with time quantifiers and discrete assignments for their solutions.
In addition to being limited to simple solvable differential equations, such axiom schemata have quite nontrivial soundness-critical side conditions.

This section leverages \irref{US} and the new differential forms in \dL to obtain a logically internalized version of differential invariants and related proof rules for differential equations \cite{DBLP:journals/logcom/Platzer10,DBLP:journals/lmcs/Platzer12} as axioms (without schema variables or side-conditions).
These axioms can prove properties of more general ``unsolvable'' differential equations. They can also prove all properties of differential equations that can be proved with solutions \cite{DBLP:journals/lmcs/Platzer12} while guaranteeing correctness of the solution as part of the proof.

\subsection{Differentials: Invariants, Cuts, Effects, and Ghosts} \label{sec:diffind}

Figure~\ref{fig:dL-ODE} shows axioms for proving properties of differential equations (\irref{DW}--\irref{DS}), and differential axioms for differentials (\irref{Dplus+Dtimes+Dcompose}) which are equations of differentials.
Axiom \irref{Dvar} identifying \(\der{x}=\D{x}\) for variables $x\in\allvars$ and axiom \irref{Dconst} for functions $f$ and number literals of arity 0 are used implicitly to save space.
Some axioms use reverse implication notation \(\phi\lylpmi\psi\) instead of the equivalent \(\psi\limply\phi\) for emphasis.%
\begin{figure}[tb]
  \begin{calculuscollections}{\columnwidth}
  \renewcommand*{\irrulename}[1]{\text{#1}}%
    \begin{calculus}
      \cinferenceRule[DW|DW]{differential evolution domain} %
      {\axkey{\dbox{\pevolvein{\D{x}=f(x)}{q(x)}}{q(x)}}}
      {}
      \cinferenceRule[DC|DC]{differential cut} %
      {\linferenceRule[lpmi]
        {\big(\axkey{\dbox{\pevolvein{\D{x}=f(x)}{q(x)}}{p(x)}} \lbisubjunct \dbox{\pevolvein{\D{x}=f(x)}{q(x)\land r(x)}}{p(x)}\big)}
        {\dbox{\pevolvein{\D{x}=f(x)}{q(x)}}{r(x)}}
      }
      {}%
      \cinferenceRule[DE|DE]{differential effect} %
      {\linferenceRule[viuqe]
        {\axkey{\dbox{\pevolvein{\D{x}=f(x)}{q(x)}}{p(x,\D{x})}}}
        {\dbox{\pevolvein{\D{x}=f(x)}{q(x)}}{\dbox{\Dupdate{\Dumod{\D{x}}{f(x)}}}{p(x,\D{x})}}}
      }
      {}%
      \cinferenceRule[DI|DI]{differential invariant} %
      {\linferenceRule[lpmi]
        {\big(\axkey{\dbox{\pevolvein{\D{x}=f(x)}{q(x)}}{p(x)}} \lbisubjunct \dbox{\ptest{q(x)}}{p(x)}\big)}
        {\big(q(x)\limply \dbox{\pevolvein{\D{x}=f(x)}{q(x)}}{\der{p(x)}}\big)}
      }
      {}%
      \cinferenceRule[DG|DG]{differential ghost variables} %
      {\linferenceRule[viuqe]
        {\axkey{\dbox{\pevolvein{\D{x}=f(x)}{q(x)}}{p(x)}}}
        {\lexists{y}{\dbox{\pevolvein{\D{x}=f(x)\syssep\D{y}=a(x)y+b(x)}{q(x)}}{p(x)}}}
      }
      {}%
      \cinferenceRule[DS|DS]{(constant) differential equation solution} %
      {\linferenceRule[viuqe]
        {\axkey{\dbox{\pevolvein{\D{x}=f}{q(x)}}{p(x)}}}
        {\lforall{t{\geq}0}{\big((\lforall{0{\leq}s{\leq}t}{q(x+f\itimes s)}) \limply \dbox{\pupdate{\pumod{x}{x+f\itimes t}}}{p(x)}\big)}}
      }
      {}%
      \cinferenceRule[Dconst|$c'$]{derive constant}
      {\linferenceRule[eq]
        {0}
        {\axkey{\der{f}}}
      }
      {}%
      \cinferenceRule[Dvar|$x'$]{derive variable}
      {\linferenceRule[eq]
        {\D{x}}
        {\axkey{\der{x}}}
      }
      {}%
      \cinferenceRule[Dplus|$+'$]{derive sum}
      {\linferenceRule[eq]
        {\der{f(\usall)}+\der{g(\usall)}}
        {\axkey{\der{f(\usall)+g(\usall)}}}
      }
      {}
      \cinferenceRule[Dtimes|$\cdot'$]{derive product}
      {\linferenceRule[eq]
        {\der{f(\usall)}\cdot g(\usall)+f(\usall)\cdot\der{g(\usall)}}
        {\axkey{\der{f(\usall)\cdot g(\usall)}}}
      }
      {}
      \cinferenceRule[Dcompose|$\compose'$]{derive composition}
      {
        \dbox{\pupdate{\pumod{y}{g(x)}}}{\dbox{\Dupdate{\Dumod{\D{y}}{1}}}
        {\big( \axkey{\der{f(g(x))}} = \der{f(y)}\stimes\der{g(x)}\big)}}
      }
      {}%
   \end{calculus}%
\end{calculuscollections}%
  \caption{Differential equation axioms and differential axioms}
  \label{fig:dL-ODE}
\end{figure}

\emph{Differential weakening} axiom \irref{DW} internalizes that differential equations never leave their evolution domain $q(x)$.
The evolution domain $q(x)$ holds after all evolutions of \(\pevolvein{\D{x}=f(x)}{q(x)}\), because differential equations cannot leave their evolution domains.
\irref{DW} derives\footnote{%
The implication \(\dbox{\pevolvein{\D{x}=f(x)}{q(x)}}{(q(x)\limply p(x))} \limply \dbox{\pevolvein{\D{x}=f(x)}{q(x)}}{p(x)}\) derives by \irref{K} from \irref{DW}.
The converse implication
\(\dbox{\pevolvein{\D{x}=f(x)}{q(x)}}{p(x)} \limply \dbox{\pevolvein{\D{x}=f(x)}{q(x)}}{(q(x)\limply p(x))}\)
derives by \irref{K} since \irref{G} derives \(\dbox{\pevolvein{\D{x}=f(x)}{q(x)}}{\big(p(x)\limply(q(x)\limply p(x))\big)}\)
from the tautology \(p(x)\limply(q(x)\limply p(x))\).
}
\({\dbox{\pevolvein{\D{x}=f(x)}{q(x)}}{p(x)}} \lbisubjunct
        {\dbox{\pevolvein{\D{x}=f(x)}{q(x)}}{(q(x)\limply p(x))}} \irlabel{diffweaken|DW}\),
which allows to export the evolution domain to the postcondition and is also called \irref{diffweaken}.
Its (right) assumption is best proved by \irref{G} yielding premise \(q(x)\limply p(x)\).
The \emph{differential cut} axiom \irref{DC} is a cut for differential equations.
It internalizes that differential equations always staying in $r(x)$ also always stay in $p(x)$ iff $p(x)$ always holds after the differential equation that is restricted to the smaller evolution domain \(\pevolvein{}{q(x)\land r(x)}\).
\irref{DC} is a differential variant of modal modus ponens axiom \irref{K}.

\emph{Differential effect} axiom \irref{DE} internalizes that the effect on differential symbols along a differential equation is a differential assignment assigning the right-hand side $f(x)$ to the left-hand side $\D{x}$.
The differential assignment \(\Dupdate{\Dumod{\D{x}}{f(x)}}\) in \irref{DE} instantaneously mimics the (continuous) effect that the differential equation \(\pevolvein{\D{x}=f(x)}{q(x)}\) has on $\D{x}$, thereby selecting the appropriate vector field for subsequent differentials.
Axiom \irref{DI} internalizes \emph{differential invariants} \cite{DBLP:journals/logcom/Platzer10},
i.e.\ that $p(x)$ holds always after a differential equation \(\pevolvein{\D{x}=f(x)}{q(x)}\) iff it holds after $\ptest{q(x)}$, provided its differential $\der{p(x)}$ always holds after the differential equation \(\pevolvein{\D{x}=f(x)}{q(x)}\).
This axiom reduces future truth to present truth when the truth of $p(x)$ does not change along the differential equation because $\der{p(x)}$ holds all along.
The differential equation also vacuously stays in $p(x)$ if it starts outside $q(x)$, since it is stuck then.
The assumption of \irref{DI} is best proved by \irref{DE} to select the appropriate vector field \(\D{x}=f(x)\) for the differential $\der{p(x)}$
and a subsequent \irref{diffweaken+G} to make the evolution domain constraint $q(x)$ available as an assumption when showing $\der{p(x)}$.
The condition \(\dbox{\ptest{q(x)}}{p(x)}\) in \irref{DI} is equivalent to \(q(x)\limply p(x)\) by axiom \irref{testb}.
While a general account of $\der{p(x)}$ is possible \cite{DBLP:journals/corr/Platzer15:dGI}, this article focuses on atomic postconditions with the equivalences \(\der{\theta\geq\eta} \mequiv \der{\theta>\eta} \mequiv \der{\theta}\geq\der{\eta}\)
and \(\der{\theta=\eta} \mequiv \der{\theta\neq\eta} \mequiv \der{\theta}=\der{\eta}\), etc.\ for \irref{DI} axioms.
Note \(\der{\theta\neq\eta}\) cannot be \(\der{\theta}\neq\der{\eta}\), because different rates of change from different initial values do not imply the values would remain different.
Conjunctions can be handled separately by
\(\dbox{\alpha}{(p(\usall)\land q(\usall))} \lbisubjunct  \dbox{\alpha}{p(\usall)} \land \dbox{\alpha}{q(\usall)}\)
which derives from \irref{K}.
Disjunctions split into separate disjuncts, which is equivalent to classical differential invariants \cite{DBLP:journals/logcom/Platzer10} but easier.
Axiom \irref{DG} internalizes \emph{differential ghosts} \cite{DBLP:journals/lmcs/Platzer12},
i.e.\ that additional differential equations can be added whose solutions exist long enough, which can enable new invariants that are not otherwise provable \cite{DBLP:journals/lmcs/Platzer12}.
Axiom \irref{DS} solves constant differential equations, and, as \rref{sec:example-proofs} will demonstrate, more complex solvable differential equations with the help of \irref{DG+DC+DI}.
Vectorial generalizations to systems of differential equations are possible for the axioms in \rref{fig:dL-ODE}.

The differential axioms for differentials (\irref{Dplus+Dtimes+Dcompose+Dconst+Dvar}) axiomatize differentials of polynomials.
They are related to corresponding rules for time-derivatives, except that those would be ill-defined in a local state, so it is crucial to work with differentials that have a local semantics in individual states.
Uniform substitutions correctly maintain that $y$ does not occur in replacements for $a(x),b(x)$ for \irref{DG} and that $x$ does not occur in replacements for $f$ in \irref{DS}, which are both soundness-critical.
Occurrences of $x$ in replacements for $f$ are acceptable when using axiom \irref{Dassignb} on \({\dbox{\Dupdate{\Dumod{\D{x}}{f}}}{p(\D{x})}} \lbisubjunct {p(f)}\).

Most axioms in \rref{fig:dL} and~\ref{fig:dL-ODE} are independent, because there is exactly one axiom per operator.
Exceptions in \rref{fig:dL} are \irref{K+I+V}, but there is a complete calculus without \irref{iterateb+V} \cite{DBLP:conf/lics/Platzer12b} and one without \irref{G+K+I+V} that uses two extra rules instead \cite{DBLP:journals/tocl/Platzer15}.
The congruence rules \irref{CQ+CE} are redundant and can be proved on a per-instance basis as well.
Axiom \irref{DW} is the only one that can use the evolution domain,
axiom \irref{DC} the only one that can change the evolution domain, and axiom \irref{DG} the only one that can change differential equations.
Axiom \irref{DE} is the only one that can use the right-hand side of the differential equation.
Axiom \irref{DI} is the only axiom that relates truth of a postcondition after a differential equation to truth at the initial state.
Finally, axiom \irref{DS} is needed for proving diamond properties of differential equations, because it is the only one (besides the limited \irref{DW}) that does not reduce a property of a differential equation to another property of a differential equation and, thus, the only axiom that ultimately proves them without the help of \irref{G+V+K}, which are not sound for $\ddiamond{\alpha}{}$.

\subsection{Example Proofs} \label{sec:example-proofs}

This section illustrates how the uniform substitution calculus for \dL can be used to realize a number of different reasoning techniques from the same proof primitives.
While the same flexibility enables these different techniques also for proofs of hybrid systems, the following examples focus on differential equations to additionally illustrate how the differential equation axioms in \rref{fig:dL-ODE} are meant to be combined.

\begin{example}[Contextual equivalence proof] \label{ex:diffind-CE-proof}
The following proof proves a property of a differential equation using differential invariants without having to solve that differential equation.
One use of rule \irref{US} is shown explicitly, other uses of \irref{US} are similar to obtain and use the other axiom instances.
\irref{CE} is used together with \irref{MP}.
\let\orgcdot\cdot%
\def\cdot{{\orgcdot}}%
\begin{sequentdeduction}[array]
\linfer[DI]
{\linfer[DE]
 {\linfer[CE]%
  {\linfer[G]
    {\linfer[Dassignb]
      {\linfer[qear]
        {\lclose}
        {\lsequent{}{x^3\cdot x + x\cdot x^3\geq0}}
      }%
      {\lsequent{}{\dbox{\Dupdate{\Dumod{\D{x}}{x^3}}}{\D{x}\cdot x+x\cdot\D{x}\geq0}}}
    }%
    {\lsequent{}{\dbox{\pevolve{\D{x}=x^3}}{\dbox{\Dupdate{\Dumod{\D{x}}{x^3}}}{\D{x}\cdot x+x\cdot\D{x}\geq0}}}}
    !
    \linfer[Dconst]%
    {\linfer[CQ] %
      {\linfer[Dvar]
        {\linfer[US]
          {\linfer[Dtimes]
            {\lclose}
            {\lseqalign{\der{f(\usall)\cdot g(\usall)}} {= \der{f(\usall)}\cdot g(\usall) + f(\usall)\cdot\der{g(\usall)}}}
          }
        {\lseqalign{\der{x\cdot x}} {= \der{x}\cdot x + x\cdot\der{x}}}
      }
      {\lseqalign{\der{x\cdot x}} {= \D{x}\cdot x + x\cdot\D{x}}}
      }
      {\lseqalign{\der{x\cdot x}\geq0}{\lbisubjunct\D{x}\cdot x+x\cdot\D{x}\geq0}}
    }
    {\lseqalign{\der{x\cdot x\geq1}}{\lbisubjunct\D{x}\cdot x+x\cdot\D{x}\geq0}}
  }%
  {\lsequent{}{\dbox{\pevolve{\D{x}=x^3}}{\dbox{\Dupdate{\Dumod{\D{x}}{x^3}}}{\der{x\cdot x\geq1}}}}}
 }%
  {\lsequent{}{\dbox{\pevolve{\D{x}=x^3}}{\der{x\cdot x\geq1}}}}
}%
{\lsequent{x\cdot x\geq1} {\dbox{\pevolve{\D{x}=x^3}}{x\cdot x\geq1}}}
\end{sequentdeduction}
Previous calculi \cite{DBLP:journals/logcom/Platzer10,DBLP:journals/lmcs/Platzer12} collapse this proof into a single proof step with complicated built-in operator implementations that silently perform the same reasoning in a non-transparent way.
The approach presented here combines separate axioms to achieve the same effect in a modular way, with axioms of individual responsibilities internalizing separate logical reasoning principles in \emph{differential-form} \dL.
Tactics combining the axioms as indicated make the axiomatic way equally convenient.
Clever proof structuring, cuts or \irref{MP} uses enable proofs in which the main argument remains as fast \cite{DBLP:journals/logcom/Platzer10,DBLP:journals/lmcs/Platzer12} while the additional premises subsequently check soundness.
Inferences in context such as those portrayed in \irref{CE+CQ} are impossible in sequent calculus \cite{DBLP:journals/jar/Platzer08}.
\end{example}

\begin{example}[Flat proof] \label{ex:diffind-direct-proof}
Rules \irref{CQ+CE} simplify the proof in \rref{ex:diffind-CE-proof} substantially but are not needed because a proof without contextual equivalence is possible:%
\renewcommand{\linferPremissSeparation}{~~~~}%
\begin{sequentdeduction}%
\linfer[MP]
  {\linfer
    {\lclose}
    {{\dots}\limply(\der{x{\cdot} x}{\geq}0 \lbisubjunct \D{x}{\cdot} x {+} x{\cdot}\D{x}{\geq}0)}
  &\linfer[Dvar]
    {\linfer[US]
      {\linfer[Dtimes]
        {\lclose}
        {\der{f(\usall)\cdot g(\usall)} = \der{f(\usall)}\cdot g(\usall) + f(\usall)\cdot\der{g(\usall)}}
      }
      {\der{x\cdot x} = \der{x}\cdot x + x\cdot\der{x}}
    }
    {\der{x\cdot x} = \D{x}\cdot x + x\cdot\D{x}}
  }
  {\linfer[G]
    {\der{x\cdot x}\geq0 \lbisubjunct \D{x}\cdot x + x\cdot\D{x}\geq0}
    {\linfer[K+K]
      {\dbox{\Dupdate{\Dumod{\D{x}}{x^3}}}{(\der{x\cdot x}\geq0 \lbisubjunct \D{x}\cdot x + x\cdot\D{x}\geq0)}}
      {\dbox{\Dupdate{\Dumod{\D{x}}{x^3}}}{\der{x\cdot x}\geq0} \lbisubjunct \dbox{\Dupdate{\Dumod{\D{x}}{x^3}}}{\D{x}\cdot x + x\cdot\D{x}\geq0}}
    }
}
\end{sequentdeduction}%

\noindent
\hfill\begin{minipage}{8.5cm}
\vspace{-0.5\baselineskip}
\begin{sequentdeduction}[array]
\linfer[DI]
{\linfer[DE]
 {\linfer[G]
  {\linfer[MP]
    {\text{see above}
    !\linfer[Dassignb]
      {\linfer[qear]
        {\lclose}
        {\lsequent{}{x^3\cdot x + x\cdot x^3\geq0}}
      }%
      {\lsequent{}{\dbox{\Dupdate{\Dumod{\D{x}}{x^3}}}{\D{x}\cdot x + x\cdot\D{x}\geq0}}}
    }%
    {\lsequent{}{\dbox{\Dupdate{\Dumod{\D{x}}{x^3}}}{\der{x\cdot x}\geq0}}}
  }%
  {\lsequent{}{\dbox{\pevolve{\D{x}=x^3}}{\dbox{\Dupdate{\Dumod{\D{x}}{x^3}}}{\der{x\cdot x}\geq0}}}}
 }%
  {\lsequent{}{\dbox{\pevolve{\D{x}=x^3}}{\der{x\cdot x}\geq0}}}
}%
{\lsequent{x\cdot x\geq1} {\dbox{\pevolve{\D{x}=x^3}}{x\cdot x\geq1}}}
\end{sequentdeduction}
\end{minipage}%
\end{example}%

\newcommand{\mydiffcond}[1][x,\D{x}]{j(#1)}%
\begin{example}[Parametric proof] \label{ex:diffind-free-parametric-proof}
The proofs in \rref{ex:diffind-CE-proof} and \ref{ex:diffind-direct-proof} use (implicit) cuts with equivalences that predict the outcome of the right premise, which is conceptually simple but inconvenient for proof search.
More constructively, a direct proof can use a free function symbol $\mydiffcond$ to obtain a straightforward parametric proof, instead:
\begin{sequentdeduction}[array]
\linfer[DI]
{\linfer[DE]
 {\linfer[CE]%
  {\linfer[G]
    {\linfer[Dassignb]
        {\lsequent{}{\mydiffcond[x,x^3]\geq0}}
      {\lsequent{}{\dbox{\Dupdate{\Dumod{\D{x}}{x^3}}}{\mydiffcond\geq0}}}
    }%
    {\lsequent{}{\dbox{\pevolve{\D{x}=x^3}}{\dbox{\Dupdate{\Dumod{\D{x}}{x^3}}}{\mydiffcond\geq0}}}}
    !
    \linfer%
    {\linfer[CQ] %
      {\lseqalign{\der{x\cdot x}} {= \mydiffcond}}
      {\lseqalign{\der{x\cdot x}\geq0}{\lbisubjunct\mydiffcond\geq0}}
    }
    {\lseqalign{\der{x\cdot x\geq1}}{\lbisubjunct\mydiffcond\geq0}}
  }%
  {\lsequent{}{\dbox{\pevolve{\D{x}=x^3}}{\dbox{\Dupdate{\Dumod{\D{x}}{x^3}}}{\der{x\cdot x\geq1}}}}}
 }%
  {\lsequent{}{\dbox{\pevolve{\D{x}=x^3}}{\der{x\cdot x\geq1}}}}
}%
{\lsequent{x\cdot x\geq1} {\dbox{\pevolve{\D{x}=x^3}}{x\cdot x\geq1}}}
\end{sequentdeduction}
After conducting this proof with two open premises, the free function symbol $\mydiffcond$ can be instantiated as needed by a uniform substitution (\irref{USR} from \rref{thm:usubst-rule}).
The above proof justifies the locally sound inference on the left whose two open premises and conclusions are instantiated by \irref{USR} leading to the new sound proof on the right:
\[
\linfer
{\lsequent{}{\mydiffcond[x,x^3]\geq0}
&\quad \lseqalign{\der{x\cdot x}} {= \mydiffcond}}
{\lsequent{x\cdot x\geq1} {\dbox{\pevolve{\D{x}=x^3}}{x\cdot x\geq1}}}
~~~\text{implies}~~~~\quad
\linfer[USR]
{\lsequent{}{x^3\cdot x + x\cdot x^3\geq0}
&\quad \lseqalign{\der{x\cdot x}} {= \D{x} \cdot x + x\cdot\D{x}}}
{\lsequent{x\cdot x\geq1} {\dbox{\pevolve{\D{x}=x^3}}{x\cdot x\geq1}}}
\]
After the instantiation of $\mydiffcond$ by \irref{USR}, the right proof completes as follows:
\begin{sequentdeduction}[array]
\linfer[USR]
        {\linfer[qear]
          {\lclose}
          {\lsequent{}{x^3\cdot x + x\cdot x^3\geq0}}
        !
        \linfer[Dvar]
        {\linfer[US]
          {\linfer[Dtimes]
            {\lclose}
            {\lseqalign{\der{f(\usall)\cdot g(\usall)}} {= \der{f(\usall)}\cdot g(\usall) + f(\usall)\cdot\der{g(\usall)}}}
          }%
        {\lseqalign{\der{x\cdot x}} {= \der{x}\cdot x + x\cdot\der{x}}}
        }%
      {\lseqalign{\der{x\cdot x}} {= \D{x} \cdot x + x\cdot\D{x}}}
      }
{\lsequent{x\cdot x\geq1} {\dbox{\pevolve{\D{x}=x^3}}{x\cdot x\geq1}}}
\end{sequentdeduction}
This technique helps invariant search, where a free predicate symbol $p(\usall)$ is instantiated lazily by \irref{USR} once all conditions become clear.
This reduction saves considerable proof effort compared to eager invariant instantiation in sequent calculi \cite{DBLP:journals/jar/Platzer08}.
\end{example}

\begin{example}[Forward computation proof] \label{ex:diffind-forward-proof}
The proof in \rref{ex:diffind-free-parametric-proof} involves less search than the proofs of the same formula in \rref{ex:diffind-CE-proof} and \ref{ex:diffind-direct-proof}.
But it still ultimately requires foresight to identify the appropriate instantiation of $\mydiffcond$ for which the proof closes.
For invariant search, such proof search is essentially unavoidable \cite{DBLP:conf/lics/Platzer12b} even if the technique in \rref{ex:diffind-free-parametric-proof} maximally postpones the search.  

When used from left to right, the differential axioms \irref{Dconst+Dvar+Dplus+Dtimes+Dcompose} compute deterministically and always simplify terms by pushing differential operators inside.
For example, all backwards proof search in the right branch of the last proof of \rref{ex:diffind-free-parametric-proof} can be replaced by a deterministic forward computation proof starting from reflexivity \(\der{x\cdot x} =\der{x\cdot x}\) and drawing on axiom instances (used in a term context via \irref{CT}) as needed in a forward proof, until the desired output shape is identified:
\[
\left\downarrow
\begin{minipage}{3.6cm}
\vspace{-6pt}
\begin{sequentdeduction}[array]
\linfer[Dvar]
{\linfer[Dtimes]
  {\linfer[qear]
    {\lclose}
    {\lseqalign{\der{x\cdot x}} {=\der{x\cdot x}}}
  }
  {\lseqalign{\der{x\cdot x}} {=\der{x}\cdot x + x \cdot \der{x}}}
}%
{\lseqalign{\der{x\cdot x}} {=\D{x} \cdot x + x \cdot \D{x}}}
\end{sequentdeduction}
\end{minipage}
\right.
\]
Efficient proof search combines this forward computation proof technique with the backward proof search from \rref{ex:diffind-free-parametric-proof} with advantages similar to other combinations of computation and axiomatic reasoning \cite{DBLP:journals/jar/DowekHK03}.
Even the remaining positions where axioms still match can be precomputed as a simple function of the axiom that has been applied, e.g., from its fixed pattern of occurrences of differential operators.
\end{example}

\begin{example}[Axiomatic differential equation solver] \label{ex:ode-solver}
Axiomatic equivalence proofs for solving differential equations involve \irref{DG} for introducing a time variable $t$, \irref{DC} to cut the solutions in, \irref{DW} to export the solution to the postcondition, inverse \irref{DC} to remove the evolution domain constraints again, inverse \irref{DG} (or the universal strengthening of \irref{DG} with $\forall{y}$ instead of $\exists{y}$ from \rref{thm:dL-sound}) to remove the original differential equations, and finally \irref{DS} to solve the differential equation for time:
\def\prem{\phi}%
\begin{sequentdeduction}[array]
\linfer[DG]
 {\linfer%
   {\linfer[DC]
     {\linfer[DC]
       {\linfer[diffweaken]
         {\linfer[G+K]%
           {\linfer[DC]
             {\linfer[DC]
               {\linfer[DG]
                 {\linfer[DG]
                   {\linfer[DS]
                     {\linfer[assignb]
                       {\linfer[qear]
                         {\lclose}
                         {\lsequent{\prem} {\lforall{s{\geq}0}{(x_0+\frac{a}{2}s^2+v_0s\geq0)}}}
                       }%
                       {\lsequent{\prem} {\lforall{s{\geq}0}{\dbox{\pupdate{\pumod{t}{0+1s}}}{x_0+\frac{a}{2}t^2+v_0t\geq0}}}}
                     }%
                     {\lsequent{\prem} {\dbox{\pevolve{\D{t}=1}}{x_0+\frac{a}{2}t^2+v_0t\geq0}}}
                   }%
                   {\lsequent{\prem} {\dbox{\pevolve{\D{v}=a\syssep\D{t}=1}}{x_0+\frac{a}{2}t^2+v_0t\geq0}}}
                 }%
                 {\lsequent{\prem} {\dbox{\pevolve{\D{x}=v\syssep\D{v}=a\syssep\D{t}=1}}{x_0+\frac{a}{2}t^2+v_0t\geq0}}
                 \hfill\triangleright}
               }%
               {\lsequent{\prem} {\dbox{\pevolvein{\D{x}=v\syssep\D{v}=a\syssep\D{t}=1}{v=v_0+at}}{x_0+\frac{a}{2}t^2+v_0t\geq0}}
               \hfill\triangleright}
             }%
             {\lsequent{\prem} {\dbox{\pevolvein{\D{x}=v\syssep\D{v}=a\syssep\D{t}=1}{v=v_0+at\land x=x_0+\frac{a}{2}t^2+v_0t}}{x_0+\frac{a}{2}t^2+v_0t\geq0}}} 
           }%
           {\lsequent{\prem} {\dbox{\pevolvein{\D{x}=v\syssep\D{v}=a\syssep\D{t}=1}{v=v_0{+}at\land x=x_0{+}\frac{a}{2}t^2{+}v_0t}}{(x{=}x_0{+}\frac{a}{2}t^2{+}v_0t\limply x{\geq}0)}}}
         }%
         {\lsequent{\prem} {\dbox{\pevolvein{\D{x}=v\syssep\D{v}=a\syssep\D{t}=1}{v=v_0+at\land x=x_0+\frac{a}{2}t^2+v_0t}}{x\geq0}}
         \hfill\triangleright} 
       }%
       {\lsequent{\prem} {\dbox{\pevolvein{\D{x}=v\syssep\D{v}=a\syssep\D{t}=1}{v=v_0+at}}{x\geq0}}
       \hfill\triangleright} 
     }%
     {\lsequent{\prem} {\dbox{\pevolve{\D{x}=v\syssep\D{v}=a\syssep\D{t}=1}}{x\geq0}}} 
   }%
   {\lsequent{\prem} {\lexists{t}{\dbox{\pevolve{\D{x}=v\syssep\D{v}=a\syssep\D{t}=1}}{x\geq0}}}}
 }%
 {\lsequent{\prem} {\dbox{\pevolve{\D{x}=v\syssep\D{v}=a}}{x\geq0}}} 
\end{sequentdeduction}
where $\prem$ is \({a\geq0\land v=v_0\geq0 \land x=x_0\geq0}\). %
The existential quantifier for $t$ is instantiated by $0$ (suppressed in the proof for readability reasons).
The 4 uses of \irref{DC} lead to 2 additional premises (marked by $\triangleright$) proving that \(v=v_0+at\) and then \(x=x_0+\frac{a}{2}t^2+v_0t\) are differential invariants (using \irref{DI+DE+diffweaken}).
Shortcuts using only \irref{diffweaken} instead are possible.
But the elaborate proof above generalizes to $\ddiamond{}{}$ because it is an equivalence proof.
The additional premises for \irref{DC} with \(v=v_0+at\) prove as follows:
\begin{sequentdeduction}[array]
\linfer[DI]
{\linfer[DE]
 {\linfer[G]
   {\linfer[CE]%
    {\linfer[Dassignb]
      {\linfer[qear]
        {\lclose}
        {\lsequent{}{a=0+a\cdot1}}
      }%
      {\lsequent{} {\dbox{\Dupdate{\Dumod{\D{v}}{a}}}{\dbox{\Dupdate{\Dumod{\D{t}}{1}}}{\D{v}=0+a\D{t}}}}}
    !
    \linfer%
    {\linfer[CQ] %
      {\linfer[Dtimes]%
        {\linfer[US]
          {\linfer[Dplus]
            {\lclose}
            {\lseqalign{\der{f(\usall)+ g(\usall)}} {= \der{f(\usall)} + \der{g(\usall)}}}
          }
        {\lseqalign{\der{v_0+at}} {= \der{v_0}+\der{a t}}}
      }
      {\lseqalign{\der{v_0+at}} {= 0+a(\D{t})}}
      }
      {\lseqalign{\D{v}=\der{v_0+at}}{\lbisubjunct\D{v}=0+a\D{t}}}
    }
    {\lseqalign{\der{v=v_0+at}}{\lbisubjunct\D{v}=0+a\D{t}}}
    }%
    {\lsequent{} {\dbox{\Dupdate{\Dumod{\D{v}}{a}}}{\dbox{\Dupdate{\Dumod{\D{t}}{1}}}{\der{v=v_0+at}}}}}
  }%
  {\lsequent{} {\dbox{\pevolve{\D{x}=v\syssep\D{v}=a\syssep\D{t}=1}}{\dbox{\Dupdate{\Dumod{\D{v}}{a}}}{\dbox{\Dupdate{\Dumod{\D{t}}{1}}}{\der{v=v_0+at}}}}}}
 }%
  {\lsequent{} {\dbox{\pevolve{\D{x}=v\syssep\D{v}=a\syssep\D{t}=1}}{\der{v=v_0+at}}}}
}%
{\lsequent{\prem} {\dbox{\pevolve{\D{x}=v\syssep\D{v}=a\syssep\D{t}=1}}{v=v_0+at}}}
\end{sequentdeduction}
After that, the additional premises for \irref{DC} with \(x=x_0+\frac{a}{2}t^2+v_0t\) prove as follows:
{\footnotesize
\begin{sequentdeduction}[array]
\hspace*{-2cm}
\linfer[DI]
{\linfer[DE]
 {\linfer[diffweaken]
 {\linfer[G]
   {\lsequent{} {\hspace*{-0.5cm}\begin{minipage}{11.5cm}\begin{sequentdeduction}[array]
   \linfer[CE]%
    {\linfer[Dassignb]
      {\linfer[qear]
        {\lclose}
        {\lsequent{}{{v=v_0+at} \limply v=at\cdot1+v_0\cdot1}}
      }%
      {\lsequent{} {{v=v_0+at} \limply \dbox{\Dupdate{\Dumod{\D{x}}{v}}}{\dbox{\Dupdate{\Dumod{\D{t}}{1}}}{\D{x}=at\D{t}+v_0\D{t}}}}}
    !
    \linfer%
    {\linfer[CQ] %
      {\linfer[Dplus+Dtimes]%
        {\linfer[qear]
          {\lclose}
          {\lseqalign{2\frac{a}{2}t\D{t}+v_0\D{t}} {= at\D{t}+v_0\D{t}}}
        }
      {\lseqalign{\der{x_0+\frac{a}{2}t^2+v_0t}} {= at\D{t}+v_0\D{t}}}
      }
      {\lseqalign{\D{x}=\der{x_0+\frac{a}{2}t^2+v_0t}}{\lbisubjunct\D{x}=at\D{t}+v_0\D{t}}}
    }
    {\lseqalign{\der{x=x_0+\frac{a}{2}t^2+v_0t}}{\lbisubjunct\D{x}=at\D{t}+v_0\D{t}}}
    }%
    {\lsequent{} {{v=v_0+at} \limply \dbox{\Dupdate{\Dumod{\D{x}}{v}}}{\dbox{\Dupdate{\Dumod{\D{t}}{1}}}{\der{x=x_0+\frac{a}{2}t^2+v_0t}}}}}
    \end{sequentdeduction}\end{minipage}}
  }%
  {\lsequent{} {\dbox{\pevolvein{\D{x}=v\syssep\D{v}=a\syssep\D{t}=1}{v=v_0+at}}{({v=v_0+at}\limply\dbox{\Dupdate{\Dumod{\D{x}}{v}}}{\dbox{\Dupdate{\Dumod{\D{t}}{1}}}{\der{x=x_0+\frac{a}{2}t^2+v_0t}}})}}}
  }%
  {\lsequent{} {\dbox{\pevolvein{\D{x}=v\syssep\D{v}=a\syssep\D{t}=1}{v=v_0+at}}{\dbox{\Dupdate{\Dumod{\D{x}}{v}}}{\dbox{\Dupdate{\Dumod{\D{t}}{1}}}{\der{x=x_0+\frac{a}{2}t^2+v_0t}}}}}}
 }%
  {\lsequent{} {\dbox{\pevolvein{\D{x}=v\syssep\D{v}=a\syssep\D{t}=1}{v=v_0+at}}{\der{x=x_0+\frac{a}{2}t^2+v_0t}}}}
}%
{\lsequent{\prem} {\dbox{\pevolvein{\D{x}=v\syssep\D{v}=a\syssep\D{t}=1}{v=v_0+at}}{x=x_0+\frac{a}{2}t^2+v_0t}}}
\end{sequentdeduction}
}%
This axiomatic differential equation solving technique is not limited to differential equation systems that can be solved in full, but also works when only part of the differential equations have definable solutions.
Contrast this constructive formal proof with the unverified use of a differential equation solver in axiom schema \irref{evolveb} from p.~\pageref{ir:evolveb}.
\end{example}%

\subsection{Differential Substitution Lemmas}

In similar ways how the uniform substitution lemmas are the key ingredients that relate syntactic and semantic substitution for the soundness of proof rule \irref{US}, this section proves the key ingredients relating syntax and semantics of differentials that will be used for the soundness proofs of the differential axioms.
Differentials $\der{\eta}$ have a local semantics in isolated states, which is crucial for well-definedness.
The \irref{DI} axiom relates truth along a differential equation to initial truth with truth of differentials along a differential equation.
The key insight for its soundness is that the analytic time-derivative of the value of a term $\eta$ along any differential equation \(\pevolvein{\D{x}=\genDE{x}}{\ivr}\) agrees with the values of its differential $\der{\eta}$ along that differential equation.
Recall from \rref{def:HP-transition} that \m{\imodels{\If}{\D{x}=\genDE{x}\land\ivr}} indicates that the function $\iget[flow]{\If}$ solves the differential equation \(\pevolvein{\D{x}=\genDE{x}}{\ivr}\) in interpretation $\iget[const]{\I}$, of which the only important part for the next lemma is that it gives $\D{x}$ the value of the time-derivative of $x$ along the solution $\iget[flow]{\If}$.

\begin{lemma}[Differential] \label{lem:differentialLemma}
  If \m{\imodels{\If}{\D{x}=\genDE{x}\land\ivr}}
  holds for some solution \m{\iget[flow]{\If}:[0,r]\to\linterpretations{\Sigma}{V}} 
of any duration $r>0$,
  then for all times $0\leq\zeta\leq r$ and all terms $\eta$ with $\freevarsdef{\eta}\subseteq\{x\}$:
  \[
  \ivaluation{\Iff[\zeta]}{\der{\eta}}
  = \D[t]{\ivaluation{\Iff[t]}{\eta}} (\zeta)
  \]
\end{lemma}
\begin{proofatend}
By \rref{def:dL-valuationTerm} the left side is:
\[
\ivaluation{\Iff[\zeta]}{\der{\eta}}
= \sum_{x\in\allvars} \iget[state]{\Iff[\zeta]}(\D{x}) \itimes \Dp[x]{\ivaluation{\Idot}{\eta}}(\iget[state]{\Iff[\zeta]})
\]
By chain rule (\rref{lem:chain} in the beginning of the appendix) the right side is:%
\[
\D[t]{\ivaluation{\Iff[t]}{\eta}} (\zeta)
=
\D{(\ivaluation{\Idot}{\eta} \compose \iget[flow]{\If})}(\zeta) = (\gradient{\ivaluation{\Idot}{\eta}})\big(\iget[flow]{\If}(\zeta)\big)
\stimes \D{\iget[flow]{\If}}(\zeta)
= \sum_{x\in\allvars} \Dp[x]{\ivaluation{\Idot}{\eta}}\big(\iget[flow]{\If}(\zeta)\big) \D{\iget[flow]{\If}}(\zeta)(x)
\]
where \((\gradient{\ivaluation{\Idot}{\eta}})\big(\iget[flow]{\If}(\zeta)\big)\), the gradient \(\gradient{\ivaluation{\Idot}{\eta}}\)
of $\ivaluation{\Idot}{\eta}$ at \(\iget[flow]{\If}(\zeta)\),
is the vector of
\(\Dp[x]{\ivaluation{\Idot}{\eta}}\big(\iget[flow]{\If}(\zeta)\big)\),
which has finite support by \rref{lem:coincidence-term} so is 0 for all but finitely many variables.
Both sides, thus, agree since
\(
\iget[state]{\Iff[\zeta]}(\D{x})
= \D[t]{\iget[state]{\Iff[t]}(x)}(\zeta)
= \D{\iget[flow]{\If}}(\zeta)(x)
\)
by \rref{def:HP-transition} for all $x\in\freevarsdef{\eta}$.
The same proof works for vectorial differential equations as long as all free variables of $\eta$ have some differential equation so that their differential symbols agree with their time-derivatives.
\qedhere
\end{proofatend}

The differential effect axiom \irref{DE} axiomatizes the effect of differential equations on the differential symbols.
The key insight for its soundness is that differential symbol $\D{x}$ already has the value $\genDE{x}$ along the differential equation \(\pevolve{\D{x}=\genDE{x}}\) such that the subsequent differential assignment \(\Dupdate{\Dumod{\D{x}}{\genDE{x}}}\) that assigns the value of $\genDE{x}$ to $\D{x}$ has no effect on the truth of the postcondition.
The differential substitution resulting from a subsequent use of axiom \irref{Dassignb} is crucial to relay the values of the time-derivatives of the state variables $x$ along a differential equation by way of their corresponding differential symbol $\D{x}$, though.
In combination, this makes it possible to soundly substitute the right-hand side of a differential equation for its left-hand side in a proof.

\begin{lemma}[Differential assignment] \label{lem:differentialAssignLemma}
  If \m{\imodels{\If}{\D{x}=\genDE{x}\land\ivr}}
  where \m{\iget[flow]{\If}:[0,r]\to\linterpretations{\Sigma}{V}} 
is a solution of any duration $r\geq0$,
  then
  \[
  \imodels{\If}{\phi \lbisubjunct \dbox{\Dupdate{\Dumod{\D{x}}{\genDE{x}}}}{\phi}}
  \]
\end{lemma}
\begin{proofatend}
\m{\imodels{\If}{\D{x}=\genDE{x}\land\ivr}} implies
\(\imodels{\Iff[\zeta]}{\D{x}=\genDE{x}\land\ivr}\),
i.e. \(\iget[state]{\Iff[\zeta]}(\D{x}) = \ivaluation{\Iff[\zeta]}{\genDE{x}}\) and also \(\imodels{\Iff[\zeta]}{\ivr}\)
for all $0\leq \zeta\leq r$.
Thus, since $\D{x}$ already has the value \(\ivaluation{\Iff[\zeta]}{\genDE{x}}\) in state $\iget[state]{\Iff[\zeta]}$, the differential assignment \(\Dupdate{\Dumod{\D{x}}{\genDE{x}}}\) has no effect, thus,
\(\iaccessible[\Dupdate{\Dumod{\D{x}}{\genDE{x}}}]{\Iff[\zeta]}{\Iff[\zeta]}\)
so that $\phi$ and \(\dbox{\Dupdate{\Dumod{\D{x}}{\genDE{x}}}}{\phi}\) are equivalent along $\iget[flow]{\If}$.
Hence, \(\imodels{\If}{(\phi \lbisubjunct \dbox{\Dupdate{\Dumod{\D{x}}{\genDE{x}}}}{\phi})}\).
\qedhere
\end{proofatend}

The final insights for differential invariant reasoning for differential equations are syntactic ways of computing differentials, which can be internalized as axioms (\irref{Dconst+Dvar+Dplus+Dtimes+Dcompose}), since differentials are represented syntactically in differential-form \dL. 
It is the local semantics as differential forms that makes it possible to soundly capture the interaction of differentials with arithmetic operators by local equations.
\begin{lemma}[Derivations] \label{lem:derivationLemma}
  The following equations of differentials are valid:
  \begin{align}%
    \der{f} & = 0
      &&\text{for arity 0 functions or numbers}~f
    \label{eq:Dconstant}\\
    \der{x} & =  \D{x}
      &&\text{for variables}~x\in\allvars\label{eq:Dpolynomial}\\
    \der{\theta+\eta} & = \der{\theta} + \der{\eta}
    \label{eq:Dadditive}\\
    \der{\theta\cdot \eta} & = \der{\theta}\cdot \eta + \theta\cdot\der{\eta}
    \label{eq:DLeibniz}
    \\
    \dbox{\pupdate{\pumod{y}{\theta}}}{\dbox{\Dupdate{\Dumod{\D{y}}{1}}}
    {}&{\big( \der{f(\theta)} = \der{f(y)}\stimes\der{\theta}\big)}}
    &&\text{for $y,\D{y}\not\in\freevarsdef{\theta}$}
    \label{eq:Dcompose}
  \end{align}%
\end{lemma}
\def\Imyypyval{\ivaluation{\I}{\theta}}%
\def\Imyyp{\vdLint[const=I,state=\omega]}%
\begin{proofatend}
The proof shows each equation separately.
The first part considers any constant function (i.e. arity 0) or number literal $f$ for \rref{eq:Dconstant} and then aligns the differential \(\der{x}\) of a term that happens to be a variable $x\in\allvars$ with its corresponding differential symbol $\D{x}\in\D{\allvars}$ for \rref{eq:Dpolynomial}.
The other cases exploit linearity for \rref{eq:Dadditive} and Leibniz properties of partial derivatives for \rref{eq:DLeibniz}.
Case \rref{eq:Dcompose} exploits the chain rule and assignments and differential assignments for the fresh $y,\D{y}$ to mimic partial derivatives.
Equation \rref{eq:Dcompose} generalizes to functions $f$ of arity $n>1$, in which case $\stimes$ is the (definable) Euclidean scalar product.
\def\aptag#1{\tag{#1}}%
\begin{align}
\ivaluation{\I}{\der{f}}
&= \sum_x \iget[state]{\I}(\D{x}) \itimes \Dp[x]{\ivaluation{\Idot}{f}}(\iget[state]{\I})
= \sum_x \iget[state]{\I}(\D{x}) \itimes \Dp[x]{\iget[const]{\Idot}(f)}(\iget[state]{\I})
= 0
\aptag{\ref{eq:Dconstant}}
\\
\ivaluation{\I}{\der{x}}
&=
\sum_y \iget[state]{\I}(\D{y}) \itimes \Dp[y]{\ivaluation{\Idot}{x}}(\iget[state]{\I})
= \iget[state]{\I}(\D{x})
= \ivaluation{\I}{\D{x}}
\aptag{\ref{eq:Dpolynomial}}
\\
\ivaluation{\I}{\der{\theta+\eta}}
&= \sum_x \iget[state]{\I}(\D{x}) \itimes \Dp[x]{\ivaluation{\Idot}{\theta+\eta}}(\iget[state]{\I})
= \sum_x \iget[state]{\I}(\D{x}) \itimes \Dp[x]{(\ivaluation{\Idot}{\theta}+\ivaluation{\Idot}{\eta})}(\iget[state]{\I})
\notag
\\&= \sum_x \iget[state]{\I}(\D{x}) \itimes \Big(\Dp[x]{\ivaluation{\Idot}{\theta}}(\iget[state]{\I}) + \Dp[x]{\ivaluation{\Idot}{\eta}}(\iget[state]{\I})\Big)
\notag
= \sum_x \iget[state]{\I}(\D{x}) \itimes \Dp[x]{\ivaluation{\Idot}{\theta}}(\iget[state]{\I}) + \sum_x \iget[state]{\I}(\D{x}) \itimes \Dp[x]{\ivaluation{\Idot}{\eta}}(\iget[state]{\I})
\notag
\\&= \ivaluation{\I}{\der{\theta}} + \ivaluation{\I}{\der{\eta}}
= \ivaluation{\I}{\der{\theta} + \der{\eta}}
\aptag{\ref{eq:Dadditive}}
\\
\ivaluation{\I}{\der{\theta\cdot\eta}}
&= \sum_x \iget[state]{\I}(\D{x}) \itimes \Dp[x]{\ivaluation{\Idot}{\theta\cdot\eta}}(\iget[state]{\I})
= \sum_x \iget[state]{\I}(\D{x}) \itimes \Dp[x]{(\ivaluation{\Idot}{\theta}\cdot\ivaluation{\Idot}{\eta})}(\iget[state]{\I})
\notag
\\
&= \sum_x \iget[state]{\I}(\D{x}) \itimes \Big(\ivaluation{\I}{\eta} \itimes \Dp[x]{\ivaluation{\Idot}{\theta}}(\iget[state]{\I})
+ \ivaluation{\I}{\theta} \itimes \Dp[x]{\ivaluation{\Idot}{\eta}}(\iget[state]{\I})\Big)
\notag
\\
&=
\ivaluation{\I}{\eta} \sum_x \iget[state]{\I}(\D{x}) \itimes \Dp[x]{\ivaluation{\Idot}{\theta}}(\iget[state]{\I})
+ \ivaluation{\I}{\theta} \sum_x \iget[state]{\I}(\D{x}) \itimes \Dp[x]{\ivaluation{\Idot}{\eta}}(\iget[state]{\I})
\notag
\\
&= 
\ivaluation{\I}{\der{\theta}}\cdot \ivaluation{\I}{\eta} + \ivaluation{\I}{\theta}\cdot\ivaluation{\I}{\der{\eta}}
= \ivaluation{\I}{\der{\theta}\cdot \eta + \theta\cdot\der{\eta}}
\aptag{\ref{eq:DLeibniz}}
\intertext{Proving that
\(\imodels{\I}{\dbox{\pupdate{\pumod{y}{\theta}}}{\dbox{\Dupdate{\Dumod{\D{y}}{1}}}{\big(\der{f(\theta)} = \der{f(y)}\stimes\der{\theta}\big)}}}\)
requires
\(\imodels{\Imyyp}{\der{f(\theta)} = \der{f(y)}\stimes\der{\theta}}\),
i.e.\ that
\(\ivaluation{\Imyyp}{\der{f(\theta)}} = \ivaluation{\Imyyp}{\der{f(y)}\stimes\der{\theta}}\), where $\iget[state]{\Imyyp}$ agrees with state $\iget[state]{\I}$ except that \(\iget[state]{\Imyyp}(y)=\ivaluation{\I}{\theta}\) and \(\iget[state]{\Imyyp}(\D{y})=1\).
This is equivalent to
\(\ivaluation{\I}{\der{f(\theta)}} = \ivaluation{\Imyyp}{\der{f(y)}}\stimes\ivaluation{\I}{\der{\theta}}\)
by \rref{lem:coincidence-term} since
\(\iget[state]{\I}=\iget[state]{\Imyyp}\) on $\scomplement{\{y,\D{y}\}}$ and
\(y,\D{y}\not\in\freevarsdef{\theta}\) by assumption,
so \(y,\D{y}\not\in\freevarsdef{\der{f(\theta)}}\) and \(y,\D{y}\not\in\freevarsdef{\der{\theta}}\).
The latter equation proves using the chain rule (\rref{lem:chain}) and a fresh variable $z$ when denoting \(\ivaluation{\Idot}{f} \mdefeq \iget[const]{\Idot}(f)\) using \rref{lem:coincidence-term}:}
    \ivaluation{\I}{\der{f(\theta)}} & =
    \sum_x \iget[state]{\I}(\D{x}) \Dp[x]{\ivaluation{\Idot}{f(\theta)}}(\iget[state]{\I})
    = \sum_x \iget[state]{\I}(\D{x}) \Dp[x]{(\ivaluation{\Idot}{f}\compose\ivaluation{\Idot}{\theta})}(\iget[state]{\I})
  \notag
    \\&
     \overset{\text{chain}}{=} \sum_x \iget[state]{\I}(\D{x}) \Dp[y]{\ivaluation{\Idot}{f}}\big(\ivaluation{\I}{\theta}\big) \stimes \Dp[x]{\ivaluation{\Idot}{\theta}}(\iget[state]{\I})
\notag
     \\&
     = \Dp[y]{\ivaluation{\Idot}{f}}\big(\ivaluation{\I}{\theta}\big) \stimes \sum_x \iget[state]{\I}(\D{x}) \Dp[x]{\ivaluation{\Idot}{\theta}}(\iget[state]{\I})
     = \Dp[y]{\ivaluation{\Idot}{f}}\big(\ivaluation{\I}{\theta}\big) \stimes \ivaluation{\I}{\der{\theta}}
\notag
     \\&= \Dp[y]{\iget[const]{\Idot}(f)}\big(\ivaluation{\I}{\theta}\big) \stimes \ivaluation{\I}{\der{\theta}}
\notag
    =
    \Dp[z]{\iget[const]{\Idot}(f)}\big(\ivaluation{\Imyyp}{y}\big) \itimes \Dp[y]{\ivaluation{\Idot}{y}}(\iget[state]{\Imyyp}) \stimes \ivaluation{\I}{\der{\theta}}
\notag
    \\&\overset{\text{chain}}{=}
    \Dp[y]{(\iget[const]{\Idot}(f)\compose\ivaluation{\Idot}{y})}(\iget[state]{\Imyyp}) \stimes \ivaluation{\I}{\der{\theta}}
\notag
    =
    \left(\Dp[y]{\ivaluation{\Idot}{f(y)}}(\iget[state]{\Imyyp})\right) \stimes \ivaluation{\I}{\der{\theta}}
\notag
    \\&=
    \left(\iget[state]{\Imyyp}(\D{y}) \Dp[y]{\ivaluation{\Idot}{f(y)}}(\iget[state]{\Imyyp})\right) \stimes \ivaluation{\I}{\der{\theta}}
\notag
    =
    \left(\sum_{x\in\{y\}} \iget[state]{\Imyyp}(\D{x}) \Dp[x]{\ivaluation{\Idot}{f(y)}}(\iget[state]{\Imyyp})\right) \stimes \ivaluation{\I}{\der{\theta}}
\notag
    \\&=
    \ivaluation{\Imyyp}{\der{f(y)}} \stimes \ivaluation{\I}{\der{\theta}}
    \aptag{\ref{eq:Dcompose}}
\end{align}
\qedhere
\end{proofatend}

\subsection{Soundness}

The uniform substitution calculus for differential-form \dL is \emph{sound}, i.e.\ all formulas that it proves from valid premises are valid.
The soundness argument is entirely modular.
The concrete \dL axioms in \rref{fig:dL} and~\ref{fig:dL-ODE} are valid formulas and the axiomatic proof rules (i.e.\ pairs of formulas) in \rref{fig:dL} are locally sound, which implies soundness.
The uniform substitution rule is sound so only concludes valid formulas from valid premises (\rref{thm:usubst-sound}), which implies that \dL axioms (and other provable \dL formulas) can only be instantiated soundly by rule \irref{US}.
Uniform substitution instances of locally sound axiomatic proof rules (and other locally sound inferences) are locally sound (\rref{thm:usubst-rule}), which implies that \dL axiomatic proof rules in \rref{fig:dL} can only be instantiated soundly by uniform substitutions (\irref{USR}).

The soundness proof follows a high-level strategy that is similar to earlier proofs \cite{DBLP:conf/lics/Platzer12b,DBLP:journals/logcom/Platzer10,DBLP:journals/lmcs/Platzer12}, but ends up in stronger results since all axioms for differential equations are equivalences now.
The availability of differentials and differential assignments as syntactic elements in differential-form \dL as well as the instantiation support from uniform substitutions makes those soundness proofs significantly more modular, too.
For example, what used to be a single proof rule for differential invariants \cite{DBLP:journals/logcom/Platzer10} can now be decomposed into separate modular axioms.

\begin{theorem}[Soundness] \label{thm:dL-sound}
  The uniform substitution calculus for \dL is \emph{sound}, that is, every formula that is provable by the \dL axioms and proof rules is valid, i.e.\ true in all states of all interpretations.
  The axioms in \rref{fig:dL} and \ref{fig:dL-ODE}  are valid formulas and the axiomatic proof rules in \rref{fig:dL} are locally sound.
\end{theorem}
\begin{proof}
The axioms (and most proof rules) in \rref{fig:dL} are special instances
of corresponding axiom schemata and proof rules for differential dynamic logic \cite{DBLP:conf/lics/Platzer12b} and, thus, sound.\footnote{%
The uniform substitution proof calculus improves modularity and gives stronger equivalence formulations of axioms, though.
}
All proof rules in \rref{fig:dL} (but not \irref{US} itself) are even locally sound, which implies soundness, i.e.\ that their conclusions are valid (in all $\iget[const]{\I}$) if their premises are.
In preparation for a completeness argument, note that rules \irref{gena+MP} can be augmented soundly to use $p(\usall)$ instead of $p(x)$ or $p$, respectively, such that the $\freevarsdef{\sigma}=\emptyset$ requirement of \rref{thm:usubst-rule} will be met during \irref{USR} instances of all axiomatic proof rules.
The axioms in \rref{fig:dL-ODE} are new and need new soundness arguments.
\begin{compactenum}
\item[\irref{DW}]
Soundness of \irref{DW} uses that differential equations never leave their evolution domain by \rref{def:HP-transition}.
To show \(\imodels{\I}{\dbox{\pevolvein{\D{x}=f(x)}{q(x)}}{q(x)}}\), consider any \m{\iget[flow]{\If}:[0,r]\to\linterpretations{\Sigma}{V}} of any duration $r\geq0$ solving
\(\imodels{\If}{\D{x}=f(x)\land q(x)}\).
Then \(\imodels{\If}{q(x)}\) especially \(\imodels{\Iff[r]}{q(x)}\).

\item[\irref{DC}]
Soundness of \irref{DC} is a stronger version of  soundness for the differential cut rule \cite{DBLP:journals/logcom/Platzer10}.
\irref{DC} is a differential version of the modal modus ponens \irref{K}.
Only the direction ``$\lylpmi$'' of the equivalence in \irref{DC} needs the outer assumption \(\dbox{\pevolvein{\D{x}=f(x)}{q(x)}}{r(x)}\), but the proof of the conditional equivalence in \irref{DC} is simpler:
\[
{\dbox{\pevolvein{\D{x}=f(x)}{q(x)}}{r(x)}}
\limply
        {\big(\dbox{\pevolvein{\D{x}=f(x)}{q(x)}}{p(x)} \lbisubjunct \dbox{\pevolvein{\D{x}=f(x)}{q(x)\land r(x)}}{p(x)}\big)}
\]
The core is that if \(\dbox{\pevolvein{\D{x}=f(x)}{q(x)}}{r(x)}\), so $r(x)$ holds after that differential equation,
and if $p(x)$ holds after the differential equation \(\pevolvein{\D{x}=f(x)}{q(x)\land r(x)}\) that is additionally restricted to $r(x)$,
then $p(x)$ holds after the differential equation \(\pevolvein{\D{x}=f(x)}{q(x)}\) with no additional restriction.
Let \(\imodels{\I}{\dbox{\pevolvein{\D{x}=f(x)}{q(x)}}{r(x)}}\).
Since all restrictions of solutions are solutions, this is equivalent to
\(\imodels{\If}{r(x)}\) for all $\iget[flow]{\If}$ of any duration solving \(\imodels{\If}{\D{x}=f(x)\land q(x)}\) and starting in \(\iget[state]{\Iff[0]}=\iget[state]{\I}\) on $\scomplement{\{\D{x}\}}$.
So, for all $\iget[flow]{\If}$ starting in \(\iget[state]{\Iff[0]}=\iget[state]{\I}\) on $\scomplement{\{\D{x}\}}$:
\(\imodels{\If}{\D{x}=f(x)\land q(x)}\) is equivalent to
\(\imodels{\If}{\D{x}=f(x)\land q(x) \land r(x)}\).
Hence, \(\imodels{\I}{\dbox{\pevolvein{\D{x}=f(x)}{q(x)\land r(x)}}{p(x)}}\)
is equivalent to \(\imodels{\I}{\dbox{\pevolvein{\D{x}=f(x)}{q(x)}}{p(x)}}\).

\item[\irref{DE}]
Axiom \irref{DE} is new to differential-form \dL. Its soundness proof exploits \rref{lem:differentialAssignLemma}.
Consider any state $\iget[state]{\I}$.
Then
\(\imodels{\I}{\dbox{\pevolvein{\D{x}=f(x)}{q(x)}}{p(x,\D{x})}}\)
iff
\(\imodels{\Iff[r]}{p(x,\D{x})}\)
for all solutions $\iget[flow]{\If}:[0,r]\to\linterpretations{\Sigma}{V}$ of
\(\imodels{\If}{\D{x}=f(x)\land q(x)}\) of any duration $r$ starting in \(\iget[state]{\Iff[0]}=\iget[state]{\I}\) on $\scomplement{\{\D{x}\}}$.
That is equivalent to: for all $\iget[flow]{\If}$,
if
\(\imodels{\If}{\D{x}=f(x)\land q(x)}\)
then
\(\imodels{\If}{p(x,\D{x})}\).
By \rref{lem:differentialAssignLemma}, \(\imodels{\If}{p(x,\D{x})}\) iff
\(\imodels{\If}{\dbox{\Dupdate{\Dumod{\D{x}}{f(x)}}}{p(x,\D{x})}}\),
so that is equivalent to
\(\imodels{\Iff[r]}{\dbox{\Dupdate{\Dumod{\D{x}}{f(x)}}}{p(x,\D{x})}}\)
for all solutions $\iget[flow]{\If}:[0,r]\to\linterpretations{\Sigma}{V}$ of
\(\imodels{\If}{\D{x}=f(x)\land q(x)}\) of any duration $r$ starting in \(\iget[state]{\Iff[0]}=\iget[state]{\I}\) on $\scomplement{\{\D{x}\}}$,
which is, consequently, equivalent to
\(\imodels{\I}{\dbox{\pevolvein{\D{x}=f(x)}{q(x)}}{\dbox{\Dupdate{\Dumod{\D{x}}{f(x)}}}{p(x,\D{x})}}}\).

\item[\irref{DI}]
Soundness of \irref{DI} has some relation to the soundness proof for differential invariants \cite{DBLP:journals/logcom/Platzer10}, yet proves an equivalence and is generalized to leverage differentials.
If \(\imodels{\I}{\dbox{\pevolvein{\D{x}=f(x)}{q(x)}}{p(x)}}\) then the solution $\iget[flow]{\If}$ of duration 0 implies that \(\imodels{\I}{p(x)}\) since $\D{x}\not\in\freevars{p(x)}$ by \rref{lem:coincidence}, provided that \(\imodels{\Iff[0]}{q(x)}\), i.e.\ \(\imodels{\I}{q(x)}\) since $\D{x}\not\in\freevars{q(x)}$, such that there is a solution at all.
Thus, \({\dbox{\pevolvein{\D{x}=f(x)}{q(x)}}{p(x)}} \limply {\dbox{\ptest{q(x)}}{p(x)}}\) is valid even without the assumption.
Converse
\(
{\big(q(x)\limply \dbox{\pevolvein{\D{x}=f(x)}{q(x)}}{\der{p(x)}}\big)}
\limply
{\dbox{\ptest{q(x)}}{p(x)} \limply\dbox{\pevolvein{\D{x}=f(x)}{q(x)}}{p(x)}}
\)
is only shown for \m{p(x)\mdefequiv g(x)\geq0}, where \(\der{p(x)} \mequiv (\der{g(x)}\geq0)\), because the variations for other formulas are the same as the variations in previous work \cite{DBLP:journals/logcom/Platzer10}.
Consider a state $\iget[state]{\I}$ in which
\(\imodels{\I}{q(x) \limply \dbox{\pevolvein{\D{x}=f(x)}{q(x)}}{\der{p(x)}}}\).
If \(\inonmodels{\I}{q(x)}\), there is nothing to show, because there is no solution of \({\pevolvein{\D{x}=f(x)}{q(x)}}\) for any duration since $\D{x}\not\in\freevars{q(x)}$, so \(\imodels{\I}{\dbox{\pevolvein{\D{x}=f(x)}{q(x)}}{p(x)}}\) holds vacuously.
Otherwise, \(\imodels{\I}{q(x)}\), which implies\\ \(\imodels{\I}{\dbox{\pevolvein{\D{x}=f(x)}{q(x)}}{\der{p(x)}}}\) by assumption.
Assume \(\imodels{\I}{\dbox{\ptest{q(x)}}{p(x)}}\), so \(\imodels{\I}{p(x)}\) since \(\imodels{\I}{q(x)}\).
To show that
\(\imodels{\I}{\dbox{\pevolvein{\D{x}=f(x)}{q(x)}}{p(x)}}\) consider any solution $\iget[flow]{\If}$ of any duration $r\geq0$.
The case $r=0$ follows from \(\imodels{\I}{p(x)}\) by \rref{lem:coincidence} since \(\freevars{p(x)}=\{x\}\) is disjoint from \(\{\D{x}\}\), which, unlike $x$, is changed by evolutions of any duration, including 0.
That leaves the case \(r>0\).

Let $\iget[flow]{\If}$ be a solution of \(\pevolvein{\D{x}=f(x)}{q(x)}\) according to \rref{def:HP-transition}, so
\(\iget[state]{\I}=\iget[state]{\Iff[0]}\) on $\scomplement{\{\D{x}\}}$
and \(\imodels{\If}{\pevolvein{\D{x}=f(x)}{q(x)}}\).
Now
\(\imodels{\I}{\dbox{\pevolvein{\D{x}=f(x)}{q(x)}}{\der{p(x)}}}\)
implies
\(\imodels{\If}{\der{p(x)}}\).
As \(r>0\), \rref{lem:differentialLemma} implies
\(0\leq\ivaluation{\Iff[\zeta]}{\der{g(x)}}
= \D[t]{\ivaluation{\Iff[t]}{g(x)}} (\zeta)\)
for all $\zeta$ since $\freevars{g(x)}=\{x\}$.
Together with \(\imodels{\Iff[0]}{p(x)}\) (by \rref{lem:coincidence} and \(\freevars{p(x)}\cap\{\D{x}\}=\emptyset\) from \(\imodels{\I}{p(x)}\)), which is
\(\imodels{\Iff[0]}{g(x)\geq0}\),
this implies \(\imodels{\Iff[\zeta]}{g(x)\geq0}\) for all $\zeta$, including $r$,
by the mean-value theorem, since
\(\ivaluation{\Iff[r]}{g(x)}-\ivaluation{\Iff[0]}{g(x)} = (r-0) \D[t]{\ivaluation{\Iff[t]}{g(x)}} (\zeta) \geq 0\) for some $\zeta\in(0,r)$.
The mean-value theorem (\rref{lem:mean-value} in appendix) is applicable since the value \(\ivaluation{\Iff[t]}{g(x)}\) of term $g(x)$ along $\iget[flow]{\If}$ is continuous in $t$ on $[0,r]$ and differentiable on $(0,r)$
as compositions of the, by \rref{def:dL-valuationTerm} smooth, evaluation function and the differentiable solution $\iget[state]{\Iff[t]}$ of a differential equation.

\item[\irref{DG}]
\def\Imd{\imodif[state]{\I}{y}{d}}%
\newcommand{\Ify}{\DALint[const=I,flow=\tilde{\varphi}]}
\newcommand{\Iffy}[1][t]{\vdLint[const=I,state=\tilde{\varphi}(#1)]}
Soundness of \irref{DG} is a constructive variation of the soundness proof for differential auxiliaries \cite{DBLP:journals/lmcs/Platzer12}.
Let \(\imodels{\I}{\lexists{y}{\dbox{\pevolvein{\D{x}=f(x)\syssep\D{y}=a(x)y+b(x)}{q(x)}}{p(x)}}}\),
that is,\\ \(\imodels{\Imd}{\dbox{\pevolvein{\D{x}=f(x)\syssep\D{y}=a(x)y+b(x)}{q(x)}}{p(x)}}\) for some value $d\in\reals$.
In order to show that \(\imodels{\I}{\dbox{\pevolvein{\D{x}=f(x)}{q(x)}}{p(x)}}\),
consider any \(\iget[flow]{\If}:[0,r]\to\linterpretations{\Sigma}{V}\) such that
\(\imodels{\If}{\D{x}=f(x)\land q(x)}\) and \(\iget[state]{\Iff[0]}=\iget[state]{\I}\) on $\scomplement{\{\D{x}\}}$.
By modifying the values of $y,\D{y}$ along $\iget[flow]{\If}$, this function can be augmented to a solution 
\(\iget[flow]{\Ify}:[0,r]\to\linterpretations{\Sigma}{V}\) such that
\(\imodels{\Ify}{\D{x}=f(x)\land\D{y}=a(x)y+b(x)\land q(x)}\) and \(\iget[state]{\Iffy[0]}(y)=d\) as shown below.
The assumption then implies \(\imodels{\Iffy[r]}{p(x)}\), which, by \rref{lem:coincidence}, is equivalent to \(\imodels{\Iff[r]}{p(x)}\) since \(y,\D{y}\not\in\freevars{p(x)}\) and \(\iget[state]{\Iff[r]}=\iget[state]{\Iffy[r]}\) on $\scomplement{\{y,\D{y}\}}$,
which implies \(\imodels{\I}{\dbox{\pevolvein{\D{x}=f(x)}{q(x)}}{p(x)}}\), since $\iget[flow]{\If}$ was arbitrary.

The construction of the modification $\iget[flow]{\Ify}$ of $\iget[flow]{\If}$ on $\{y,\D{y}\}$ proceeds as follows.
By Picard-Lindel\"of's theorem (\rref{thm:PicardLindelof-global} in the appendix), there is a unique solution \(y:[0,r]\to\reals\) of the initial-value problem
\begin{equation}
\begin{aligned}
  y(0) &=d\\
  \D{y}(t) &= F(t,y(t)) \mdefeq y(t) (\ivaluation{\Iff[t]}{a(x)}) + \ivaluation{\Iff[t]}{b(x)}
\end{aligned}
\label{eq:diffghost-extra-ODE}
\end{equation}
because \(F(t,y)\) is continuous on \([0,r]\times\reals\)
(since \(\ivaluation{\Iff[t]}{a(x)}\) and \(\ivaluation{\Iff[t]}{b(x)}\) are continuous in $t$ as compositions of the, by \rref{def:dL-valuationTerm} smooth, evaluation function and the continuous solution $\iget[state]{\Iff[t]}$ of a differential equation)
and because \(F(t,y)\) satisfies the Lipschitz condition
\[
\norm{F(t,y)-F(t,z)} = \norm{(y-z)(\ivaluation{\Iff[t]}{a(x)})} \leq \norm{y-z} \max_{t\in[0,r]} \ivaluation{\Iff[t]}{a(x)}
\]
where the maximum exists, because it is a maximum of a continuous function on the compact set $[0,r]$.
The modification $\iget[flow]{\Ify}$ agrees with $\iget[flow]{\If}$ on $\scomplement{\{y,\D{y}\}}$.
On $\{y,\D{y}\}$, the modification $\iget[flow]{\Ify}$ is defined as \(\iget[state]{\Iffy[t]}(y) = y(t)\) and \(\iget[state]{\Iffy[t]}(\D{y}) = F(t,y(t))\), respectively, for the solution $y(t)$ of \rref{eq:diffghost-extra-ODE}.
In particular \(\iget[state]{\Iffy[t]}(\D{y})\) agrees with the time-derivative \(\D{y}(t)\) of the value \(\iget[state]{\Iffy[t]}(y) = y(t)\) of $y$ along $\iget[flow]{\Ify}$.
By construction, \(\iget[state]{\Iffy[0]}(y)=d\) and
\(\imodels{\Ify}{\D{x}=f(x)\land\D{y}=a(x)y+b(x)\land q(x)}\),
because \(\D{y}=a(x)y+b(x)\) holds along $\iget[flow]{\Ify}$ by \rref{eq:diffghost-extra-ODE} and
because \(\iget[state]{\Iff[t]}=\iget[state]{\Iffy[t]}\) on $\scomplement{\{y,\D{y}\}}$ so that
\({\D{x}=f(x)}\land{q(x)}\) continues to hold along $\iget[flow]{\Ify}$ by \rref{lem:coincidence-term} because \(y,\D{y}\not\in\freevars{{\D{x}=f(x)}\land{q(x)}}\).

\def\Imyd{\imodif[state]{\I}{y}{d}}%
\newcommand{\Ifrest}{\DALint[const=I,flow=\restrict{\varphi}{\scomplement{\{y,\D{y}\}}}]}%
\newcommand*{\Iffrest}[1][\zeta]{\vdLint[const=I,state=\restrict{\varphi}{\scomplement{\{y,\D{y}\}}}(#1)]}%
Conversely, let \(\imodels{\I}{\dbox{\pevolvein{\D{x}=f(x)}{q(x)}}{p(x)}}\).
This direction shows a stronger version of
\(\imodels{\I}{\lexists{y}{\dbox{\pevolvein{\D{x}=f(x)\syssep\D{y}=a(x)y+b(x)}{q(x)}}{p(x)}}}\)
by showing for all terms $\eta$ that\\
\(\imodels{\I}{\lforall{y}{\dbox{\pevolvein{\D{x}=f(x)\syssep\D{y}=\eta}{q(x)}}{p(x)}}}\).
Consider any $d\in\reals$ and term $\eta$ and show\\
\(\imodels{\Imyd}{\dbox{\pevolvein{\D{x}=f(x)\syssep\D{y}=\eta}{q(x)}}{p(x)}}\).
Consider any \(\iget[flow]{\If}:[0,r]\to\linterpretations{\Sigma}{V}\) such that\\
\(\imodels{\If}{\D{x}=f(x)\land\D{y}=\eta\land q(x)}\)
with \(\iget[state]{\Iff[0]}=\iget[state]{\Imyd}\) on $\scomplement{\{\D{x},\D{y}\}}$.
Then the restriction \(\iget[flow]{\Ifrest}\) of $\iget[flow]{\If}$ to $\scomplement{\{y,\D{y}\}}$ with \(\iget[state]{\Iffrest[t]}=\iget[state]{\Imyd}\) on $\{y,\D{y}\}$ for all $t\in[0,r]$ still solves
\(\imodels{\Ifrest}{\D{x}=f(x)\land q(x)}\) by \rref{lem:coincidence-term} since \(\iget[flow]{\Ifrest}=\iget[flow]{\If}\) on $\scomplement{\{y,\D{y}\}}$ and \(y,\D{y}\not\in\freevars{{\D{x}=f(x)}\land{q(x)}}\).
It also satisfies \(\iget[state]{\Iffrest[0]}=\iget[state]{\Imyd}\) on $\scomplement{\{\D{x}\}}$,
because \(\iget[state]{\Iff[0]}=\iget[state]{\Imyd}\) on $\scomplement{\{\D{x},\D{y}\}}$ yet \(\iget[state]{\Iffrest[t]}(\D{y})=\iget[state]{\Imyd}(\D{y})\).
Thus, by assumption, \(\imodels{\Iffrest[r]}{p(x)}\),
which implies \(\imodels{\Iff[r]}{p(x)}\)
by \rref{lem:coincidence}, because
\(\iget[flow]{\If}=\iget[flow]{\Ifrest}\) on $\scomplement{\{y,\D{y}\}}$ and $y,\D{y}\not\in\freevars{p(x)}$.

In particular, axiom \irref{DG} continues to be sound when replacing $\exists{y}$ by $\forall{y}$.

\item[\irref{DS}]
\def\Izeta{\imodif[state]{\I}{t}{\zeta}}%
\def\constODE{f}%
Soundness of the solution axiom \irref{DS} follows from existence and uniqueness of global solutions of constant differential equations.
Consider any state $\iget[state]{\I}$.
By \rref{thm:PicardLindelof-global}, there is a unique global solution $\iget[flow]{\If}:[0,\infty)\to\linterpretations{\Sigma}{V}$ defined as \(\iget[state]{\Iff[\zeta]}(x) \mdefeq \ivaluation{\Izeta}{x+\constODE t}\)
and \(\iget[state]{\Iff[\zeta]}(\D{x}) \mdefeq \D[t]{\iget[state]{\Iff[t]}(x)}(\zeta) = \iget[const]{\Izeta}(\constODE)\)
and \(\iget[state]{\Iff[\zeta]} = \iget[state]{\I}\) on $\scomplement{\{x,\D{x}\}}$.
This solution satisfies
\(\iget[state]{\Iff[0]}=\iget[state]{\I}(x)\) on $\scomplement{\{\D{x}\}}$ 
and
  \m{\imodels{\If}{\D{x}=\constODE}},
  i.e.
  \(\imodels{\Iff[\zeta]}{\D{x}=\constODE}\)
  for all \(0\leq \zeta\leq r\).
All solutions of \(\pevolve{\D{x}=\constODE}\) from initial state $\iget[state]{\I}$ are restrictions of $\iget[flow]{\If}$ to subintervals of $[0,\infty)$.
The (unique) state $\iget[state]{\It}$ that satisfies \(\iaccessible[\pupdate{\pumod{x}{x+\constODE t}}]{\Izeta}{\It}\) satisfies the agreement \(\iget[state]{\It}=\iget[state]{\Iff[\zeta]}\) on $\scomplement{\{\D{x}\}}$, so that, by $\D{x}\not\in\freevars{p(x)}$, \rref{lem:coincidence} implies that \(\imodels{\It}{p(x)}\) iff \(\imodels{\Iff[\zeta]}{p(x)}\).

First consider axiom
\({\dbox{\pevolve{\D{x}=\constODE}}{p(x)}} \lbisubjunct {\lforall{t{\geq}0}{\dbox{\pupdate{\pumod{x}{x+\constODE t}}}{p(x)}}}\) for the special case $q(x)\mequiv\ltrue$.
If \(\imodels{\I}{\dbox{\pevolve{\D{x}=\constODE}}{p(x)}}\),
then \(\imodels{\Iff[\zeta]}{p(x)}\) for all $\zeta\geq0$,
because the restriction of $\iget[flow]{\If}$ to $[0,\zeta]$ solves \(\D{x}=\constODE\) from $\iget[state]{\I}$,
thus \(\imodels{\It}{p(x)}\) since \(\iget[state]{\It}=\iget[state]{\Iff[\zeta]}\) on $\scomplement{\{\D{x}\}}$ and $\D{x}\not\in\freevars{p(x)}$ by \rref{lem:coincidence},
which implies
\(\imodels{\Izeta}{\dbox{\pupdate{\pumod{x}{x+\constODE t}}}{p(x)}}\),
so \(\imodels{\I}{\lforall{t{\geq}0}{\dbox{\pupdate{\pumod{x}{x+\constODE t}}}{p(x)}}}\) as $\zeta\geq0$ was arbitrary.

Conversely, \(\imodels{\I}{\lforall{t{\geq}0}{\dbox{\pupdate{\pumod{x}{x+\constODE t}}}{p(x)}}}\)
implies \(\imodels{\Izeta}{\dbox{\pupdate{\pumod{x}{x+\constODE t}}}{p(x)}}\)
for all $\zeta\geq0$,
i.e. $\imodels{\It}{p(x)}$ when \(\iaccessible[\pupdate{\pumod{x}{x+\constODE t}}]{\Izeta}{\It}\).
\rref{lem:coincidence} again implies \(\imodels{\Iff[\zeta]}{p(x)}\) for all $\zeta\geq0$ since $\D{x}\not\in\freevars{p(x)}$, so \(\imodels{\I}{\dbox{\pevolve{\D{x}=\constODE}}{p(x)}}\), since all solutions are restrictions of $\iget[flow]{\If}$.

Soundness of \irref{DS} follows using that all solutions $\iget[flow]{\If}:[0,r]\to\linterpretations{\Sigma}{V}$ of \(\pevolvein{\D{x}=f(x)}{q(x)}\) satisfy \(\imodels{\Iff[\zeta]}{q(x)}\) for all $0\leq\zeta\leq r$, which, using \rref{lem:coincidence} as above, is equivalent to \(\imodels{\I}{\lforall{0{\leq}s{\leq}t}{q(x+\constODE s)}}\) when \(\iget[state]{\I}(t)=r\).

\item[\irref{Dplus+Dtimes+Dcompose+Dconst+Dvar}] Soundness of the derivation axioms \irref{Dplus+Dtimes+Dcompose} as well as \irref{Dconst+Dvar} follows from \rref{lem:derivationLemma}, since they are special instances of \rref{eq:Dadditive}, \rref{eq:DLeibniz} and \rref{eq:Dcompose} as well as \rref{eq:Dconstant} and \rref{eq:Dpolynomial}, respectively.
For axiom \irref{Dcompose} observe that $y,\D{y}\not\in\freevars{g(x)}$.

\item[\irref{G}]
Let the premise \(p(\usall)\) be valid in some $\iget[const]{\I}$, i.e.\ \m{\iget[const]{\I}\models{p(\usall)}}, i.e.\ \(\imodels{\It}{p(\usall)}\) for all $\iget[state]{\It}$.
Then, the conclusion \(\dbox{a}{p(\usall)}\) is valid in the same $\iget[const]{\I}$,
i.e.\ \(\imodels{\I}{\dbox{a}{p(\usall)}}\) for all $\iget[state]{\I}$,
because \(\imodels{\It}{p(\usall)}\) for all $\iget[state]{\It}$, so also for all $\iget[state]{\It}$ with \(\iaccessible[a]{\I}{\It}\).
Thus, \irref{G} is locally sound.

\item[\irref{gena}]
\def\Imd{\imodif[state]{\I}{x}{d}}%
Let the premise \(p(x)\) be valid in some $\iget[const]{\I}$, i.e.\ \m{\iget[const]{\I}\models{p(x)}}, i.e.\ \(\imodels{\It}{p(x)}\) for all $\iget[state]{\It}$.
Then, the conclusion \(\lforall{x}{p(x)}\) is valid in the same $\iget[const]{\I}$,
i.e.\ \(\imodels{\I}{\lforall{x}{p(x)}}\) for all $\iget[state]{\I}$, i.e.\ \(\imodels{\Imd}{p(x)}\) for all $d\in\reals$,
because \(\imodels{\It}{p(x)}\) for all $\iget[state]{\It}$, so in particular for all $\iget[state]{\It}=\iget[state]{\Imd}$ for any real $d\in\reals$.
Thus, \irref{gena} is locally sound.

\item[\irref{CQ}]
Let the premise \(f(\usall)=g(\usall)\) be valid in some $\iget[const]{\I}$, i.e.\ \m{\iget[const]{\I}\models{f(\usall)=g(\usall)}}, which is \(\imodels{\I}{f(\usall)=g(\usall)}\) for all $\iget[state]{\I}$,
i.e.\ \(\ivaluation{\I}{f(\usall)}=\ivaluation{\I}{g(\usall)}\) for all $\iget[state]{\I}$.
Consequently, \(\ivaluation{\I}{f(\usall)}\in\iget[const]{\I}(p)\) iff \(\ivaluation{\I}{g(\usall)}\in\iget[const]{\I}(p)\).
So, \(\iget[const]{\I}\models {p(f(\usall)) \lbisubjunct p(g(\usall))}\).
Thus, \irref{CQ} is locally sound.

\item[\irref{CE}]
Let the premise \(p(\usall)\lbisubjunct q(\usall)\) be valid in some $\iget[const]{\I}$, i.e.\ \m{\iget[const]{\I}\models{p(\usall) \lbisubjunct q(\usall)}}, which is \(\imodels{\I}{p(\usall) \lbisubjunct q(\usall)}\) for all $\iget[state]{\I}$.
Consequently, \(\imodel{\I}{p(\usall)} = \imodel{\I}{q(\usall)}\).
Thus,
\(\imodel{\I}{\contextapp{C}{p(\usall)}}
= \iget[const]{\I}(C)\big(\imodel{\I}{p(\usall)}\big)
= \iget[const]{\I}(C)\big(\imodel{\I}{q(\usall)}\big)
= \imodel{\I}{\contextapp{C}{q(\usall)}}\).
This implies
\(\iget[const]{\I}\models{\contextapp{C}{p(\usall)} \lbisubjunct \contextapp{C}{q(\usall)}}\),
hence the conclusion is valid in $\iget[const]{\I}$.
Thus, \irref{CE} is locally sound.

\item[\irref{MP}]
Modus ponens \irref{MP} is locally sound with respect to interpretation $\iget[const]{\I}$ \emph{and} state $\iget[state]{\I}$, which implies local soundness.
If \(\imodels{\I}{p\limply q}\) and \(\imodels{\I}{p}\) then \(\imodels{\I}{q}\).

\item[\irref{US}]
Rule \irref{USR} preserves local soundness by \rref{thm:usubst-rule} and rule \irref{US} is sound by \rref{thm:usubst-sound}, just not locally sound.
\qedhere
\end{compactenum}
\end{proof}

Observe that uniform substitutions are not limited to merely instantiating \dL axioms and axiomatic proof rules.
Rule \irref{US} can be used to instantiate any \dL formula soundly (\rref{thm:usubst-sound}), which, in particular, gives a simple mechanism for derived axioms and lemmas, which are just \dL formulas that have a proof.
Uniform substitutions can instantiate any locally sound proof as well (\rref{thm:usubst-rule}), which, in particular, gives a simple mechanism for derived axiomatic rules, definitions, and invariant search with lazy instantiation of invariants.
These are just proofs from the \dL rules and axioms in \rref{fig:dL} and \ref{fig:dL-ODE} whose premises and conclusions are uniformly substituted to instantiate the requisite function or predicate symbols (recall \rref{ex:diffind-free-parametric-proof}).

\subsection{Completeness}
\newcommand{\reduct}[1]{#1^\flat}%
\newcommand{\LBase}{\textit{L}\xspace}%

By \rref{thm:dL-sound}, the \dL calculus is \emph{sound}, so every \dL formula that is provable using the \dL axioms and proof rules is valid, i.e.\ true in all states of all interpretations.
The more intriguing converse whether the \dL calculus is \emph{complete}, i.e.\ can prove all \dL formulas that are valid, has an answer, too.
Previous calculi for \dL were proved to be complete relative to differential equations \cite{DBLP:journals/jar/Platzer08,DBLP:conf/lics/Platzer12b} and also proved complete relative to discrete dynamics \cite{DBLP:conf/lics/Platzer12b}.
A generalization of the Hilbert calculus to hybrid games was even proved complete schematically \cite{DBLP:journals/tocl/Platzer15}.
The uniform substitution calculus for differential-form \dL is, to a large extent, a specialization of previous calculi tailored to significantly simplify soundness arguments.
Yet, completeness does not transfer when restricting proof calculi.
In fact, one key question is whether the restrictions imposed upon proofs for soundness purposes by the simple technique of uniform substitutions does also preserve completeness.
Indeed, completeness can be shown to carry over from a previous schematic completeness proof for differential game logic \cite{DBLP:journals/tocl/Platzer15} using expressiveness results from previous completeness proofs \cite{DBLP:journals/jar/Platzer08,DBLP:conf/lics/Platzer12b} by augmenting the schematic completeness proof with instantiability proofs.

The first challenge is to prove that uniform substitutions are flexible enough to prove all required instances of the \dL axioms and axiomatic proof rules.
For simplicity, consider $p(\usall)$ to be a quantifier symbol of arity 0.
  A \dL formula $\varphi$ is called \emph{surjective} iff rule \irref{US} can instantiate $\varphi$ to any of its axiom schema instances, which are those formulas that are obtained by uniformly replacing program constants $a$ by any hybrid programs and quantifier symbols $\contextapp{C}{}$ by formulas.
  An axiomatic rule is called \emph{surjective} iff \irref{USR} can instantiate it to any of its proof rule schema instances.

\begin{lemma}[Surjective axioms] \label{lem:surjectiveaxiom}
  If $\varphi$ is a \dL formula that is built only from quantifier symbols of arity 0 and program constants but no function or predicate symbols,
  then $\varphi$ is surjective.
  Axiomatic rules consisting of surjective \dL formulas are surjective.
\end{lemma}
\begin{proofatend}
Let $\tilde{\varphi}$ be the desired instance of the axiom schema belonging to $\varphi$, that is, let $\tilde{\varphi}$ be obtained from $\varphi$ by uniformly replacing each quantifier symbol $\contextapp{C}{}$ by some formula, na\"ively but consistently (same replacement for $\contextapp{C}{}$ in all places) and accordingly for program constants $a$.
The proof follows a structural induction on $\varphi$ to show that there is a uniform substitution $\sigma$ with $\freevars{\sigma}=\emptyset$ such that \(\applyusubst{\sigma}{\varphi}=\tilde{\varphi}\).
The proof for formulas is by a mostly straightforward simultaneous induction with programs:
\begin{enumerate}
\item Consider quantifier symbol \(\contextapp{C}{}\) of arity 0 and let $\tilde{\varphi}$ be the desired instance.
Define $\sigma=\usubstlist{\usubstmod{\contextapp{C}{}}{\tilde{\varphi}}}$, which has \(\freevars{\sigma}=\emptyset\), because it only substitutes quantifier symbols.
Then
\(\applyusubst{\sigma}{\contextapp{C}{}}
\mequiv \applysubst{\sigma}{\contextapp{C}{}}
\mequiv \tilde{\varphi}\).
The substitution is admissible for all arguments, since there are none.

\item \label{case:surjective-and}
Consider \(\phi\land\psi\) and let \(\tilde{\phi}\land\tilde{\psi}\) be the desired instance (which has to have this shape to qualify as a schema instance).
By induction hypothesis, there are uniform substitutions $\sigma,\tau$ with $\freevars{\sigma}=\freevars{\tau}=\emptyset$ such that
\(\applyusubst{\sigma}{\phi}=\tilde{\phi}\) and
\(\applyusubst{\sigma}{\psi}=\tilde{\psi}\).
Then the union \(\usubstjoin{\sigma}{\tau}\) of uniform substitutions $\sigma$ and $\tau$ is defined, because for all symbols $a$ of any syntactic category: if $a\in\replacees{\sigma}$ and $a\in\replacees{\tau}$, then \(\applysubst{\sigma}{a}=\applysubst{\tau}{a}\) since all replacements are uniform, so the same replacement is used everywhere in $\phi\land\psi$ for the same symbol~$a$.
Consequently,
\(\applyusubst{(\usubstjoin{\sigma}{\tau})}{\phi}=\applyusubst{\sigma}{\phi}=\tilde{\phi}\) and
\(\applyusubst{(\usubstjoin{\sigma}{\tau})}{\psi}=\applyusubst{\tau}{\psi}=\tilde{\psi}\), because all symbols that are replaced are replaced uniformly everywhere so either do not occur in $\phi$ or are already handled by $\sigma$ in the same way (and likewise either do not occur in $\psi$ or are already handled by $\tau$).
Finally, \(\freevars{\usubstjoin{\sigma}{\tau}}=\freevars{\sigma}\cup\freevars{\tau}=\emptyset\).

\item \rref{case:surjective-and} generalizes to a general \emph{uniform replacement argument}: the induction hypothesis and uniform replacement assumptions imply for each subexpression \(\theta \circ \eta\) of $\varphi$ with any operator $\circ$ that the corresponding desired instance has to have the same shape \(\tilde{\theta} \circ \tilde{\eta}\) and that there are uniform substitutions $\sigma,\tau$ with \(\freevars{\sigma}=\freevars{\tau}=\emptyset\) such that their union
\(\usubstjoin{\sigma}{\tau}\) is defined and 
\(\applyusubst{(\usubstjoin{\sigma}{\tau})}{\theta \circ \eta}
= \applyusubst{\sigma}{\theta} \circ \applyusubst{\tau}{\eta}=\tilde{\theta} \circ \tilde{\eta}\)
and \(\freevars{\usubstjoin{\sigma}{\tau}}=\freevars{\sigma}\cup\freevars{\tau}=\emptyset\).

This shows the cases $\phi\lor\psi$, $\phi\limply\psi$, $\phi\lbisubjunct\psi$ and, after a moment's thought, also $\lnot\phi$.

\item Consider \(\lforall{x}{\phi}\) with desired instance \(\lforall{x}{\tilde{\phi}}\), which has to have this shape.
By induction hypothesis, there is a uniform substitution $\sigma$ with $\freevars{\sigma}=\emptyset$ such that \(\applyusubst{\sigma}{\phi}=\tilde{\phi}\).
Thus, \(\applyusubst{\sigma}{\lforall{x}{\phi}} = \lforall{x}{\applyusubst{\sigma}{\phi}} = \lforall{x}{\tilde{\phi}}\), which is $\{x\}$-admissible because $\freevars{\sigma}=\emptyset$.

\item The case \(\lexists{x}{\phi}\) is accordingly.

\item Consider \(\dbox{\alpha}{\phi}\) with desired instance \(\dbox{\tilde{\alpha}}{\tilde{\phi}}\).
By induction hypothesis and the uniform replacement argument, there are uniform substitutions $\sigma,\tau$ such that
\(\applyusubst{(\usubstjoin{\sigma}{\tau})}{\dbox{\alpha}{\phi}}
= \dbox{\applyusubst{\sigma}{\alpha}}{\applyusubst{\tau}{\phi}}
= \dbox{\tilde{\alpha}}{\tilde{\phi}}\)
which is admissible, because $\usubstjoin{\sigma}{\tau}$ is $\boundvars{\applyusubst{(\usubstjoin{\sigma}{\tau})}{\alpha}}$-admissible for $\dbox{\alpha}{\phi}$ since \(\freevars{\usubstjoin{\sigma}{\tau}}=\emptyset\).

\item The case \(\ddiamond{\alpha}{\phi}\) is accordingly.
\end{enumerate}

\noindent
The proof for hybrid programs is by simultaneous induction with formulas, where most cases are in analogy to the previous cases, except:
\begin{enumerate}
\item Consider program constant $a$ with desired instance $\tilde{a}$.
Then \(\sigma=\usubstlist{\usubstmod{a}{\tilde{a}}}\) has \(\freevars{\sigma}=\emptyset\) and satisfies \(\applyusubst{\sigma}{a}=\applysubst{\sigma}{a}=\tilde{a}\).

\item Consider the case \(\pevolvein{\D{x}=\theta}{\ivr}\) with desired instance \(\pevolvein{\D{x}=\tilde{\theta}}{\tilde{\ivr}}\), which has to have this shape.
By induction hypothesis and uniform replacement argument, there are uniform substitutions $\sigma,\tau$ such that
\(\applyusubst{(\usubstjoin{\sigma}{\tau})}{\pevolvein{\D{x}=\theta}{\ivr}}
\mequiv \pevolvein{\D{x}=\applyusubst{\sigma}{\theta}}{\applyusubst{\tau}{\ivr}}
\mequiv \pevolvein{\D{x}=\tilde{\theta}}{\tilde{\ivr}}\).
Admissibility follows from \(\freevars{\usubstjoin{\sigma}{\tau}}=\emptyset\).

\item Consider the case \(\prepeat{\alpha}\) with desired instance \(\prepeat{(\tilde{\alpha})}\), which has to have this shape.
By induction hypothesis, there is a uniform substitution $\sigma$ such that
\(\applyusubst{\sigma}{\alpha}
\mequiv \tilde{\alpha}\)
and \(\freevars{\sigma}=\emptyset\).
Then
\(\applyusubst{\sigma}{\prepeat{\alpha}}
\mequiv \prepeat{(\applyusubst{\sigma}{\alpha})}
\mequiv \prepeat{(\tilde{\alpha})}\),
which is $\boundvars{\applyusubst{\sigma}{\alpha}}$-admissible since \(\freevars{\sigma}=\emptyset\).

\item The case \(\alpha; \beta\) is similar and case \(\pchoice{\alpha}{\beta}\) follows directly from the uniform replacement argument.
\end{enumerate}

\noindent
The corresponding result for axiomatic rules built from surjective \dL formulas follows since surjective \dL formulas can be instantiated by rule \irref{US} to any instance, which, thus, continues to hold for the premises and conclusions in rule \irref{USR}.
\qedhere
\end{proofatend}
\noindent
\rref{lem:surjectiveaxiom} generalizes to quantifier symbols with arguments that have no function or predicate symbols, since those are always $\allvars$-admissible.
Generalizations to function and predicate symbol instances are possible with adequate care.
The axiom \irref{testb} is surjective, because it does not have any bound variables, so admissibility of its instances is obvious.
Similarly rules \irref{MP} and, with the twist from the proof of \rref{thm:dL-sound}, rule \irref{gena} become surjective.
Axioms \irref{allinst+alldist+vacuousall} can be augmented for surjectivity in similar ways,
where \irref{vacuousall} is surjective when $p$ is instantiated such that $x$ does not occur free, which is a soundness-critical restriction,
and \irref{allinst} is instantiated respecting its shape.

A previous schematic completeness result \cite{DBLP:journals/tocl/Platzer15} shows completeness relative to any differentially expressive\footnote{%
A logic \LBase closed under first-order connectives
is \emph{differentially expressive} (for \dL) if every \dL formula $\phi$ has an equivalent $\reduct{\phi}$ in \LBase and all equivalences of the form \(\ddiamond{\pevolve{\D{x}=\genDE{x}}}{G} \lbisubjunct \reduct{(\ddiamond{\pevolve{\D{x}=\genDE{x}}}{G})}\) for $G$ in \LBase are provable in its calculus.
} logic.
\rref{lem:surjectiveaxiom} makes it easy to augment this proof to show that the schema instantiations required for completeness are provable by \irref{US+USR} from axioms or axiomatic rules.
Both the first-order logic of differential equations \cite{DBLP:journals/jar/Platzer08} and discrete dynamic logic \cite{DBLP:conf/lics/Platzer12b} are differentially expressive for \dL.

\begin{theorem}[Relative completeness] \label{thm:dL-complete}%
  The \dL calculus is a \emph{sound and complete axiomatization} of hybrid systems relative to \emph{any} differentially expressive logic \LBase, i.e.\
  every valid \dL formula is provable in the \dL calculus from \LBase tautologies.
\end{theorem}
\begin{proofatend}
\let\Oracle\LBase%
This proof refines the completeness proof for the axiom schemata of differential game logic \cite{DBLP:journals/tocl/Platzer15} with explicit proofs of instantiability by \irref{US} and \irref{USR}.
Write \m{\infers[\Oracle] \phi} to indicate that \dL formula $\phi$ can be derived in the \dL proof calculus from valid \LBase formulas.
Soundness follows from \rref{thm:dL-sound}, so it remains to prove completeness.
  For every valid \dL formula $\phi$ it has to be proved that $\phi$
  can be derived from valid \LBase tautologies within the \dL calculus:
  from \m{\entails\phi} prove \m{\infers[\Oracle] \phi}.
  The proof proceeds as follows: By propositional
  recombination, inductively identify fragments of $\phi$ that correspond to
  \m{\phi_1 \limply \ddiamond{\alpha}{\phi_2}}
  or
  \m{\phi_1 \limply \dbox{\alpha}{\phi_2}} logically.
  Find structurally simpler formulas from which these properties can be derived in the \dL calculus by uniform substitution instantiations, taking care that the resulting formulas are simpler than the original one in a well-founded order.
  Finally, prove that the original \dL formula can be re-derived from the subproofs in the \dL calculus by uniform substitution instantiations.

  The first insight is that, with the rules \irref{MP} and \irref{gena} and (by \rref{lem:surjectiveaxiom}, all) relevant instances of \irref{allinst+alldist+vacuousall} and real arithmetic, the \dL calculus contains a complete axiomatization of first-order logic.
  Thus, all first-order logic tautologies can be used without further notice in the remainder of the proof.
  Furthermore, by \rref{lem:surjectiveaxiom}, all instances of \irref{diamond+choiceb+composeb+iterateb+K+I} can be proved by rule \irref{US} in the \dL calculus.

  By appropriate propositional derivations, assume $\phi$ to be given in conjunctive normal form.
  Assume that negations are pushed inside over modalities using the dualities
  \m{\lnot\dbox{\alpha}{\phi} \mequiv \ddiamond{\alpha}{\lnot\phi}}
  and
  \m{\lnot\ddiamond{\alpha}{\phi} \mequiv \dbox{\alpha}{\lnot\phi}} that are provable by axiom \irref{diamond},
  and that negations are pushed inside over quantifiers using definitorial first-order equivalences
  \m{\lnot\lforall{x}{\phi} \mequiv \lexists{x}{\lnot\phi}}
  and
  \m{\lnot\lexists{x}{\phi} \mequiv \lforall{x}{\lnot\phi}}.
  The remainder of the proof follows an induction on a well-founded partial order~$\prec$ from previous work \cite{DBLP:journals/tocl/Platzer15} induced on \dL formulas by the lexicographic ordering of the overall structural complexity of the hybrid programs in the formula and the structural complexity of the formula itself, with the logic \LBase placed at the bottom of the partial order~$\prec$.
  The base logic \LBase is considered of lowest complexity by relativity, because \m{\entails F} immediately implies \m{\infers[\Oracle] F} for all formulas $F$ of \LBase.
  The monotonicity rules derive from \irref{G+K+diamond} by \rref{lem:surjectiveaxiom} with a classical argument:
  \[
      \cinferenceRule[M|M]{$\ddiamond{}{}$ monotone / $\ddiamond{}{}$-generalization} %
      {\linferenceRule[formula]
        {p(\usall)\limply q(\usall)}
        {\ddiamond{a}{p(\usall)}\limply\ddiamond{a}{q(\usall)}}
      }{}
      \qquad
      \cinferenceRule[Mbox|M$_{\dibox{\cdot}}$]{$\ddiamond{}{}$ monotone / $\ddiamond{}{}$-generalization} %
      {\linferenceRule[formula]
        {p(\usall)\limply q(\usall)}
        {\dbox{a}{p(\usall)}\limply\dbox{a}{q(\usall)}}
      }{}
  \]
  The proof follows the syntactic structure of \dL formulas.
  \begin{enumerate}
  \addtocounter{enumi}{-1}%
  \item \label{case:dGL-complete-0}
    If $\phi$ has no hybrid programs, then~$\phi$ is a first-order formula; hence provable by assumption (even decidable if in first-order real arithmetic \cite{tarski_decisionalgebra51}, i.e.\ no uninterpreted symbols occur).
  
  \item $\phi$ is of the form \m{\lnot\phi_1}; then~$\phi_1$ is first-order and quantifier-free, as negations are assumed to be pushed inside, so \rref{case:dGL-complete-0} applies.
  
  \item $\phi$ is of the form \m{\phi_1 \land \phi_2}, then \m{\entails\phi_1} and \m{\entails\phi_2}, so individually deduce simpler proofs for
   \m{\infers[\Oracle] \phi_1} and \m{\infers[\Oracle] \phi_2}
    by IH, which combine propositionally
    to a proof for \m{\infers[\Oracle] \phi_1\land\phi_2} using \irref{MP} twice with the propositional tautology \m{\phi_1 \limply (\phi_2 \limply \phi_1\land\phi_2)}.
  
  \item The case where $\phi$ is of the form \m{\lexists{x}{\phi_2}}, \m{\lforall{x}{\phi_2}}, \m{\ddiamond{\alpha}{\phi_2}} or \m{\dbox{\alpha}{\phi_2}} is included in \rref{case:dGL-complete-or} with \(\phi_1\mequiv\lfalse\).

  \item \label{case:dGL-complete-or}
    \newcommand{\precondf}{F}%
    \newcommand{\postcondf}{G}%
    $\phi$ is a disjunction and---without loss of generality---has one of the following forms
    (otherwise use provable associativity and commutativity to reorder):
    \[
    \begin{array}{r@{~}c@{~}l}
      \phi_1 &\lor& \ddiamond{\alpha}{\phi_2}\\
      \phi_1 &\lor& \dbox{\alpha}{\phi_2}\\
      \phi_1 &\lor& \lexists{x}{\phi_2}\\
      \phi_1 &\lor& \lforall{x}{\phi_2}
      .
    \end{array}
    \]
    Let \m{\phi_1 \lor \dmodality{\alpha}{\phi_2}} be a unified notation for those cases.
    Then, \m{\phi_2\prec\phi},
    since~\(\phi_2\) has less modalities or quantifiers.
    Likewise, \m{\phi_1\prec\phi} because
    \m{\dmodality{\alpha}{\phi_2}} contributes one modality or quantifier to
    $\phi$ that is not part of $\phi_1$.
    When abbreviating the simpler formulas $\lnot\phi_1$ by $\precondf$ and $\phi_2$ by $\postcondf$, the validity \m{\entails \phi} yields \m{\entails \lnot \precondf \lor \dmodality{\alpha}{\postcondf}}, so \m{\entails \precondf \limply \dmodality{\alpha}{\postcondf}}, from which the remainder of the proof inductively derives
    \begin{equation} \label{eq:derivable-triple}
      \infers[\Oracle]{\precondf \limply \dmodality{\alpha}{\postcondf}}
      .
    \end{equation}
    
    The proof of \rref{eq:derivable-triple} is by structural induction on $\dmodality{\alpha}{}$.
    This proof focuses on the quantifier and $\dbox{}{}$ cases, because most $\ddiamond{}{}$ cases derive by axiom \irref{diamond} with \rref{lem:surjectiveaxiom} from the $\dbox{}{}$ equivalences.

  \begin{enumerate}
  \item \label{case:complete-forall}
  If $\dmodality{\alpha}{}$ is the operator $\lforall{x}{}$ then \m{\entails \precondf \limply \lforall{x}{\postcondf}}, where $x$ can be assumed not to occur in $\precondf$ by a bound variable renaming.
  Hence, \m{\entails \precondf \limply \postcondf}.
  Since \(\postcondf \prec \lforall{x}{\postcondf}\), 
  because it has less quantifiers, also \((\precondf\limply \postcondf) \prec (\precondf\limply\lforall{x}{\postcondf})\),
  hence \(\infers[\Oracle] \precondf \limply \postcondf\) is derivable by IH.
  Then, \(\infers[\Oracle] \precondf \limply \lforall{x}{\postcondf}\) derives with \rref{lem:surjectiveaxiom} by generalization rule \irref{gena}, since $x$ does not occur in $\precondf$:
  \begin{sequentdeduction}[Hilbert+array]
    \linfer[vacuousall]
      {\linfer[alldist]
        {\linfer[gena]
          {\lsequent{}{\precondf \limply \postcondf}}
          {\lsequent{}{\lforall{x}{(\precondf \limply \postcondf)}}}
        }%
        {\lsequent{}{\lforall{x}{\precondf} \limply \lforall{x}{\postcondf}}}
      }%
      {\lsequent{}{\precondf \limply \lforall{x}{\postcondf}}}
  \end{sequentdeduction}
  The instantiations succeed by the remark after \rref{lem:surjectiveaxiom} using for \irref{vacuousall} that $x\not\in\vars{F}$.
  The formula \(\precondf \limply \lforall{x}{\postcondf}\) is even decidable if in first-order real arithmetic \cite{tarski_decisionalgebra51}.
  The remainder of the proof concludes
  \((\precondf\limply \psi) \prec (\precondf\limply \phi)\) from \(\psi \prec \phi\) without further notice.
  The operator $\lforall{y}{}$ can be obtained correspondingly by uniform renaming.

  \item If $\dmodality{\alpha}{}$ is the operator $\lexists{x}{}$ then \m{\entails \precondf \limply \lexists{x}{\postcondf}}.
  If $\precondf$ and $\postcondf$ are \LBase formulas, then, since \LBase is closed under first-order connectives, so is the valid formula \m{\precondf \limply \lexists{x}{\postcondf}}, which is, then, provable by IH and even decidable if in first-order real arithmetic \cite{tarski_decisionalgebra51}.
  
  Otherwise, $\precondf,\postcondf$ correspond to \LBase formulas by expressiveness of \LBase, which implies the existence of an \LBase formula $\reduct{\postcondf}$ such that
  \(\entails \reduct{\postcondf} \lbisubjunct \postcondf\).
  Since \LBase is closed under first-order connectives \cite{DBLP:journals/tocl/Platzer15}, the valid formula \m{\precondf \limply \lexists{x}{(\reduct{\postcondf})}} is provable by IH, because
  \((\precondf \limply \lexists{x}{(\reduct{\postcondf})}) \prec (\precondf \limply \lexists{x}{\postcondf})\) since $\reduct{\postcondf}\in\LBase$ while $\postcondf\not\in\LBase$.
  Now, \(\entails \reduct{\postcondf} \lbisubjunct \postcondf\)
  implies \(\entails \reduct{\postcondf} \limply \postcondf\), which is derivable by IH, because
  \((\reduct{\postcondf} \limply \postcondf) \prec \phi\) since $\reduct{\postcondf}$ is in \LBase.
  From \(\infers[\Oracle] \reduct{\postcondf} \limply \postcondf\), the derivable dual of axiom \irref{alldist}, 
  (\(\lforall{x}{(p(x)\limply q(x))} \limply (\lexists{x}{p(x)} \limply \lexists{x}{q(x)})\)), derives
  \(\infers[\Oracle] \lexists{x}{(\reduct{\postcondf})} \limply \lexists{x}{\postcondf}\),
  which combines with
  \(\infers[\Oracle] \precondf \limply \lexists{x}{(\reduct{\postcondf})}\)
  essentially by rule \irref{MP} to
  \(\infers[\Oracle] \precondf \limply \lexists{x}{\postcondf}\).
  \begin{sequentdeduction}[Hilbert+array]
  \linfer[MP]%
    {\lsequent{\precondf} {\lexists{x}{(\reduct{\postcondf})}}
    !\linfer[alldist]%
      {\linfer[gena]
        {\lsequent{} {\reduct{\postcondf} \limply \postcondf}}
        {\lsequent{} {\lforall{x}{(\reduct{\postcondf} \limply \postcondf)}}}
      }%
      {\lsequent{} {\lexists{x}{(\reduct{\postcondf})} \limply \lexists{x}{\postcondf}}}
    }
    {\lsequent{\precondf} {\lexists{x}{\postcondf}}}
  \end{sequentdeduction}
  The instantiations succeed by \rref{lem:surjectiveaxiom} and its subsequent remark.

  \item \m{\entails \precondf \limply \ddiamond{\pevolve{\D{x}=\genDE{x}}}{\postcondf}}
  implies
  \m{\entails \precondf \limply \reduct{(\ddiamond{\pevolve{\D{x}=\genDE{x}}}{\reduct{\postcondf}})}},
  which is derivable by IH,
  as \((\precondf \limply\reduct{(\ddiamond{\pevolve{\D{x}=\genDE{x}}}{\reduct{\postcondf}})})\) \(\prec \phi\)
  since \(\reduct{(\ddiamond{\pevolve{\D{x}=\genDE{x}}}{\reduct{\postcondf}})}\) is in \LBase.
  \LBase is differentially expressive, so
  \(\infers[\Oracle] \ddiamond{\pevolve{\D{x}=\genDE{x}}}{\reduct{\postcondf}} \lbisubjunct \reduct{(\ddiamond{\pevolve{\D{x}=\genDE{x}}}{\reduct{\postcondf}})}\)
  is provable.
  Hence
  \m{\infers[\Oracle] \precondf \limply \ddiamond{\pevolve{\D{x}=\genDE{x}}}{\reduct{\postcondf}}} derives by propositional congruence.
  Now \(\reduct{\postcondf}\limply\postcondf\) is simpler (since $\reduct{\postcondf}$ is in $\LBase$) so derivable by IH,
  so \(\ddiamond{\pevolve{\D{x}=\genDE{x}}}{\reduct{\postcondf}}\limply\ddiamond{\pevolve{\D{x}=\genDE{x}}}{\postcondf}\)
  derives by \irref{M}.
  Together, both derive
  \m{\infers[\Oracle] \precondf \limply \ddiamond{\pevolve{\D{x}=\genDE{x}}}{\postcondf}} propositionally.
  
  \item \m{\entails \precondf \limply \dbox{\pevolve{\D{x}=\genDE{x}}}{\postcondf}}
  implies
  \m{\entails \precondf \limply \lnot\ddiamond{\pevolve{\D{x}=\genDE{x}}}{\lnot \postcondf}}.
  Thus,
  \m{\entails \precondf \limply \lnot\reduct{(\ddiamond{\pevolve{\D{x}=\genDE{x}}}{\lnot \reduct{\postcondf}})}},
  which is derivable by IH,
  because \((\precondf \limply\lnot\reduct{(\ddiamond{\pevolve{\D{x}=\genDE{x}}}{\reduct{\lnot \postcondf}})}) \prec \phi\)
  as \(\reduct{(\ddiamond{\pevolve{\D{x}=\genDE{x}}}{\reduct{\lnot \postcondf}})}\) is in \LBase.
  Since \LBase is differentially expressive,
  \(\infers[\Oracle] \ddiamond{\pevolve{\D{x}=\genDE{x}}}{\reduct{\lnot \postcondf}} \lbisubjunct \reduct{(\ddiamond{\pevolve{\D{x}=\genDE{x}}}{\lnot \reduct{\postcondf}})}\)
  is provable, so
  \(\infers[\Oracle] \precondf \limply \lnot\ddiamond{\pevolve{\D{x}=\genDE{x}}}{\lnot \reduct{\postcondf}}\)
  derives from \(\infers[\Oracle] \precondf \limply \lnot\reduct{(\ddiamond{\pevolve{\D{x}=\genDE{x}}}{\reduct{\lnot \postcondf}})}\) by propositional congruence.
  Axiom \irref{diamond}, thus, derives
  \(\infers[\Oracle] \precondf \limply \dbox{\pevolve{\D{x}=\genDE{x}}}{\reduct{\postcondf}}\)
  with \rref{lem:surjectiveaxiom}.
  Now \(\reduct{\postcondf}\limply\postcondf\) is simpler (as $\reduct{\postcondf}$ is in $\LBase$) so derivable by IH,
  so \(\dbox{\pevolve{\D{x}=\genDE{x}}}{\reduct{\postcondf}}\limply\dbox{\pevolve{\D{x}=\genDE{x}}}{\postcondf}\)
  derives by \irref{M}.
  Together, both derive
  \m{\infers[\Oracle] \precondf \limply \dbox{\pevolve{\D{x}=\genDE{x}}}{\postcondf}} propositionally.

\item \m{\entails \precondf \limply \dbox{\pevolvein{\D{x}=\genDE{x}}{\ivr}}{\postcondf}},
then this formula has an equivalent \cite[Lemma\,3.4]{DBLP:journals/tocl/Platzer15} without evolution domains which can be used as a definitorial abbreviation to conclude this case.
  Similarly for \m{\entails \precondf \limply \ddiamond{\pevolvein{\D{x}=\genDE{x}}{\ivr}}{\postcondf}}.

  \item \m{\entails \precondf \limply \dbox{\pupdate{\pumod{y}{\theta}}}{\postcondf}} then this formula can be proved, using a fresh variable $x\not\in\vars{\theta}\cup\vars{\postcondf}$, with the following derivation by bound variable renaming (rule \irref{Brename}), in which $\urename[\postcondf]{y}{x}$ is the result of uniformly renaming $y$ to $x$ in $\postcondf$

\irlabel{Brename|BR}%
  \begin{sequentdeduction}[array]
    \linfer[Brename] %
    {\linfer[assignbeq] %
      {\lsequent{\precondf} {\lforall{x}{(x=\theta\limply\urename[\postcondf]{y}{x})}}}
    {\lsequent{\precondf} {\dbox{\pupdate{\pumod{x}{\theta}}}{\urename[\postcondf]{y}{x}}}}
    }%
    {\lsequent{\precondf} {\dbox{\pupdate{\pumod{y}{\theta}}}{\postcondf}}}
  \end{sequentdeduction}
  using the following equational form of the assignment axiom \irref{assignb}
\[
      \dinferenceRule[assignbeq|$\dibox{:=}_{=}$]{assignment / equational axiom}
      {\linferenceRule[equiv]
        {\lforall{x}{(x=f \limply p(\usall))}}
        {\dbox{\pupdate{\umod{x}{f}}}{p(\usall)}}
      }
      {}%
\]
  The above proof only used equivalence transformations, so its premise is valid iff its conclusion is, which it is by assumption.
  The assumption, thus, implies
  \(\entails \precondf \limply \lforall{x}{(x=\theta\limply\urename[\postcondf]{y}{x})}\).
  Since \((\precondf \limply \lforall{x}{(x=\theta\limply\urename[\postcondf]{y}{x})}) \prec (\precondf \limply \dbox{\pupdate{\pumod{y}{\theta}}}{\postcondf})\),
  because there are less hybrid programs,
  \(\infers[\Oracle] \precondf \limply \lforall{x}{(x=\theta\limply\urename[\postcondf]{y}{x})}\)
  by IH. %
  The above proof, thus, derives \(\infers[\Oracle] \precondf \limply \dbox{\pupdate{\pumod{y}{\theta}}}{\postcondf}\).
  
  The equational assignment axiom \irref{assignbeq} above can either be adopted as an axiom in place of \irref{assignb}.
  Or it can be derived from axiom \irref{assignb} with the uniform substitution 
  \(\sigma = \usubstlist{\usubstmod{q(\usarg)}{p(\usarg,X)}}\)
  when splitting the variables $\usall$ into the variable $x$ and the other variables $X$ such that $x\not\in X$:
  \newcommand*{\ucons}[2]{#1,#2}%
\irlabel{FOL|FOL}%
  \begin{sequentdeduction}[array]
  \linfer%
  {\linfer[US]
    {\linfer[assignb]
      {\linfer[FOL] %
        {\lclose}
        {\lseqalign{q(f)} {\lbisubjunct \lforall{x}{(x=f \limply q(x))}}}
      }%
      {\lseqalign{\dbox{\pupdate{\umod{x}{f}}}{q(x)}} {\lbisubjunct \lforall{x}{(x=f \limply q(x))}}}
    }%
    {\lseqalign{\dbox{\pupdate{\umod{x}{f}}}{p(\ucons{x}{X})}} {\lbisubjunct \lforall{x}{(x=f \limply p(\ucons{x}{X}))}}}
  }%
  {\lseqalign{\dbox{\pupdate{\umod{x}{f}}}{p(\usall)}} {\lbisubjunct \lforall{x}{(x=f \limply p(\usall))}}}
  \end{sequentdeduction}

  It only remains to be shown that \irref{assignbeq} can be instantiated as indicated in the above proof.
  This follows from \rref{lem:surjectiveaxiom} with the additional observation that the required uniform substitution \(\usubstlist{\usubstmod{f}{\theta}}\) of function symbol $f$ of arity 0 without argument $\usall$ will not cause a clash during \irref{US}, because the only bound variable $x$ in \irref{assignbeq} is not free in the substitution since $x\not\in\vars{\theta}$.

  Other proofs involving stuttering and renaming are possible.
  Direct proofs of \m{\precondf \limply \dbox{\pupdate{\pumod{y}{\theta}}}{\postcondf}} by axiom \irref{assignb} are possible if the substitution is admissible.

  \item \m{\entails \precondf \limply \dbox{\ptest{\ivr}}{\postcondf}} implies \m{\entails \precondf \limply (\ivr \limply \postcondf)}.
  Since \((\ivr\limply \postcondf) \prec \dbox{\ptest{\ivr}}{\postcondf}\), because it has less modalities,
  \(\infers[\Oracle] \precondf \limply (\ivr \limply \postcondf)\) is derivable by IH.
  Hence, with the remark after \rref{lem:surjectiveaxiom}, axiom \irref{testb} instantiates to \(\dbox{\ptest{\ivr}}{\postcondf} \lbisubjunct (\ivr \limply \postcondf)\),
  so derives
  \(\infers[\Oracle] \precondf \limply \dbox{\ptest{\ivr}}{\postcondf}\)
  by propositional congruence, which is used without further notice subsequently.
  
  \item \m{\entails \precondf \limply \dbox{\pchoice{\beta}{\gamma}}{\postcondf}} implies \m{\entails \precondf \limply \dbox{\beta}{\postcondf} \land \dbox{\gamma}{\postcondf}}.
  Since \(\dbox{\beta}{\postcondf}\land\dbox{\gamma}{\postcondf} \prec \dbox{\pchoice{\beta}{\gamma}}{\postcondf}\), because, even if the propositional and modal structure increased, the structural complexity of both hybrid programs $\beta$ and $\gamma$ is smaller than that of $\pchoice{\beta}{\gamma}$ (formula $\postcondf$ did not change),
  \(\infers[\Oracle] \precondf \limply \dbox{\beta}{\postcondf} \land \dbox{\gamma}{\postcondf}\) is derivable by IH.
  Hence, with \rref{lem:surjectiveaxiom}, axiom \irref{choiceb} instantiates to \(\dbox{\pchoice{\beta}{\gamma}}{\postcondf} \lbisubjunct \dbox{\beta}{\postcondf} \lor \dbox{\gamma}{\postcondf}\),
  so derives \(\infers[\Oracle] \precondf \limply \dbox{\pchoice{\beta}{\gamma}}{\postcondf}\) by propositional congruence.
  
  \item \m{\entails \precondf \limply \dbox{\beta;\gamma}{\postcondf}}, which implies \m{\entails \precondf \limply \dbox{\beta}{\dbox{\gamma}{\postcondf}}}.
  Since \(\dbox{\beta}{\dbox{\gamma}{\postcondf}} \prec \dbox{\beta;\gamma}{\postcondf}\), because, even if the number of modalities increased, the overall structural complexity of the hybrid programs decreased because there are less sequential compositions, 
  \(\infers[\Oracle] \precondf \limply \dbox{\beta}{\dbox{\gamma}{\postcondf}}\) is derivable by IH.
  Hence, with \rref{lem:surjectiveaxiom}, \(\infers[\Oracle] \precondf \limply \dbox{\beta;\gamma}{\postcondf}\) derives by axiom \irref{composeb} by propositional congruence.

\item \m{\entails \precondf \limply \dbox{\prepeat{\beta}}{\postcondf}}
    can be derived by induction as follows.
    Formula \m{\dbox{\prepeat{\beta}}{\postcondf}}, which expresses that all numbers of repetitions of~$\prepeat{\beta}$ satisfy~$\postcondf$,
    is an inductive invariant of $\prepeat{\beta}$, because \m{\dbox{\prepeat{\beta}}{\postcondf}\limply\dbox{\beta}{\dbox{\prepeat{\beta}}{\postcondf}}} is valid, even provable by \irref{iterateb}.
    Thus, its equivalent \LBase encoding is also an inductive invariant:
    \[
    \inv \mequiv
    \reduct{(\dbox{\prepeat{\beta}}{\postcondf})}
    .
    \]
    Then \m{\precondf \limply \inv} and \m{\inv \limply \postcondf} are
    valid (zero repetitions are possible), so derivable by IH, since \((\precondf\limply \inv) \prec \phi\) and \((\inv\limply \postcondf) \prec \phi\) hold, because $\inv$ is in \LBase.
    As above, \m{\inv \limply \dbox{\beta}{\inv}} is valid, and thus derivable by IH, since~$\beta$ has less loops than $\prepeat{\beta}$.
    By \irref{Mbox} as well as
    rule \irref{invind} (from \(p(\usall)\limply\dbox{a}{p(\usall)}\) conclude \(p(\usall)\limply\dbox{\prepeat{a}}{p(\usall)}\)), which derives from \irref{I+G} by \rref{lem:surjectiveaxiom}, the respective rules can be instantiated by \rref{lem:surjectiveaxiom} and the resulting derivations combine by \irref{MP}:
    \begin{sequentdeduction}[array]
    \linfer[MP]
    {\lsequent{\precondf} {\inv}
    !
    \linfer[MP]
    {%
    \linfer[invind]
     {\lsequent{\inv} {\dbox{\beta}{\inv}}}
     {\lsequent{\inv} {\dbox{\prepeat{\beta}}{\inv}}}
    !
      \linfer[Mbox]
      {\lsequent{\inv} {\postcondf}}
      {\lsequent{\dbox{\prepeat{\beta}}{\inv}} {\dbox{\prepeat{\beta}}{\postcondf}}}
    }
    {\lsequent{\inv} {\dbox{\prepeat{\beta}}{\postcondf}}}
    }
  {\lsequent{\precondf} {\dbox{\prepeat{\beta}}{\postcondf}}}
    \end{sequentdeduction}

\item \m{\entails \precondf \limply \ddiamond{\prepeat{\beta}}{\postcondf}}.
   \def\vec#1{#1}%
    Let $\vec{x}$ be the vector of free variables \m{\freevars{\ddiamond{\prepeat{\beta}}{\postcondf}}}.
    Since $\ddiamond{\prepeat{\beta}}{\postcondf}$ is a least pre-fixpoint \cite{DBLP:journals/tocl/Platzer15}, for all \dL formulas $\psi$ with \(\freevars{\psi}\subseteq\freevars{\ddiamond{\prepeat{\beta}}{\postcondf}}\):
    \[
    \entails \lforall{\vec{x}}{(\postcondf\lor\ddiamond{\beta}{\psi}\limply\psi)} \limply (\ddiamond{\prepeat{\beta}}{\postcondf} \limply\psi)
    \]
    In particular, this holds for a fresh predicate symbol $p$ with arguments $\vec{x}$:
    \[
    \entails \lforall{\vec{x}}{(\postcondf\lor\ddiamond{\beta}{p(\vec{x})}\limply p(\vec{x}))} \limply (\ddiamond{\prepeat{\beta}}{\postcondf} \limply p(\vec{x}))
    \]
    Using \m{\entails \precondf \limply \ddiamond{\prepeat{\beta}}{\postcondf}}, this implies
    \[
    \entails \lforall{\vec{x}}{(\postcondf\lor\ddiamond{\beta}{p(\vec{x})}\limply p(\vec{x}))} \limply (\precondf \limply p(\vec{x}))
    \]
    As \((\lforall{\vec{x}}{(\postcondf\lor\ddiamond{\beta}{p(\vec{x})}\limply p(\vec{x}))} \limply (\precondf \limply p(\vec{x}))) \prec \phi\), because, even if the formula complexity increased, the structural complexity of the hybrid programs decreased, since $\phi$ has one more loop, this fact is derivable by IH:
    \[
    \infers[\Oracle] \lforall{\vec{x}}{(\postcondf\lor\ddiamond{\beta}{p(\vec{x})}\limply p(\vec{x}))} \limply (\precondf \limply p(\vec{x}))
    \]
    The uniform substitution \(\sigma=\usubstlist{\usubstmod{p(\vec{x})}{\ddiamond{\prepeat{\beta}}{\postcondf}}}\) is admissible since \(\freevars{\sigma}=\emptyset\) as $\ddiamond{\prepeat{\beta}}{\postcondf}$ has free variables $\vec{x}$.
    Since, furthermore, $p\not\in \intsigns{\precondf}\cup\intsigns{\postcondf}\cup\intsigns{\beta}$, \irref{US} derives:
    \begin{sequentdeduction}[array]
    \linfer[US]
      {\lsequent{} {\lforall{\vec{x}}{(\postcondf\lor\ddiamond{\beta}{p(\vec{x})}\limply p(\vec{x}))} \limply (\precondf \limply p(\vec{x}))}}
      {\lsequent{} {\lforall{\vec{x}}{(\postcondf\lor\ddiamond{\beta}{\ddiamond{\prepeat{\beta}}{\postcondf}}\limply \ddiamond{\prepeat{\beta}}{\postcondf})} \limply (\precondf \limply \ddiamond{\prepeat{\beta}}{\postcondf})}}
    \end{sequentdeduction}
    The dual \(\ddiamond{\prepeat{a}}{p(\usall)} \lbisubjunct p(\usall) \lor \ddiamond{a}{\ddiamond{\prepeat{a}}{p(\usall)}}\) resulting from axiom \irref{iterateb} with axiom \irref{diamond} by \rref{lem:surjectiveaxiom} continues this derivation by \rref{lem:surjectiveaxiom}:
\renewcommand{\linferPremissSeparation}{~}%

\begin{minipage}{\textwidth}
\advance\leftskip-1.1cm
\begin{minipage}{\textwidth}
    \begin{sequentdeduction}[Hilbert+array]
    \linfer[MP]
    {
    \lsequent{} {\lforall{\vec{x}}{(\postcondf\lor\ddiamond{\beta}{\ddiamond{\prepeat{\beta}}{\postcondf}}\limply \ddiamond{\prepeat{\beta}}{\postcondf})} \limply (\precondf \limply \ddiamond{\prepeat{\beta}}{\postcondf})}
    !
    \linfer[gena]
    {\linfer[iterateb+diamond]
      {\lclose}
      {\lsequent{} {\postcondf\lor\ddiamond{\beta}{\ddiamond{\prepeat{\beta}}{\postcondf}} \limply \ddiamond{\prepeat{\beta}}{\postcondf}}}
    }
    {\lsequent{} {\lforall{\vec{x}}{(\postcondf\lor\ddiamond{\beta}{\ddiamond{\prepeat{\beta}}{\postcondf}} \limply \ddiamond{\prepeat{\beta}}{\postcondf})}}}
    }
    {\lsequent{\precondf} {\ddiamond{\prepeat{\beta}}{\postcondf}}}
    \end{sequentdeduction}
    \end{minipage}
    \end{minipage}
    Observe that rule \irref{gena} (and \irref{MP}) instantiates as needed with \irref{USR} by \rref{lem:surjectiveaxiom}.
    \end{enumerate}
    \noindent
    This concludes the derivation of \rref{eq:derivable-triple}, because all operators \m{\dmodality{\alpha}{}} for the form \rref{eq:derivable-triple} have been considered.
    From \rref{eq:derivable-triple}, which is 
    \m{\infers[\Oracle] \lnot\phi_1 \limply \dmodality{\alpha}{\phi_2}}, hence,
    \m{\infers[\Oracle] \phi_1 \lor \dmodality{\alpha}{\phi_2}}
    derives propositionally.
    
  \end{enumerate}
This completes the proof of completeness (\rref{thm:dL-complete}), because all syntactical forms of \dL formulas have been covered.
\qedhere
\end{proofatend}
\noindent
With the expected exceptions of loops and differential equations, the proof of \rref{thm:dL-complete} confirms that successive unification with axiom keys gives a complete proof strategy.
The search for applicable positions is deterministic using recursive computations as in \rref{ex:diffind-forward-proof}.
Loops and differential equations need corresponding (differential) invariant search using parametric predicates $\mydiffcond$ as in \rref{ex:diffind-free-parametric-proof}.

This result proves that a very simple mechanism, essentially the single proof rule of uniform substitution, makes it possible to prove differential dynamic logic formulas from a parsimonious soundness-critical core with a few concrete formulas as axioms and without losing the completeness that axiom schema calculi enjoy.

\section{Conclusions}

Uniform substitutions lead to a simple and modular proof calculus that is entirely based on axioms and axiomatic rules, instead of soundness-critical schema variables with side-conditions in axiom schemata and proof rules.
The \irref{US} calculus is straightforward to implement, since axioms are just formulas and axiomatic rules are just pairs of formulas and since the uniform substitutions themselves have a straightforward recursive definition.
The key ingredient enabling such modularity for differential equations are differential forms that have a local semantics and make it possible to reduce reasoning about \emph{differential equations} to local reasoning about (inequalities or) \emph{equations of differentials}.
The increased modularity also enables flexible reasoning by fast contextual equivalence that uniform substitutions provide almost for free.

Overall, uniform substitutions lead to a simple and modular, sound and complete proof calculus for differential dynamic logic that is entirely based on axioms and axiomatic rules.
Prover implementations merely reduce to uniform substitutions using the static semantics, starting from one copy of each axiom and axiomatic rule.
This leads to significantly simpler and more parsimonious implementations.
The soundness-critical core of the uniform substitution prover \KeYmaeraX \cite{DBLP:conf/cade/FultonMQVP15}, for example, is 2.5\% of the size of the core of the sequent calculus prover \KeYmaera \cite{DBLP:conf/cade/PlatzerQ08}, which, even if implemented in a different programming language, has more complex implementations of proof rules and schema variable matching or built-in operators.

\appendix

\section{Proofs} \label{app:proofs}
The proofs use the following classical results, where $\gradient{g}$ denotes the gradient of the function $g$ so the vector of all partial derivatives (if it exists).
\begin{lemma}[{Mean-value theorem \cite[\S10.10]{Walter:Ana1}}] \label{lem:mean-value}
  If \(f:[a,b]\to\reals\) is continuous and differentiable in $(a,b)$,
  then there is a $\xi\in(a,b)$ such that
  \[
  f(b)-f(a) = f'(\xi) (b-a)
  \]
\end{lemma}

\begin{lemma}[{Chain rule \cite[\S3.10]{Walter:Ana2}}] \label{lem:chain}
  If \(f:U\to\reals^m\) is differentiable at $t\in U\subseteq\reals$
  and \(g:V\to\reals\), with \(f(U)\subseteq V\subseteq\reals^m\), is differentiable at $f(t)\in V$,
  then \(g\compose f:U\to\reals\) is differentiable at $t$ with derivative
  \[
  \D{(g\compose f)}(t) = (\gradient{g})\big(f(t)\big) \stimes \D{f}(t)
  = \sum_{j=1}^m \Dp[y_j]{g}\big(f(t)\big)\D{f_j}(t)
  \]
\end{lemma}

\begin{theorem}[{Global uniqueness theorem of Picard-Lindel\"of \cite[\S10.VII]{Walter:ODE}}] \label{thm:PicardLindelof-global}
  Let~$f:\interval{[0,a]}\times\reals^n\to\reals^n$ be a continuous function that is Lipschitz continuous with respect to~$y$ and let $y_0\in\reals^n$.
  Then, there is a unique solution of the following initial value problem on~$\interval{[0,a]}$
  \begin{equation*}
    \D{y}(t)=f(t,y)
    \qquad
    y(0)=y_0
  \end{equation*}
\end{theorem}

\printproofs

\section*{Acknowledgment}
This material is based upon work supported by the National Science Foundation under
NSF CAREER Award CNS-1054246.

The views and conclusions contained in this document are those of the author and should not be interpreted as representing the official policies, either expressed or implied, of any sponsoring institution or government.
Any opinions, findings, and conclusions or recommendations expressed in this publication are  those of the author(s) and do not necessarily reflect the views of any sponsoring institution or government.

\bibliographystyle{plainurl}\bibliography{platzer,bibliography}
\end{document}